\documentclass[11pt,letterpaper]{article}

\usepackage[utf8]{inputenc}
\usepackage[numbers,sort&compress]{natbib}
\usepackage[margin=1in]{geometry}
\usepackage{amsmath,amssymb,amsthm}
\usepackage{mathtools}
\usepackage{graphicx}
\usepackage{float}
\usepackage[linesnumbered,ruled,vlined]{algorithm2e}
\usepackage{booktabs}
\usepackage{xcolor}
\usepackage[normalem]{ulem}
\usepackage[shortlabels]{enumitem}
\usepackage{thm-restate}
\usepackage{hyperref}
\usepackage[nameinlink]{cleveref}

\hypersetup{
  colorlinks=true,
  linkcolor=blue,
  citecolor=blue,
  urlcolor=blue
}

\newtheorem{theorem}{Theorem}[section]
\newtheorem{lemma}[theorem]{Lemma}
\newtheorem{corollary}[theorem]{Corollary}

\newtheorem{claim}[theorem]{Claim}
\theoremstyle{definition}
\newtheorem{definition}[theorem]{Definition}

\theoremstyle{remark}
\newtheorem{remark}[theorem]{Remark}



\crefname{algocf}{Algorithm}{Algorithms}
\Crefname{algocf}{Algorithm}{Algorithms}

\newcommand{\todo}[2][]{}
\newcommand{\todoqi}[1]{}
\newcommand{\todosimon}[1]{}
\newcommand{\serge}[1]{}
\newcommand{\katie}[1]{}
\newcommand{\simon}[1]{}
\newcommand{\qi}[1]{}

\newcommand{\poly}{\operatorname{poly}}

\title{Faster Exponential-Time Approximate Counting via Bounded Self-Reductions}

\author{
Katie Clinch\\
University of Queensland, Australia\\
\texttt{k.clinch@uq.edu.au}
\and
Serge Gaspers\\
UNSW Sydney, Australia\\
\texttt{serge.gaspers@unsw.edu.au}
\and
Simon Mackenzie\\
UNSW Sydney, Australia\\
\texttt{simon.william.mackenzie@gmail.com}
\and
Qi Wang\\
UNSW Sydney, Australia\\
\texttt{wangqi118010296@outlook.com}
}

\date{}

\begin{document}

\maketitle
\begin{abstract}
We give faster exponential-time randomised approximation algorithms for counting problems where polynomial-time approximation is unavailable and exact exponential-time counting remains expensive.
For general \(n\)-vertex graphs, our independent-set counter runs in \(O^{\ast}(1.1869^{n})\) time, improving the previous \(O^{\ast}(1.2041^{n})\) general-graph bound.
For \(n\)-variable \#\textsc{2-SAT}, we obtain an \(O^{\ast}(1.2373^{n})\)-time approximation algorithm, narrowly below Wahlstr{\"o}m's currently cited \(O^{\ast}(1.2377^{n})\) variable-parameter exact bound.

The new algorithmic point is to take the square root after decomposition.
For a single bounded unweighted self-reduction with \(f(x)\) positive leaves and recursion-compatible upper bound \(b(x)\), an enumerate-or-sample estimator gives an \((\varepsilon,\delta)\)-approximation in
\[
  O^{\ast}\!\left(\sqrt{b(x)}\,\varepsilon^{-2}\log \tfrac1\delta\right)
\]
time.
After preprocessing decomposes an input into many bounded cores, the combined estimator pays
\[
  O^{\ast}\!\left(\sqrt{\sum_i b_i(x_i)}\,\varepsilon^{-2}\log \tfrac1\delta\right),
\]
rather than estimating the cores separately at cost \(\sum_i \sqrt{b_i(x_i)}\).

The same conversion improves the bases for counting maximal cliques, minimal separators, and perfect matchings in subcubic graphs.
Bounded unweighted self-reductions provide the formal language; at the level of counting classes, the resulting unweighted formulation has the same Karp closure as TotP.
With explicit recursion-tree access, the framework yields black-box quantum speed-ups.
\end{abstract}

\section{Introduction}

Counting the solutions to combinatorial problems is a fundamental challenge in algorithms and complexity theory.
The systematic study of counting problems from a computational complexity perspective may be said to have properly started in 1979 when Valiant introduced the complexity class \#P \cite{valiant1979complexity}.

Counting problems arise across numerous domains, from determining the number of satisfying assignments of a logical formula (\#\textsc{SAT} \cite{bulatov2013expressibility, dyer2004relative}, \#\textsc{2-SAT} \cite{dahllof2005counting, ge2018new, wahlstrom2008tighter}, \#\textsc{DNF} \cite{jerrum1986random, karp1989monte}) to counting structures in graphs (independent sets \cite{goldberg2021faster, GaspersLee23}, matchings \cite{dyer2004relative}, maximal cliques \cite{makino2004new}, cycles \cite{alon1997finding}, arbitrary subgraphs \cite{vassilevska2009finding}, colourings \cite{galanis2016approximately}), with applications such as evaluating partition functions in statistical physics \cite{goldberg2012approximating,jerrum1993polynomial}.
Because exact counting is intractable in general (\#P-complete even for 2-CNF formulas \cite{wahlstrom2008tighter}, or computing the volume of a convex body \cite{lovasz1993random}), research has also studied \textit{approximate} counting.
Fully Polynomial Randomised Approximation Schemes (FPRASs) output estimates within specified relative error with high probability in time polynomial in the input size and the inverse error parameters \cite{karp1983monte, jerrum1986random}; FPTASs are deterministic analogues \cite{ibarra1975fast, lawler1982fully}.
Early successes include FPRASs for DNF counting \cite{karp1983monte}, volume estimation for convex bodies \cite{dyer1991random}, and the permanent of a non-negative matrix \cite{jerrum2004polynomial}.

For self-reducible problems, a class that includes many natural counting problems, the existence of an approximation algorithm that achieves a polynomial factor error implies the existence of an FPRAS \cite{SinclairJ89}.
Furthermore, for such problems, almost uniform generation is inter-reducible with approximate counting, indicating that they are of similar complexity \cite{jerrum1986random}.
In spite of this, some counting problems remain intractable even in the approximation regime.
For instance, a multiplicative approximator for \#\textsc{SAT} can distinguish zero satisfying assignments from at least one satisfying assignment, and hence would decide \textsc{SAT}; therefore \#\textsc{SAT} has no FPRAS unless \(\mathrm{NP}=\mathrm{RP}\) \cite{dyer2004relative}.
This inapproximability extends under approximation-preserving reductions to problems such as \#\textsc{2-SAT}, \#\textsc{Independent-Set}, and \#\textsc{Maximal-Clique} \cite{dyer2004relative, goldberg2016approximately}.

For these problems, polynomial-time approximation is not expected, so the fine-grained question becomes: \emph{how fast can an exponential-time approximation algorithm be?}
Even a small reduction in the base of the exponent can substantially increase the accessible instance size.
This paper gives a reusable way to turn search-tree bounds into faster exponential-time approximate counters, and shows that the resulting algorithms improve several standard general-case exponential bases.

Our main applications are \#\textsc{Independent-Set} and \#\textsc{2-SAT}.
For fixed accuracy and confidence, we approximate the number of independent sets in an \(n\)-vertex graph in \(O^{\ast}(1.1869^{n})\) time, improving the previous general-graph \(O^{\ast}(1.2041^{n})\) algorithm of Goldberg, Lapinskas and Richerby \cite{goldberg2021faster}.
For \#\textsc{2-SAT}, we obtain an \(O^{\ast}(1.2373^{n})\)-time approximation algorithm on \(n\) variables, improving on the currently cited \(O^{\ast}(1.2377^{n})\) variable-parameter exact worst-case bound due to Wahlstr{\"o}m~\cite{wahlstrom2008tighter}.
The latter improvement is numerically modest, but it is conceptually useful: it gives a concrete setting where approximation is currently provably faster than the best known exact worst-case algorithm for the same problem.
We also obtain direct improvements for counting maximal cliques, minimal separators, and perfect matchings in subcubic graphs.

The central estimator is deliberately simple.
Suppose that a self-reduction represents \(f(x)\) as the number of positive leaves of a recursion tree, and suppose that we have a computable upper bound \(b(x)\) on that number which is compatible with the recursion.
We first enumerate up to \(\sqrt{b(x)}\) positive leaves.
If the enumeration finishes, the count is exact.
Otherwise \(f(x)\ge \sqrt{b(x)}\), so a sample drawn from the bound succeeds with probability at least \(1/\sqrt{b(x)}\), and standard Chernoff bounds give an \((\varepsilon,\delta)\)-approximation in
\[
O^{\ast}\!\left(\sqrt{b(x)}\,\varepsilon^{-2}\log\tfrac1\delta\right)
\]
time.

The strongest applications require one further composition principle.
A preprocessing step may decompose the original instance into many hard cores \(x_1,\ldots,x_\ell\), where core \(i\) has a bound \(b_i(x_i)\).
Running the estimator separately on each core would cost \(\sum_i \sqrt{b_i(x_i)}\).
Instead, we sample from the disjoint union of the bound masses and pay
\[
O^{\ast}\!\left(\sqrt{\sum_i b_i(x_i)}\,\varepsilon^{-2}\log\tfrac1\delta\right).
\]
This aggregate square-root effect is the mechanism behind the strongest applications.

This is not a formal corollary of the single-core estimator.
Applying \Cref{thm:fullalgo} independently to each core would still pay a sum of square roots, and would therefore miss the strongest applications.
The new point is that bounded-uSR cores provide, uniformly and black-box, the two objects needed to avoid that loss: a small-count enumerator and a bound-ticket sampler over the combined mass.

The framework language in the paper records exactly the structure needed for this estimator.
We use bounded unweighted self-reductions: polynomial-depth recursion trees in which internal nodes sum their children and each leaf contributes \(0\) or \(1\).
The feasibility predicate for subinstances lets us enumerate positive leaves with polynomial delay, while the recursion-compatible bound lets us sample without using feasibility.
From a complexity-class perspective, these forms are related to TotP, the class of functions that count all computation paths of nondeterministic polynomial-time Turing machines \cite{bakali2016self, bakali2017completeness, antonopoulos2022completeness, bakali2024power}.
We introduce the corresponding unweighted formulation uTotP and prove that its Karp closure is TotP.
This equivalence is a scope theorem: it justifies that the unweighted viewpoint does not exclude natural TotP problems, while the running-time gains come from the explicit bounds and decompositions supplied in the applications.

\subsection{The Engine: Extremal Bound + Feasibility Oracle $\Rightarrow$ Enumerate-or-Sample}
\label{sec:overview:engine}

\paragraph{Setup.}
Fix a counting function $f:\Sigma^{\ast}\to\mathbb N$ together with an unweighted self-reduction specified by a branching factor $r(\cdot)$, a child map $h(\cdot,\cdot)$, and a leaf predicate $t(\cdot)\in\{0,1\}$.
For an instance $y$, if $r(y)=0$ then $y$ is a leaf and contributes $t(y)$; otherwise its children are $h(y,1),\ldots,h(y,r(y))$ and the unweighted recurrence is
\[
  f(y)=\sum_{i=1}^{r(y)} f(h(y,i)).
\]
This defines a recursion tree rooted at $x$; a leaf $y$ is \emph{positive} if $t(y)=1$, and $f(x)$ equals the number of positive leaves.

The central idea of the paper is simple.
Assume we have an explicit, polynomial-time computable upper bound $B=b(x)$ on
the number of positive leaves of the recursion tree rooted at $x$, and that
this bound is \emph{recursion compatible} in the sense that for every sub-instance
$y$,
\[
\sum_{i=1}^{r(y)} b(h(y,i)) \le b(y).
\]
(This is the usual recurrence-style upper bound from branching analyses and
extremal counting.)

\medskip
\noindent\textbf{Feasibility for enumeration (implicit in TotP).}
Stage~1 uses a polynomial-time feasibility predicate
\[
\mathsf{Dec}_f(y) \;:=\; \bigl[f(y)>0\bigr],
\]
i.e., a decision procedure that tells whether the subtree rooted at $y$ contains
at least one positive leaf.
In the TotP setting this is \emph{not} an additional oracle assumption for the
bounded cores we use: the relevant self-reducible functions lie in
\(\#\mathrm{PE}\), so deciding whether \(f(y)>0\) is polynomial-time
computable by definition.
We keep $\mathsf{Dec}_f$ explicit because it cleanly exposes where feasibility is
used by our generic polynomial-delay enumerator (\Cref{sec:PolyDelay}); if an
application admits a polynomial-delay enumerator by other means, Stage~1 below
can use it directly.

\medskip
\noindent\textbf{Two-stage estimator.}
Let $k=\lceil\sqrt{B}\rceil$.

\begin{enumerate}[(i)]
\item \emph{Enumerate up to $k$ solutions.}
Run a polynomial-delay enumerator for positive leaves.
In our framework this enumerator is obtained from $\mathsf{Dec}_f$ and outputs leaves in lexicographic order.
If the enumerator halts before outputting $k$ leaves, then $f(x)<k$ and we return the exact count.

\item \emph{Otherwise sample.}
If the enumerator reaches $k$ leaves, then $f(x)\ge k\ge \sqrt{B}$, so $f(x)/B \ge 1/\sqrt{B}$.
We then draw $T=\Theta(\sqrt{B}\,\varepsilon^{-2}\log(1/\delta))$ samples from a sampler derived purely from the bound $b$ and the recursion, and estimate $f(x)$ from the empirical success probability.
Notably, this sampling stage does \emph{not} require $\mathsf{Dec}_f$.
\end{enumerate}

This yields the generic $O^*\!\bigl(\sqrt{b(x)}\,\varepsilon^{-2}\log(1/\delta)\bigr)$ algorithm (\Cref{thm:fullalgo}).
The engine is simple, but packaging it as a reusable interface lets us (i) plug in classical extremal
recurrence bounds directly, (ii) refactor computations to expose an unweighted leaf-counting
\emph{core}, and (iii) decompose an instance into many such cores and exploit a $\sqrt{\sum_i b_i}$
aggregation effect where $b_i$ is the bound on the $i^{th}$ core. This beats the naive $\sum_i \sqrt{b_i}$ and enables new algorithmic strategies for approximate counting.
The interface also exposes local access to the induced recursion tree, enabling black-box quantum speed-ups whenever the same interface is satisfied.

\subsection{Main Results}
\label{sec:main-results}

We now state the main results.  The classical application table below is the
main algorithmic take-away; the generic statements that power these bounds are
the single-core \(\sqrt{b(x)}\) estimator and, for the decomposition-based
applications, the \(\sqrt{\sum_i b_i}\) sum-of-cores theorem.
Throughout, a \emph{bounded uSR form} of a counting function $f$ (\Cref{def:uredp}) is a tuple
$\mathcal F^b=(f,h,t,r,q,b)$ that (i) specifies an \emph{unweighted} self-reduction (a recursion tree
in which internal nodes sum their children and leaves contribute $t(\cdot)\in\{0,1\}$, so that
$f(x)$ equals the number of $t=1$ leaves) and (ii) provides a polynomial-time computable,
\emph{recursion-compatible} upper bound $b(\cdot)$ on that leaf count (\Cref{def:boundedFunc}). Formal definitions appear in
\Cref{sec:uTotP}.

\paragraph{1. Generic $\sqrt{b(x)}$ approximate counting (classical).}
Given a bounded uSR form $\mathcal F^b$ for $f$ and input $x$, we obtain an
$(\varepsilon,\delta)$-approximation to $f(x)$ in time
\[
  O^{\ast}\!\Bigl(\sqrt{b(x)}\,\varepsilon^{-2}\log\tfrac1\delta\Bigr).
  \qquad\text{(\Cref{thm:fullalgo}).}
\]

\paragraph{2. Sums of bounded subinstances and the $\sqrt{\sum b_i}$ effect.}
More generally, suppose a preprocessing step expresses the target quantity as a sum of
$\ell$ bounded-uSR \emph{subinstances} (also called \emph{cores} in \Cref{sec:overview:cores,sec:Framework}), say
$x_1,\ldots,x_\ell$, where zero-valued cores have been discarded and
subinstance $i$ comes with its own bounded uSR presentation and bound
$b_i(x_i)$. (Intuitively, a ``core'' is just a residual subinstance on which the computation reduces to
unweighted leaf-counting in a recursion tree.)
Then we can approximate the sum of their contributions in time
\[
  O^{\ast}\!\Bigl(\sqrt{\sum_{i=1}^{\ell} b_i(x_i)}\,
     \varepsilon^{-2}\log\tfrac1\delta\Bigr).
  \qquad\text{(\Cref{thm:combineuTotP}).}
\]
This composition result is what drives our strongest algorithmic improvements. It enables the decomposition of problems into ``easy'' and ``hard'' counting instances, and allows us to aggregate the ``hard'' instances and get the square root acceleration on the sum rather than individually.

\paragraph{3. Black-box quantum accelerations.}
Our interface composes with standard quantum primitives to yield uniform speed-ups:
\[
  O^{\ast}\!\bigl(b(x)^{1/3}\,\varepsilon^{-1}\log\tfrac1\delta\bigr)
  \quad\text{and}\quad
  O^{\ast}\!\bigl(b(x)^{1/4}\,\varepsilon^{-3/2}\log^{2}\tfrac1\delta\bigr)
  \qquad\text{(\Cref{sec:quantum-aaronson,sec:Ambainis}).}
\]

\paragraph{4. Complexity-theoretic scope.}
We define uTotP as the Karp closure of uSR problems in $\#\mathrm{PE}$ and prove
$\mathrm{uTotP}=\mathrm{TotP}$ (\Cref{thm:uTotpeqTotp}).

\paragraph{5. Applications.}
We instantiate the framework on \#\textsc{maximal-clique}, \#\textsc{minimal-separator},
low-degree \#\textsc{perfect-matching}, \#\textsc{independent-set}, and
\#\textsc{2-SAT}; see
\Cref{tab:apps}.

\begin{table}[t]
\centering
\caption{Classical runtimes for our applications.}
\label{tab:apps}
\renewcommand{\arraystretch}{1.1}
\begin{tabular}{@{}lccc@{}}
\toprule
Problem & This work (classical) & Baseline & Section \\
\midrule
\#\textsc{maximal-clique} & \(O^{\ast}(3^{n/6})\) & \(O^{\ast}(3^{n/3})^{\dagger}\)\cite{bron1973finding} & \ref{sec:maxclique} \\
\#\textsc{Minimal-Separator} & \(O^{\ast}(1.2721^n)\) & \(O^{\ast}(1.6181^n)^{\dagger}\)\cite{takata2010space} & \ref{sec:separators} \\
\#\textsc{perfect-matching} ($\Delta\le 3$) & \(O^{\ast}(1.1611^n)\) & \(O^{\ast}(1.1918^n)^{\dagger}\) \cite{furer2012counting} & \ref{sec:pm-subcubic} \\
\#\textsc{independent-set} & \(O^{\ast}(1.1869^{n})\) & \(O^{\ast}(1.2041^{n})\)\cite{goldberg2021faster} & \ref{sec:IS-basic} \\
\#\textsc{2-SAT}$^{\ddagger}$ & \(O^{\ast}(1.2373^{n})\) & \(O^{\ast}(1.2377^{n})^{\dagger}\) \cite{wahlstrom2008tighter} & \ref{sec:2-SAT} \\
\bottomrule
\multicolumn{4}{@{}l@{}}{\footnotesize ${}^{\dagger}$\,Exact counting algorithm.\quad ${}^{\ddagger}$The \#\textsc{2-SAT} comparison is by variables; an explicit approximation gap is noted in \cite{cardinal2018solving,schmitt2013exploiting,thurley2011approximation}.}

\end{tabular}
\end{table}

We also obtain black-box quantum variants for the bounded-uSR core-estimation step (via \Cref{sec:quantum-aaronson,sec:Ambainis}); the resulting runtimes for the first three applications are summarised in \Cref{tab:quantum-apps}.
For \#\textsc{independent-set} and \#\textsc{2-SAT}, our fastest algorithms are
dominated by the (classical) preprocessing/decomposition tree that generates the family of bounded-uSR
cores; substituting a quantum core estimator alone therefore does not change the leading base
without rebalancing the preprocessing cutoffs. Such a rebalance is possible but yields only a modest
improvement and, to integrate it cleanly into the fully black-box oracle framework, one would
also like to compile the resulting sum-of-cores into a single bounded-uSR instance (cf. \Cref{def:boundedFunc}),
which we do not formalise to keep the exposition focused.

\begin{table}[t]
\centering
\caption{Quantum variants for the first three applications.}
\label{tab:quantum-apps}
\renewcommand{\arraystretch}{1.1}
\begin{tabular}{@{}lccc@{}}
\toprule
Problem & Classical & Quantum ($b^{1/3}$) & Quantum ($b^{1/4}$) \\
\midrule
\#\textsc{maximal-clique}
& \(O^{\ast}(3^{n/6})\)
& \(O^{\ast}(3^{n/9})\)
& \(O^{\ast}(3^{n/12})\) \\
\#\textsc{minimal-separator} & \(O^{\ast}(1.2721^n)\)  & \(O^{\ast}(1.1740^n)\) & \(O^{\ast}(1.1279^n)\) \\
\#\textsc{perfect-matching} ($\Delta\le 3$) & \(O^{\ast}(1.1611^n)\) & \(O^{\ast}(1.1047^n)\) & \(O^{\ast}(1.0776^n)\) \\
\bottomrule
\end{tabular}
\end{table}

\section{Overview and Roadmap}\label{sec:overview}

This section is a reader's guide.
The core estimator (enumerate up to $\lceil\sqrt{b(x)}\rceil$ and otherwise sample) was
introduced informally in \Cref{sec:overview:engine}; here we explain how the
rest of the paper turns that idea into a reusable toolkit.
In particular, we outline (i) the interface that exposes a recursion tree and a
recursion-compatible bound, (ii) how we synthesise the required enumerator and
sampler in a black-box way, (iii) why the additional tree-oracle structure is
needed for genuine plug-and-play quantum speed-ups, and (iv) how to extend the
approach beyond purely unweighted recurrences via decomposition into hard
enumeration cores.
Formal definitions and proofs appear in \Cref{sec:prelim,sec:uTotP,sec:Framework}.

\subsection{The Interface: Unweighted Recursions, Bounds and Feasibility Checks}\label{sec:overview:interface}

The algorithmic setting throughout the paper is a recursion tree of subinstances.
Each internal node contributes only the sum of its children, and each leaf
contributes $0$ or $1$; the answer $f(x)$ is the number of \emph{positive leaves}.
This is formalised as \emph{unweighted self-reducibility} (uSR) in
\Cref{sec:uTotP}.
A \emph{bounding function} $b$ upper-bounds the number of positive leaves and is
\emph{recursion compatible}, meaning that the bound mass assigned to a node
dominates the total mass assigned to its children.
Finally, we need a polynomial-time feasibility predicate for subinstances, i.e., a way to test whether a node has any positive leaf below it.
We write this as $\mathsf{Dec}_f(y)=[f(y)>0]$.
In our uTotP setting (uSR plus the $\#\mathrm{PE}$ condition; see \Cref{sec:uTotP}), this feasibility test is available in polynomial time and is exactly what enables the generic polynomial-delay enumeration routine used in Stage~1 (\Cref{sec:PolyDelay}).
Given this interface, the $\sqrt{b(x)}$ estimator from \Cref{sec:overview:engine}
becomes applicable.

\subsection{Why the Index-Tree Machinery Is Not Cosmetic}\label{sec:overview:oracles}

The index-tree machinery is not needed to obtain the first quantum improvement.
Combining the bound-driven sampler with standard quantum approximate counting already yields a generic $b(x)^{1/3}$ dependence: it only requires evaluating the predicate $m\mapsto t(h(x,\sigma_{\mathcal F^b,x}(m)))$ over the known universe $[b(x)]$, together with a short ``small vs.\ large'' preprocessing phase that can use any polynomial-delay enumerator (no structural access to the recursion tree is required) (\Cref{sec:quantum-aaronson}).

What does require additional structure is the further improvement from $b(x)^{1/3}$ to $b(x)^{1/4}$.  To go beyond marked-item counting we also accelerate the enumeration/lower-bounding step using Ambainis--Kokainis tree-size estimation, which assumes \emph{local oracle access} to a rooted tree (degree queries, $i$-th child queries, and often a parent oracle) in time $\poly(|x|)$.
Providing this access uniformly across problems is exactly why we build a canonical \emph{index tree} (via displacement indices) together with explicit $\mathsf{Child}/\mathsf{Parent}/\mathsf{Sibling}$ oracles.
This is the additional ingredient that makes the $b(x)^{1/4}$ speed-up black-box: once the uSR recursion is fixed, the same oracle interface can be fed into Ambainis--Kokainis and chained with standard quantum approximate counting to obtain the $1/4$ exponent (\Cref{sec:Ambainis}).

\subsection{Beyond Unweighted Problems: Decomposition into Hard Cores}\label{sec:overview:cores}

A single global unweighted recursion can have a trivial bound (e.g.\ $2^n$), even when the instance is mostly easy and only a small part behaves like ``counting $\approx$ enumeration.''
To capture this, we allow a preprocessing step that decomposes an instance into an easy remainder (handled by specialised counting tools) and a collection of hard enumeration cores $x_1,\ldots,x_\ell$, each equipped with its own bounded-uSR form and bound $b_i(x_i)$.

Our key extension is that we can approximate the \emph{sum} of the core contributions in time $O(\sqrt{\sum_i b_i(x_i)})$ (\Cref{thm:combineuTotP}), i.e.\ we pay a square root of the \emph{combined} core mass rather than a sum of square roots.
This mechanism underlies the strongest improvements in the more challenging applications, when the techniques are applied to problems which already have fast counting algorithms.

\subsection{Overview of the Applications}\label{sec:overview:apps}

We organise the applications to illustrate progressively more of the framework's capabilities: from direct ``extremal bound $\Rightarrow$ sampler'' plug-ins, to settings where we deliberately reshape the computation to expose an unweighted core, to settings where the main gain comes from decomposing into many hard enumeration cores and exploiting the $\sqrt{\sum b_i}$ phenomenon.

\begin{enumerate}[(i)]
\item \#\textsc{Maximal-Clique} (\Cref{sec:maxclique})

Here the framework behaves like a compiler for extremal bounds: Moon--Moser \cite{moon1965cliques} gives a tight bound on the number of maximal cliques, which directly serves as $b$, and the generic estimator yields an immediate square-root speed-up (\Cref{sec:maxclique}).
Quantitatively, Moon--Moser gives \(M(G)\le 3^{n/3}\), so the estimator runs in
\[
  O^{\ast}\!\left(\sqrt{3^{n/3}}\right)=O^{\ast}(3^{n/6}).
\]
The Bron--Kerbosch recursion \cite{bron1973finding} supplies the bounded-uSR form, and feasibility remains polynomial-time for the recursive states maintained by its \(R,P,X\) invariant.
Tsukiyama's algorithm \cite{tsukiyama1977new} shows the existence of a structured polynomial-delay enumeration perspective; our translation gives the corresponding sampler from the Moon--Moser recurrence.

\item \#\textsc{minimal-separator}  (\Cref{sec:separators})

The case of minimal separators is our next worked example.
Though as in maximal cliques, counting is not known to be faster than enumeration, this example adds a few complications to the previous one.
In this application, though polynomial delay enumerators are known, they do not provide the tree structure that Tsukiyama's algorithm does for maximal cliques.
By applying our generic interface, which exploits the fact that organising the solutions as a polynomial-depth tree is a consequence of the self-reducibility of the problem, we automatically generate a Tsukiyama-like algorithm for minimal separators.
The second complication compared to maximal cliques is that the feasibility check is not as straightforward.
Indeed, the completion problem is NP-Complete for minimal separators \cite{kenig2023listing} when formulated as an arbitrary subset problem.
This illustrates how it is important for the interface to work that the sub-instances obtained via the branching rule need to remain polynomial time checkable.

\item \#\textsc{Perfect-Matching} for $\Delta\le 3$ (\Cref{sec:pm-subcubic} and Appendix~\ref{sec:ImprovedPerfect})

This application highlights the ``trade weights for a square root'' design principle.
For exact counting algorithms, multiplicative and weighted branches are core to getting the best runtimes~\cite{furer2012counting}.
We restructure the computation to expose an unweighted leaf-counting core equipped with an explicit bound, so that both the classical $\sqrt{b}$ speed-up and the quantum $b^{1/3}/b^{1/4}$ speed-ups apply to the bottleneck step.

In Appendix~\ref{sec:ImprovedPerfect} we show that we can actually improve the branching to get an upper bound that matches the lower bound for the number of solutions.
This implies that the runtime derived by the algorithm is the best possible for our enumerate-or-sample algorithm (\Cref{alg:FullAlgo}).
However, unlike with \#\textsc{Maximal-Clique}, where we also have a tight bound on the number of solutions, there are known exact branching algorithms which use weights to speed up the counting of low-degree perfect matchings.
This suggests that we might be able to speed up approximate counting by separating some of the easy-to-count subinstances from the tree.
At a very high level, the goal is to perform a well-chosen shallow exploration of the tree to detect easy-to-count subinstances, decreasing the bound on what remains before applying our enumerate-or-sample routine.
Though we suspect it would lead to faster runtimes, we do not explore such algorithms for \#\textsc{Perfect-Matching}, but rather showcase the technique on the following two applications, where the gap between the best counting algorithm runtimes and the number of solutions is too large to surmount without decomposition.

\item \#\textsc{Independent-Set} (\Cref{sec:IS-basic})

This example is built around the sum-of-cores theorem (\Cref{thm:combineuTotP}).
A high-degree preprocessing branches only until either (i) the remainder is ``easy'' (handled by specialised
tools), or (ii) a budget is exhausted; the outcome is a decomposition into $\ell$ hard residual graphs
$x_1,\dots,x_\ell$, each equipped with a bounded-uSR form and a bound $b_i(x_i)$.
The key point is that $\ell$ is exponentially large, so a naive application of the single-instance estimator
(\Cref{thm:fullalgo}) \emph{per core} would cost $\sum_i \tilde O(\sqrt{b_i})$ and would not beat the best known bounds.
\Cref{thm:combineuTotP} is what makes the decomposition algorithmically meaningful: it gives a \emph{single} enumerate-or-sample
procedure over the disjoint union of cores, running in
\[
  \tilde O\!\left(\sqrt{\sum_{i=1}^\ell b_i(x_i)}\right),
\]
i.e.\ we pay the square root \emph{after} aggregating the bound mass.
In our analysis the preprocessing produces at most $2^{\alpha n}$ hard cores, each with bound at most $2^{\alpha n}$,
so $\sum_i b_i(x_i)\le 2^{2\alpha n}$ and \Cref{thm:combineuTotP} collapses this to time $2^{\alpha n}$, matching the cost of the
preprocessing itself. This is the quantitative source of the improved base in our \#\textsc{Independent-Set}
approximation algorithm for arbitrary graphs, improving over the general-graph
bound of Goldberg, Lapinskas and Richerby~\cite{goldberg2021faster}.

\item \#\textsc{2-SAT} (\Cref{sec:2-SAT})

The \#\textsc{2-SAT} acceleration is mainly a \emph{qualitative} milestone: the improvement in the
exponential base is modest, but the displayed bound lies below the currently
cited variable-parameter exact worst-case base for \#\textsc{2-SAT}.
This is of interest not because our framework can beat exact counting, as it does so for several
problems, but because it shows how to obtain such a separation for this specific, structurally
intermediate problem.

The status of \#\textsc{2-SAT} is somewhat atypical. For ``harder'' counting problems such as
\#\textsc{3-SAT}, the expressiveness of the constraint language leaves ample room for approximation
algorithms to be substantially faster than exact counting~\cite{thurley2011approximation,schmitt2013exploiting},
whereas for ``easier'' problems with richer combinatorial structure---e.g.\ \#\textsc{Independent-Set},
equivalently the all-negative $2$-CNF restriction with clauses $(\neg x_u \lor \neg x_v)$ (and, up to
complementing variables, also the all-positive restriction)---there are long-standing approximation
toolkits (e.g., correlation-decay/FPRAS regimes on restricted graph classes~\cite{sinclair2017spatial,goldberg2012approximating,weitz2006counting,goldberg2021faster}).
By contrast, general \#\textsc{2-SAT} sits in an intermediate regime where neither source of speed-up
applies cleanly, and it has been unclear how to obtain any provable improvement over the best exact
running time.

Our formulation uses the same high-level template as before: a preprocessing/self-reduction
decomposes the instance into polynomial-time solvable subinstances plus a \emph{sum} of hard
bounded-uSR cores, and we apply the enumeration--sampling engine to this aggregated sum (exploiting
the ``$\sqrt{\sum_i b_i}$'' effect rather than paying a sum of per-core costs).
The remaining ingredient is a problem-specific preprocessing bound: one needs a
sharper persistence argument for the degree-$6$ phase to turn the local branch
vectors into a final worst-case exponent. The proof in \Cref{sec:2-SAT}
supplies this bound and yields the displayed variable-parameter running time.

\end{enumerate}

\subsection{Structure of the Paper}\label{sec:overview:roadmap}

We begin in \Cref{sec:prelim} by fixing notation and the $(\varepsilon,\delta)$ approximation guarantees used throughout.
In \Cref{sec:uTotP} we formalise the unweighted framework (uSR/uTotP) and state the complexity-theoretic equivalence $\mathrm{uTotP}=\mathrm{TotP}$, the proof of which is given in appendix~\ref{sec:DeferredProofs}.
The technical core is developed in \Cref{sec:Framework}, which constructs the generic polynomial-delay enumerator and the bound-driven sampler, and proves both the $\sqrt{b(x)}$ approximate counting theorem and the $\sqrt{\sum_i b_i}$ extension for sums of cores.  We then instantiate the framework in \Cref{sec:Applications} on maximal cliques, minimal separators, subcubic perfect matchings, \#\textsc{Independent-Set}, and \#\textsc{2-SAT}.
Finally, \Cref{sec:Quantum} records the $b^{1/3}$ and $b^{1/4}$ quantum variants obtained by plugging standard quantum approximate-counting and tree-size-estimation primitives into the same interface. In particular, this section gives the details of why the organisation of solutions into a tree, derived in \Cref{sec:PolyDelay}, is essential to get the full quantum acceleration.

\section{Preliminaries}\label{sec:prelim}

\subsection{Notation and Conventions}
\label{sec:notation}

We collect the conventions that are used repeatedly in the technical sections.
The paper otherwise uses standard graph, formula, and complexity-theoretic
notation.

\begin{table}[H]
\centering
\renewcommand{\arraystretch}{1.15}
\small
\begin{tabular}{@{}lp{0.72\linewidth}@{}}
\toprule
\textbf{Symbol} & \textbf{Meaning / comment}\\
\midrule
$\mathbb{N}$              & $\{0,1,2,\ldots\}$, the set of non-negative integers\\
$\mathbb{N}_{+}$          & $\{1,2,3,\ldots\}$, the set of positive integers\\
$\mathbb{N}_{\bot}$       & $\mathbb{N}\cup\{\bot\}$, naturals augmented with a special symbol $\bot$\\
$\mathbb{N}^{\le c}$      & All tuples of elements of $\mathbb{N}$ whose length is \emph{at most} $c$\\
$[m]$                     & The set $\{1,\ldots,m\}$ \hfill($m\in\mathbb{N}_{+}$)\\
$\Sigma^{\ast}$           & All finite words over the implicit alphabet $\Sigma$\\
$\lambda$                 & Reserved dummy invalid word; genuine encodings are prefixed so that \(\lambda\) is never a valid input\\
$|x|$                     & Length of a word or tuple $x$\\
$\mathbf{v}.i$            & Appending $i$ to a tuple $\mathbf{v}$\\
$|\mathbf{v}|$            & Length of the tuple $\mathbf{v}$\\
$h(x,\mathbf{v})$         & Iterated application of $h$, defined in \eqref{eq:hdef}\\
$<_{\mathrm{lex}}$        & Strict lexicographic order on words or tuples\\
$\operatorname{poly}(\cdot)$ & An unspecified polynomial in its argument\\
$O^\ast(\cdot)$           & Running-time notation hiding polynomial factors in the input size\\
$\displaystyle\sum_{i\in\emptyset}(\cdot)$ & Empty-sum convention: the value is $0$\\
\bottomrule
\end{tabular}
\caption{Frequently used notation. Unless stated otherwise, all logarithms are base $2$.}
\label{tab:notation}
\end{table}

\paragraph{Tuples and words.}
Throughout the paper we treat tuples as words: appending an integer to a tuple
is written $\mathbf{v}.i$, and concatenation is associative.  For a tuple
$\mathbf{v}=(v_1,\dots,v_d)$ we use the usual 1-based indexing
$\mathbf{v}_j=v_j$.

\paragraph{Encodings.}
We represent all finite objects (graphs, formulas, tuples, integers, etc.) by
binary strings under a fixed reasonable encoding.  We write $\langle X\rangle$
for the encoding of an object $X$, and we freely identify $X$ with
$\langle X\rangle$ when convenient.  All encoding/decoding and tuple
packing/unpacking used in the paper is computable in time polynomial in the
encoding length.

\paragraph{Runtimes.}
Whenever a family of functions is said to run in $\operatorname{poly}(|x|)$
time, the underlying polynomial is understood to be independent of
approximation parameters $(\varepsilon,\delta)$ unless explicitly stated.

\subsection{($\varepsilon,\delta$)-Approximations for Counting}
\subsubsection{Problem Statement}
\begin{definition}[$(\varepsilon,\delta)$-approximation]\label{def:epsdelta}
    For a given counting problem meeting our criteria, our approximate counting algorithm provides an $(\varepsilon,\delta)$-approximation to the true number of solutions $M$ when:
    \begin{enumerate}[(a)]
        \item The estimate result $\hat{M}$ has at most $\varepsilon \cdot M$ relative error if it succeeds
        \item The probability of failure is at most $\delta$
    \end{enumerate}
\end{definition}

\subsubsection{Chernoff Bound}
Throughout our paper, we apply the multiplicative Chernoff bound~\cite{mitzenmacher2017probability}
to control relative deviations of sums of independent Bernoulli random variables.
Let $X=\sum_{i=1}^T X_i$ be a sum of independent $\{0,1\}$-valued variables and let
$\mu=\mathbb{E}[X]$. For any $0<\varepsilon\le 1$,
\[
\Pr\!\big(|X-\mu|\ge \varepsilon \mu\big)\ \le\ 2\exp\!\left(-\frac{\varepsilon^2\mu}{3}\right).
\]
(We reserve $\delta$ for the failure probability in $(\varepsilon,\delta)$-approximations.)

\subsection{TotP}

We assume the reader is familiar with standard notation in complexity theory.
\subsubsection{Standard Characterisations of TotP}

\begin{definition}
\#P is the class of functions \(f : \{0,1\}^{*} \to \mathbb{N}\) for which there exists a nondeterministic
polynomial-time Turing machine (NPTM) \(M_f\) such that the number of accepting paths of \(M_f\) on
input \(x\) equals \(f(x)\).

\(\mathrm{TotP}\) is the class of functions \(f : \{0,1\}^{*} \to \mathbb{N}\) for which there exists a nondeterministic
polynomial-time Turing machine (NPTM) \(M_f\) such that the number of \emph{all} computation
paths of \(M_f\) on input \(x\) equals \(f(x) + 1\) (the offset by 1 is necessary since a machine cannot have no paths).
\end{definition}

\subsubsection{Characterisation Using Self-Reducibility}
Another useful characterisation of the TotP class, proved in \cite{pagourtzis2006complexity}, is as the Karp-closure of self reducible problems in \#PE, under the following notion of self reducibility.
\#PE is the class of functions \(f\) in \#P for which the decision version, i.e., the problem of deciding
whether \(f(x) > 0\), is in \(\mathsf{P}\).
Self-reducibility is defined as follows:
\begin{definition}[Self-reducibility]\label{def:SR}
A function \( f : \Sigma^{\ast} \to \mathbb{N} \) is said to have a
\emph{polynomial-time self-reducible definition} if there exist polynomial-time
computable functions \(r,q:\Sigma^\ast\to\mathbb N\), maps
\[
  h:\Sigma^\ast\times\mathbb N_+\to\Sigma^\ast,\qquad
  g:\Sigma^\ast\times\mathbb N_+\to\mathbb N,
\]
and \(t:\Sigma^\ast\to\{0,1\}\) such that the following hold.
We reserve \(\lambda\) as a dummy invalid instance; ordinary encodings can be
prefixed so that \(\lambda\) is never a genuine input.

\begin{enumerate}[(i)]
\item\textbf{Polynomial bounds.}
      There is a constant \(c\) such that, for every \(x\),
      \[
        r(x)\le (|x|+1)^c
        \qquad\text{and}\qquad
        q(x)\le (|x|+1)^c .
      \]

\item\textbf{Invalid branches.}
      We have \(r(\lambda)=q(\lambda)=t(\lambda)=0\), and for every
      \(x\) and \(i>r(x)\), \(h(x,i)=\lambda\).

\item\textbf{Recursive decomposition.}
      For every \(x\),
      \[
        f(x)\;=\;t(x)\;+\;\sum_{i=1}^{\,r(x)} g(x,i) f\!\bigl(h(x,i)\bigr).
      \]

\item\textbf{Polynomial depth.}
      Every valid recursive path from \(x\) has length at most \(q(x)\):
      if \(x_0=x\), \(x_{m}=h(x_{m-1},j_m)\), and
      \(1\le j_m\le r(x_{m-1})\) for all \(m\), then
      \(r(x_d)=0\) for every \(d\ge q(x)\).

\item\textbf{Instance size preservation.}
      Along every valid recursive path from \(x\), every reached instance has
      length at most \(\poly(|x|)\).
\end{enumerate}
\end{definition}

\begin{definition}[$\#PE_{SR}$]\label{def:SRPE}
     The set of functions in $\#PE$ that satisfy \Cref{def:SR} is denoted as $\#PE_{SR}$.
\end{definition}

\begin{definition}[TotP]\label{def:TotP}
    TotP is the Karp closure of $\#PE_{SR}$.
\end{definition}

\section{Unweighted TotP (uTotP)} \label{sec:uTotP}
\subsection{Unweighted Self-Reducibility}
In this section we introduce our reframing of the class TotP.

\begin{definition}[Unweighted self-reducibility]
\label{def:uredp}
A function \( f : \Sigma^{\ast} \to \mathbb{N} \) is said to have an
\emph{unweighted polynomial-time self-reducible definition} if there exist
polynomial-time computable functions \(r,q:\Sigma^\ast\to\mathbb N\), a
polynomial-time computable map \(h:\Sigma^\ast\times\mathbb N_+\to\Sigma^\ast\),
and a polynomial-time computable map \(t:\Sigma^\ast\to\{0,1\}\) such that the
following hold. As above, \(\lambda\) is a reserved dummy invalid instance.
\begin{enumerate}[(i)]
\item\textbf{Polynomial bounds.}\label{def:qrfunc}
      There is a constant \(c\) such that, for every \(x\),
      \[
        r(x)\le (|x|+1)^c
        \qquad\text{and}\qquad
        q(x)\le (|x|+1)^c .
      \]

\item\textbf{Invalid branches.}\label{cond:noBranch}
      \[
        r(\lambda)=q(\lambda)=t(\lambda)=0,\qquad
        h(\lambda,i)=\lambda\ \text{for all }i,\qquad
        i>r(x)\Rightarrow h(x,i)=\lambda .
      \]

\item\textbf{Recursive decomposition.}\label{cond:USR}
      \[
        f(x)\;=\;t(x)\;+\;\sum_{i=1}^{\,r(x)} f\!\bigl(h(x,i)\bigr).
      \]

\item\textbf{Polynomial depth.}\label{cond:polyDepth}
      Every valid recursive path from \(x\) has length at most \(q(x)\):
      if \(x_0=x\), \(x_m=h(x_{m-1},j_m)\), and
      \(1\le j_m\le r(x_{m-1})\) for all \(m\), then
      \(r(x_d)=0\) for every \(d\ge q(x)\).

\item\textbf{Instance size preservation.}\label{cond:InstanceSize}
      Along every valid recursive path from \(x\), every reached instance has
      length at most \(\poly(|x|)\).
\item\textbf{Leaf value and unweightedness.}\label{cond:unweighted}
      \[
  \big(q(x) = 0 \;\Rightarrow\; r(x) = 0\big) \quad \text{and} \quad \big(r(x) > 0 \;\Rightarrow\; t(x) = 0\big).
\]

      Thus every non-leaf contributes nothing except the sum of its
      recursive children, while every leaf contributes either $0$ or $1$.
\end{enumerate}

\end{definition}

The maps in a uSR definition specify the recursion tree whose positive leaves
are counted, as formalised below.

\begin{definition}[uSR definition of a function]\label{def:func}
    For the purposes of this paper, the \textit{uSR definition of a function} will be given by a tuple $\mathcal{F}=(f,h,t,r,q)$, where $f,h,t,r$ and $q$ are as defined in \Cref{def:uredp}.
    When meaning is unambiguous we may simply refer to the uSR definition of a function as the definition of a function.
    We may also refer to $\mathcal{F}=(f,h,t,r,q)$ as an \emph{uSR form} of $f$.
\end{definition}

We note that unweighted self-reducibility is a property of the definition rather
than the function. However, when the definition is unambiguous, we also say that
the function is unweighted self-reducible.

\subsection{Our New Characterisation of TotP}\label{sec:NewChar}

First define the following set of functions
\begin{definition}[$\#PE_{uSR}$]
     The set of functions in $\#PE$ that have a definition satisfying the properties \Cref{def:uredp} is denoted as $\#PE_{uSR}$.
\end{definition}

We can now define the class uTotP as such:

\begin{definition}[uTotP]\label{def:uTotP}
    uTotP is the Karp closure of $\#PE_{uSR}$.
\end{definition}

The only difference with the self-reducibility based definition of TotP (\Cref{def:TotP}) is that we are enforcing the stronger property of unweighted self-reducibility, constraining functions $g$ to always be $1$ and $t$ to always be $0$ except in leaves where it may be $0$ or $1$.
Intuitively, a computation tree for an \(f\in\#PE_{uSR}\) simply \emph{counts its leaves}: each internal node passes the sum of its children upward, while each leaf returns \(1\) or \(0\).
\Cref{thm:uTotpeqTotp} shows that the two characterisations are in fact equivalent in terms of the complexity class they define.

\begin{restatable}{theorem}{TOTPUTOTP}\label{thm:uTotpeqTotp}
    uTotP=TotP.
\end{restatable}

\begin{proof}
    See Appendix~\ref{sec:DeferredProofs}.
\end{proof}

\subsection{Bounding Functions}

For the reframing to be exploited algorithmically, we need to introduce the concept of a bounding function.

\begin{definition}[Bounding function]\label{def:bounding}
Let $f:\Sigma^{*}\to\mathbb N$ be a unweighted polynomial time self-reducible function with an appropriate unweighted polynomial time self-reducible definition \(\mathcal{F}=(f,h,t,r,q)\). A \emph{bounding function} for \(\mathcal{F}\)  is a polynomial-time computable map $b:\Sigma^{*}\rightarrow\mathbb N$ such that
\begin{align}
     \forall{x}\in\Sigma^{*},&\quad f(x)\;\le\;b(x)\\
     \forall{x}\in \Sigma^{*},& \quad \displaystyle\sum_{i=1}^{r(x)}b(h(x,i))\leq b(x) \label{eq:recursiveComp}\\
     \forall{x} \in \Sigma^{*}, & \quad (r(x)=0)\implies (b(x)\le 1)\label{eq:smallbound}
     \\
     \forall{x} \in \Sigma^{*}, & \quad t(x)\leq b(x)\label{eq:leafbound}
\end{align}

We specifically call \Cref{eq:recursiveComp} \emph{recursive compatibility}.
\end{definition}

The last two conditions imply the leaf tightness used by the sampler: if \(y\)
is a positive leaf, then \(r(y)=0\) and \(t(y)=1\), so
\(1=t(y)\le b(y)\le 1\), and hence \(b(y)=1\).
The formal algorithms below take \(b\) to be integer-valued so that the sampler
can be written as a uniform draw from \([b(x)]\). This is the finite-ticket
presentation of the bound-mass view used in the main text: real branching
numbers such as \(1.1869\) or \(1.2373\) summarise integer recurrence bounds, or
safe integer roundings of real potentials, with the same \(O^\ast\) exponent.

\begin{definition}[Definition of a uSR bounded function]\label{def:boundedFunc}
    We say that $\mathcal{F}=(f,h,t,r,q,b)$ is an \emph{ uSR bounded definition of a function} if $b$ is a function satisfying \Cref{def:bounding}.
    When clear from context we may use the term definition of a function to mean the definition of a bounded function.
    We may also denote a uSR bounded function with the superscript notation $\mathcal{F}^b$ when $b$ needs to be mentioned explicitly.
\end{definition}

\section{Algorithmic Interface for uTotP}\label{sec:Framework}

This section turns the informal ``enumerate-or-sample'' engine from the introduction
into a reusable algorithmic interface for functions given in bounded uSR form.
Fix a bounded uSR form $\mathcal F^b=(f,h,t,r,q,b)$ and an instance $x$.
The unweighted recursion induces an implicit rooted tree of subinstances:
internal nodes branch to subinstances $h(\cdot,i)$ and leaves contribute $t(\cdot)\in\{0,1\}$.
Thus $f(x)$ is exactly the number of \emph{positive leaves} reachable from $x$.

\paragraph{Motivation for the interface.}
The point of introducing a bounded uSR \emph{interface} is that it makes
counting algorithms modular.
Once a core counting task is presented via
\((h,t,r,q)\) together with (i) a recursion-compatible bound \(b(\cdot)\)
and (ii) feasibility access \(\mathsf{Dec}_f\),
we can treat the recursion tree as a black box and obtain generic procedures:
a polynomial-delay enumerator (used when the count is small),
a bound-driven sampler (used when the count is large),
and, crucially, compositions of these procedures across many cores.
This composability enables the preprocess-and-enumerate-or-sample strategy which results in a \(\sqrt{\text{(total bound)}}\)-time approximation routine and that we use to find faster algorithms for problems which have traditionally been approached using weighted recursive trees.

\paragraph{Interface assumptions.}
We assume (i) a polynomial-time computable upper bound $B=b(x)$ that is
recursion-compatible, and (ii) access to a feasibility predicate
$\mathsf{Dec}_f(y)=[f(y)>0]$ for subinstances $y$ (used only for enumeration).
This feasibility access is not an additional promise beyond uTotP/TotP: uTotP is defined as the Karp closure of unweighted self-reducible functions in $\#\mathrm{PE}$ (\Cref{sec:uTotP}), so $\mathsf{Dec}_f$ is polynomial-time for the cores we work with; and \Cref{thm:uTotpeqTotp} shows that this scope covers all of TotP via Karp reductions.

\paragraph{Roadmap and conceptual arc.}
\textsc{Count} is a two-stage \emph{enumerate-or-sample} estimator.
Stage~1 enumerates up to \(k=\lceil\sqrt{B}\rceil\) positive leaves; if fewer exist we
recover \(f(x)\) exactly.
Otherwise \(f(x)\ge k\), so a uniformly random ``bound ticket'' succeeds with
probability \(f(x)/B\ge 1/\sqrt{B}\), and Stage~2 estimates this probability via sampling
and Chernoff bounds.

The key message is that both stages only use the abstract interface:
Stage~1 needs feasibility (to traverse only solution-bearing branches),
while Stage~2 needs only the bound mass \(b(\cdot)\) (to implement ticket sampling).

The technical work of this section is to show that the above split is not rhetorical.
Given only a bounded uSR form \((h,t,r,q,b)\) and feasibility access \(\mathsf{Dec}_f\),
we can derive both ingredients needed by the estimator, with no further problem-specific design, along with additional structure needed for quantum acceleration.
\Cref{sec:BranchIndices} introduces branch indices and the displacement-index coordinate system,
which is the technical device that lets us navigate the recursion tree while pruning dead subtrees using only \(\mathsf{Dec}_{f}\).
\Cref{sec:PolyDelay} uses \(\mathsf{Dec}_f\) to define a pruned, solution-bearing index tree (via
displacement indices) and to implement local navigation oracles; this yields a
lexicographic polynomial-delay enumerator of positive leaves and supplies exactly the
``small-count'' primitive required by Stage~1.
\Cref{sec:Sampling}, by contrast, never consults \(\mathsf{Dec}_f\): recursion compatibility of \(b\)
alone supports a ticket-unranking sampler \(\sigma_{F^b,x}\) whose success probability is
precisely \(f(x)/b(x)\), giving the Bernoulli trials needed in Stage~2.
\Cref{sec:GeneralAlgo} combines these two derived primitives into \textsc{Count} and
completes the Chernoff analysis.

Conceptually, this abstracts a recurring mismatch in the literature: branching recurrences
and polynomial-delay enumeration are usually developed together (often tailored to a
specific search order), while (almost-)uniform sampling is typically obtained by separate
machinery and does not automatically align with the same recursion tree.
The bounded-uSR interface makes ``one recursion + feasibility oracle \(\Rightarrow\) both a structured enumerator
and a bound-driven sampler'' a black-box implication, under explicit and checkable
assumptions.

This packaging is deliberate: in the classical branching-algorithm literature, a
recurrence bound typically certifies a backtracking enumerator, whereas obtaining
a (near-)uniform sampler (and, even more, one that respects the same search-tree
structure) usually requires additional, problem-specific work.
Here the bounded-uSR interface isolates a clean set of sufficient conditions under which
both objects follow uniformly from the same recursion data.

\Cref{sec:toTotP} introduces an essential extension of the basic algorithm.
Many TotP algorithms naturally decompose an instance into (i) parts that admit
problem-specific counting shortcuts and (ii) a residual collection of hard
enumeration cores on which the computation degenerates to leaf-counting in
a recursion tree.
Our interface makes this decomposition algorithmically actionable by requiring that
each core be exposed as a bounded-uSR instance
\((\mathcal F_i^{b_i},x_i)\), i.e., an explicit unweighted recursion
\((h_i,t_i,r_i,q_i)\) together with a recursion-compatible bound \(b_i\).
This certification is exactly what lets us synthesise, uniformly and black-box,
(i) a polynomial-delay enumerator (for the ``small count'' regime) and
(ii) a bound-driven sampler (for the ``large count'' regime).

Crucially, once the cores are packaged this way, we can approximate their
total contribution without approximating each \(f_i(x_i)\) separately.
A naive approach would run the estimator on every core and sum the outputs,
incurring time \(\sum_i \widetilde O(\sqrt{B_i})\) and requiring nontrivial
allocation of failure probability and relative-error budgets across many terms.
Instead, \Cref{thm:combineuTotP} implements a single enumerate-or-sample
procedure over the disjoint union of cores:
draw one ticket uniformly from the aggregate bound mass \(B=\sum_i B_i\),
route it to core \(i\) using prefix sums, and test whether it hits a positive leaf.
This mixed experiment succeeds with probability \(F/B\), where \(F=\sum_i f_i(x_i)\),
so the Chernoff analysis goes through unchanged and yields an
\((\varepsilon,\delta)\)-approximation in time
\(\widetilde O\!\left(\sqrt{\sum_i B_i}\,\varepsilon^{-2}\log(1/\delta)\right)\).
In this sense, ``decompose into hard cores'' becomes a quantitative black-box
primitive rather than a problem-specific trick.

\paragraph{Notation.}
For convenience, \Cref{tab:framework-notation} summarises the main symbols and oracles used throughout this section.

\begin{table}[t]
\centering
\caption{Main symbols and oracles used in \Cref{sec:Framework}.}
\label{tab:framework-notation}
\renewcommand{\arraystretch}{1.05}
\begin{tabular}{ll}
\toprule
Symbol / oracle & Meaning (informal)\\
\midrule
$h(x,i)$ & $i$-th recursive sub-instance of $x$\\
$r(x)$ & fan-out (number of recursive calls)\\
$t(x)\in\{0,1\}$ & leaf predicate (positive leaf iff $t=1$)\\
$b(x)$ & recursion-compatible upper bound on $f(x)$\\
$\mathsf{Dec}_f(x)$ & feasibility: does subtree at $x$ contain a positive leaf?\\
$\beta_{\mathcal F,x}(\mathbf u)$ & displacement index $\to$ branch index\\
\textsc{DFSEnum}$_{\mathcal F,x}$ & polynomial-delay enumerator of positive leaves\\
$\sigma_{\mathcal F^b,x}(m)$ & sampler: map $m\in[b(x)]$ to a (possibly invalid) branch index\\
\bottomrule
\end{tabular}
\end{table}

\subsection{Branch and Displacement Indices}\label{sec:BranchIndices}

Our algorithms operate on the recursion tree induced by $h$.
A \emph{branch index} names a root-to-node path in the original recursion tree, a \emph{complete branch index} names a leaf of that tree, and a \emph{displacement index} is the corresponding coordinate system after infeasible subtrees have been deleted using \(\mathsf{Dec}_f\).
For enumeration and for later quantum speed-ups, it is convenient to also work in a
second coordinate system that skips branches whose subtrees contain no solutions.
This leads to \emph{displacement indices} and the map
$\beta_{\mathcal F,x}$ converting between the two.
Displacement indices provide a coordinate system for the \emph{pruned} recursion tree
containing only solution-bearing branches, and they support local navigation oracles
(\(\mathsf{Child}\), \(\mathsf{Parent}\), \(\mathsf{Sibling}\)) that we later reuse for quantum speed-ups.

To simplify notation, we define \(h:\Sigma^* \times \mathbb{N}^d \to \Sigma^*\) for all \(d\in \mathbb{N}\) by iterating the map \(h:\Sigma^* \times \mathbb{N}\to \Sigma^*\) from \Cref{def:uredp} as follows:

\begin{equation}\label{eq:hdef}
      h(x,\mathbf{v}) :=
    \begin{cases}
        h(\cdots h(h(x, v_1), v_2), \ldots, v_{d}) &\text{if $d>0$, and}\\
        x &\text{otherwise (i.e. $d=0$ and $\mathbf{v} = \emptyset$)},
    \end{cases}
\end{equation}

for \( \mathbf{v} = [v_1,v_2,\ldots , v_d] \in \mathbb{N}_+^d \).

\begin{definition}[Branch index]\label{def:branchIndex}
    Given a function definition \(\mathcal{F}\) and \(x\in \Sigma^{\ast}\), a \emph{branch index} for $(\mathcal{F},x)$ is a tuple \(v\in \mathbb{N}_+^{\leq q(x)}\).
    When obvious from context or irrelevant what $\mathcal{F}$ and $x$ are, we simply call $v$ a branch index.
    We say that a branch index is \emph{valid} if \(h(x,v)\neq \lambda\) and we say that it is \emph{complete} if it is valid and \(r(h(x,\mathbf{v}))=0\).
\end{definition}

\noindent\emph{Intuition:} $\mathbf v=(v_1,\dots,v_d)$ records which child was taken at each
recursive step; ``complete'' means the recursion has reached a leaf.

It will be useful to introduce another coordinate system, called the displacement index.
Before we define displacement indices formally, we introduce \Cref{alg:disToBranch} to give some context.

\begin{algorithm}[H]
\DontPrintSemicolon
\caption{\(\beta_{\mathcal{F},x}(\mathbf{u})\) - Goes from displacement index to branch index}
\label{alg:disToBranch}
\KwIn{\(\mathbf{u}\in \mathbb{N}^{q(x)}\)}
\KwOut{A branch index for \(\mathcal{F}\) and \(x\).}

\(\mathbf{v} \gets \emptyset\)

\(d \gets 1\)

\While{$r\bigl(h(x,\mathbf{v})\bigr)\neq 0$}{ \label{line:stopCond}
    $i\gets |\mathbf{v}|+1$ \tcp*{Current depth}
    $\mathcal{C} \gets \bigl\{\,j \in [r(h(x,\mathbf{v}))] : \mathsf{Dec}_f(h(x,\mathbf{v}.j))\,\bigr\}$ \tcp*{Find children with solutions}

    \If{$\mathbf{u}_{i} < |\mathcal{C}|$}{
        $\mathbf{v}\gets \mathbf{v}.(\mathcal{C}_{\mathbf{u}_{i}+1})$\label{line:appendChild}
    }
    \Else{
        $\mathbf{v} \gets \mathbf{v}.(r(h(x,\mathbf{v}))+1)$ \tcp*{Branch index $\mathbf{v}$ guarantees $h(x,\mathbf{v})=\lambda$}
    }
    $d \gets d+1$
}
\If(\tcp*[f]{the branch is valid and all unused components of \(\mathbf{u}\) are zero})%
   {$\displaystyle h(x,\mathbf v)\neq\lambda
     \text{ and }\bigl(\forall j\ge d : \mathbf{u}_j=0\bigr)$}%
{%
   \Return{$\mathbf{v}$}\tcp*{Return complete and valid branch index} \label{line:Valid1}
}%
\Else{
   \Return{$(r(x)+1)$}\tcp*{Invalid fallback branch index}\label{line:Invalid1}
}
\end{algorithm}
\begin{lemma}\label{lem:betaRuntime}
    $\beta_{\mathcal{F},x}$ runs in polynomial time.
\end{lemma}
\begin{proof}
    The while loop always appends an index to $\mathbf{v}$. According to \Cref{cond:polyDepth} and the definition of $q$ (\Cref{def:qrfunc}), after polynomial iterations of the loop we will get $r(h(x,\mathbf v))=0$.
\end{proof}

\begin{definition}[Displacement index]\label{def:disIndex}
    Given a function definition \(\mathcal{F}\) and \(x\in \Sigma^{\ast}\), a \emph{displacement index} for $(\mathcal{F},x)$ is a tuple \(\mathbf{u}\in \mathbb{N}^{q(x)}\).
    When obvious from context or irrelevant what $\mathcal{F}$ and $x$ are, we simply call $\mathbf{u}$ a displacement index.
    We say that a displacement index $\mathbf{u}$ is \emph{valid} when $h(x,\beta_{\mathcal{F},x}(\mathbf{u}))\neq \lambda$.
\end{definition}

\noindent\emph{Intuition:} $\mathbf u$ records, at each depth, the \emph{rank among feasible children}
(those with at least one solution), rather than the raw child number in $[r(\cdot)]$.

Conceptually, $\beta_{\mathcal F,x}(\mathbf u)$ walks down the recursion tree,
and at depth $i$ chooses the $(\mathbf u_i+1)$-st child whose subtree is nonempty.
Once this walk reaches a leaf, every unused suffix coordinate of \(\mathbf u\)
must be zero; otherwise the vector is sent to the invalid fallback branch.

\subsection{Structured Polynomial Delay Algorithm}\label{sec:PolyDelay}

Stage~1 of \Cref{alg:FullAlgo} requires an \emph{output-sensitive} routine:
either enumerate all solutions when $f(x)$ is small, or stop after $\sqrt{B}$ outputs.
We therefore construct a generic polynomial-delay enumerator for positive leaves.
The main idea is to traverse the tree of displacement indices, which contains only
solution-bearing branches and has polynomial depth, giving a polynomial delay algorithm.
Furthermore, the enumeration is represented by a locally navigable tree, a structure that we exploit in \Cref{sec:Quantum}.

Since in this paper we are working with functions $f\in \#PE$, we assume that we have a polynomial time algorithm to determine whether $f>0$.
In the rest of the paper we will make use of the following oracle to represent the decision version of our function.
\[
\mathsf{Dec}_f(x)\;:=\;\begin{cases}
\textsc{true} & \text{if } f(x)>0,\\
\textsc{false} & \text{otherwise.}
\end{cases}
\]

\subsubsection{Child, Parent, and Sibling Oracles}\label{sec:GraphOracles}

\begin{algorithm}[H]
    \DontPrintSemicolon
    \caption{\(\mathsf{Child}_{\mathcal{F},x}(\mathbf{u},i)\)}
    \label{alg:child}
    \KwIn{Displacement index \(\mathbf{u}\), rank $i\ge 1$}
    \KwOut{Displacement index \(\mathbf{u}'\) or \(\bot\) if fewer than $i$ children exist}

    $\mathbf w\gets\mathbf u$                                              \tcp*{mutable working copy}
\If{$\exists j:\mathbf u_j>0$}{
  $r\gets\max\{j\mid\mathbf u_j>0\}$\tcp*{right-most non-zero}
}
\Else{
  $r\gets1$\tcp*{root displacement index}
}
$p\gets|\mathbf w|$\tcp*{current coordinate pointer}

\While{$p\ge r$}{
    $\mathbf w_p\gets\mathbf w_p+1$\tcp*{tentative bump}

    \If{$h\!\bigl(x,\beta_{\mathcal F,x}(\mathbf w)\bigr)\neq\lambda$}{
         $i\gets i-1$\;
         \If{$i=0$}{
         \Return{$\mathbf w$}
         }
    }
    $\mathbf w_p\gets \mathbf w_p-1$\tcp*{revert on failure}
    $p\gets p-1$\tcp*{move one slot left}

}
\Return{$\bot$}\tcp*{fewer than $i$ valid children}

\end{algorithm}

\begin{lemma}\label{lem:polyChild}
\(\mathsf{Child}_{\mathcal F,x}\) runs in polynomial time.
\end{lemma}

\begin{proof}
The pointer \(p\) starts at \(q(x)=|\mathbf u|\) and is decremented once
per loop, so the \textbf{while}-loop executes at most \(q(x)+1\)
iterations; recall \(q(x)=\operatorname{poly}(|x|)\).

Inside an iteration we (i) update a single coordinate of
\(\mathbf w\), (ii) compute \(\beta_{\mathcal F,x}(\mathbf w)\) and one call to
\(h\), and (iii) adjust \(i\).
Both \(\beta_{\mathcal F,x}\) and \(h\) are polynomial-time by construction, hence
an iteration costs \(\operatorname{poly}(|x|)\).

Multiplying a polynomial bound per iteration by \(O(q(x))\) iterations
gives an overall running time \(\operatorname{poly}(|x|)\).
\end{proof}

\begin{algorithm}[H]
    \DontPrintSemicolon
    \caption{\(\mathsf{Parent}_{\mathcal{F},x}(\mathbf{u})\)}
    \label{alg:parent}
    \KwIn{Displacement index \(\mathbf{u}\)}
    \KwOut{$\mathbf u'$ (parent displacement) or $\bot$ if $\mathbf u$ is the root}
    \If{$\forall j:\;\mathbf u_j = 0$}{
    \Return{$\bot$} \tcp*{root has no parent}
}
$\mathbf w \gets \mathbf u$ \tcp*{working copy}
$r \gets \max\{j \mid \mathbf w_j > 0\}$\;
$\mathbf w_r \gets \mathbf w_r - 1$\;

\Return{$\mathbf w$}\tcp*{this is $\mathbf u'$}
\end{algorithm}

\begin{lemma}\label{lem:polyParent}
\(\mathsf{Parent}_{\mathcal F,x}\) runs in polynomial time.
\end{lemma}

\begin{proof}
The vector \(\mathbf u\) has length \(q(x)=\poly(|x|)\).
Scanning once to locate the right-most non-zero entry costs
\(O(q(x))\); the subsequent decrement is constant time.
Hence the whole routine is \(O(q(x))=\poly(|x|)\).
\end{proof}

\begin{algorithm}[H]
\DontPrintSemicolon
\caption{\(\mathsf{Sibling}_{\mathcal{F},x}(\mathbf{u})\)}
\label{alg:sibling}
\KwIn{Displacement index \(\mathbf{u}\)}
\KwOut{$\mathbf u'$ (next sibling) or $\bot$ if $\mathbf u$ is the root or $\mathbf{u}$ is the last child}
\If{$\forall j:\;\mathbf u_j = 0$}{
    \Return{$\bot$} \tcp*{root node}
}
$\mathbf p \gets \mathsf{Parent}_{\mathcal F,x}(\mathbf u)$\;
$i \gets 1$\;
\While{True}{
$\mathbf{w}\gets \mathsf{Child}_{\mathcal{F},x}(\mathbf{p},i)$\;
\If{$\mathbf{w}=\bot$}{
\Return{$\bot$} \tcp*{last child}
}
\If{$\mathbf{w}=\mathbf{u}$}{
    $i\gets i+1$\;
    \Return{$\mathsf{Child}_{\mathcal{F},x}(\mathbf{p},i)$}\tcp*{may be bot if no next child}
}
$i\gets i+1$
}
\end{algorithm}

\begin{lemma}\label{lem:polySibling}
\(\mathsf{Sibling}_{\mathcal F,x}\) runs in polynomial time.
\end{lemma}

\begin{proof}
The length of every displacement vector is \(q(x)=\poly(|x|)\).

\emph{Outside the loop.}
One call to \(\mathsf{Parent}_{\mathcal F,x}\) costs
\(\poly(|x|)\) by \Cref{lem:polyParent}.

\emph{Inside the loop.}
Each pass invokes \(\mathsf{Child}_{\mathcal F,x}\) once, which is
\(\poly(|x|)\) by \Cref{lem:polyChild}, and then increments \(i\).
For a fixed parent displacement vector \(\mathbf p\), each valid child is found
by a distinct value of \(i\), and the next call after the last valid child
returns \(\bot\).  Since displacement vectors have length \(q(x)\), this gives
at most \(q(x)+1\) calls to \(\mathsf{Child}_{\mathcal F,x}\).

Thus the total running time is
\(\poly(|x|)\times O(q(x)) = \poly(|x|)\).
\end{proof}

\begin{definition}[Tree of displacement indices]\label{def:indexTree}
    We denote by $\mathcal{T}_{\mathcal{F},x}$ the tree of displacement indices defined by $\mathsf{Child}_{\mathcal{F},x}$.
\end{definition}

\begin{lemma}\label{lem:depthT}
    The depth of $\mathcal{T}_{\mathcal{F},x}$ is bounded by a polynomial.
\end{lemma}

\begin{proof}
Let
\[
  R_x:=\max\{r(y): y \text{ is reachable from } x
       \text{ by a valid recursive path}\}.
\]
By the polynomial fan-out and instance-size assumptions in the uSR definition,
\(R_x\le \poly(|x|)\).
For any valid displacement index \(\mathbf u\), each used coordinate is the
zero-based rank of a feasible child and is therefore at most \(R_x-1\), and
there are at most \(q(x)\) used coordinates.
Along a parent chain in \(\mathcal T_{\mathcal F,x}\), the quantity
\(\sum_j \mathbf u_j\) decreases by one at each step.
Hence every root-to-node path has length at most \(q(x)(R_x-1)\), which is
polynomial in \(|x|\).
\end{proof}

\subsubsection{DFS Polynomial Delay Algorithm}
\begin{algorithm}[H]
\DontPrintSemicolon
\caption{\textsc{DFSEnum}$_{\mathcal{F},x}$ - DFS based polynomial delay algorithm}
    \label{alg:dfs-enum}
    \KwOut{One branch index for each displacement index in $\mathcal{T}_{\mathcal{F},x}$}
    \SetKw{Output}{output}
    \SetKw{Continue}{continue}
    \If{$\mathsf{Dec}_{f}(x)=\textsc{FALSE}$}{
        \Return\tcp*{no solutions, output nothing}
    }
    $\mathbf{u} \gets (0,\ldots,0)$ \tcp*{root displacement index}
    \Output{$\beta_{\mathcal{F},x}(\mathbf u)$}\tcp*{output first solution}\label{line:Output1}
    $i \gets 0$\tcp*{boolean to remember if ascending or descending}
    \While{\textsc{True}}{

    \If{$i=0$\tcp*{try to descend}}{
        $\mathbf{v} \gets \mathsf{Child}_{\mathcal{F},x}(\mathbf{u},1)$\;
            \If{$\mathbf{v}\neq \bot$}{
                $\mathbf{u}\gets \mathbf{v}$\;
                \Output{$\beta_{\mathcal{F},x}(\mathbf u)$}\;\label{line:Output2}
                \Continue\tcp*{restart main loop}
            }
            $i \gets 1$\;
        }

        $\mathbf{v}\gets \mathsf{Sibling}_{\mathcal{F},x}(\mathbf{u})$\tcp*{try to find a sibling}
        \If{$\mathbf{v}\neq \bot$}{
            $\mathbf{u} \gets \mathbf{v}$\;
            $i \gets 0$\;
            \Output{$\beta_{\mathcal{F},x}(\mathbf u)$}\;\label{line:Output3}
            \Continue
        }
        $\mathbf{u}\gets \mathsf{Parent}_{\mathcal{F},x}(\mathbf{u})$\tcp*{no sibling, go back to parent}
        \If{$\mathbf{u}=\bot$}{\Return}
    }
\end{algorithm}

\begin{lemma}\label{lem:polyDelayProof}
     \textsc{DFSEnum}$_{\mathcal F,x}$ is a polynomial-delay algorithm.
\end{lemma}
\begin{proof}
    This follows from every oracle call being polynomial time computable (\Cref{lem:polyChild,lem:polyParent,lem:polySibling}) combined with the fact that the tree has polynomial depth (\Cref{lem:depthT}).
Each iteration of the main while loop either outputs a new index or backtracks to a parent node, which can only happen $\operatorname{depth}\!\bigl(\mathcal{T}_{\mathcal{F},x}\bigr)$ times at most.
\end{proof}

We will now show that the number of indices output by \Cref{alg:dfs-enum} is $f(x)$ and that it never outputs the same index twice.
In order to do so we define the set of complete branch indices for which $t$ evaluates to $1$ on the corresponding instance.
Note that for the $t$ to evaluate to $1$ the branch index has to be complete.
We will sometimes refer to this set as the set of \emph{positive leaves}.
Formally we define

\begin{equation}\label{eq:setofPositiveBranches}
\mathcal{H}_{\mathcal{F},x} \coloneqq
\left\{
  \mathbf{v} \in \mathbb{N}_+^{\leq q(x)} \;\middle|\;
  t\bigl(h(x,\mathbf{v}) \bigr)=1
\right\}
\end{equation}

along with the lexicographically ordered version of the same set

\begin{equation}\label{eq:orderedSetofPositiveBranches}
\mathcal{H}_{\mathcal{F},x}^{\mathrm{lex}} \coloneqq \left( \mathbf{v}_1, \dots, \mathbf{v}_k \right)
\quad \text{where }
\left\{
\begin{array}{l}
\{ \mathbf{v}_1, \dots, \mathbf{v}_k \} = \mathcal{H}_{\mathcal{F},x} \\
\mathbf{v}_1 <_{\mathrm{lex}} \dots <_{\mathrm{lex}} \mathbf{v}_k
\end{array}
\right.
\end{equation}

\begin{lemma}\label{lem:completeIndices}
    $|\mathcal{H}_{\mathcal{F},x}|=f(x)$
\end{lemma}
\begin{proof}
    We prove this via induction on the partial order over $\Sigma^{\ast}$ defined by $h$.

        \noindent\textbf{Base cases.} When $r(x)=0$, then \Cref{cond:noBranch} on uSR function definitions implies that $\forall i\in \mathbb{N}, \; h(x,i)=\lambda$.
        By the invalid-branch convention in \Cref{cond:noBranch}, any non-empty continuation from \(\lambda\) remains at \(\lambda\) and has leaf value \(0\).
        Therefore the only input for which $t(h(x,\mathbf{v}))$ will be non-zero when $r(x)=0$ is the empty sequence, i.e. $t(h(x,\emptyset))=t(x)$.
        If $t(x)=1$, then $f(x)=1$ and $\mathcal{H}_{\mathcal{F},x}=\{\emptyset\}$, if $t(x)=0$ then $f(x)=0$ and $\mathcal{H}_{\mathcal{F},x}=\emptyset$.
        In both cases $|\mathcal{H}_{\mathcal{F},x}|=f(x)$.

        \noindent\textbf{Inductive argument.} We assume the statement holds for all $h(x,i)$ with $i\in [r(x)]$.
        Since \(r(x)>0\), condition~\ref{cond:unweighted} gives \(t(x)=0\), so the empty tuple is not in \(\mathcal H_{\mathcal F,x}\).
        We have
        \begin{align*}
            |\mathcal{H}_{\mathcal{F},x}|&=\bigl|\bigl\{\mathbf{v}\mid t(h(x,\mathbf{v}))=1\bigr\}\bigr|\\
            &=\Bigl|\bigcup_{i=1}^{r(x)}\bigl\{i.\mathbf{v}'\mid t(h(h(x,i), \mathbf{v}'))=1\bigr\}\Bigr|\\
            &=\sum_{i=1}^{r(x)}\Bigl|\bigl\{i.\mathbf{v}'\mid t(h(h(x,i),\mathbf{v}'))=1\bigr\}\Bigr| && \text{($i\neq j\Rightarrow i.\mathbf{v}'\neq j.\mathbf{v'}$)}\\
            &=\sum_{i=1}^{r(x)}f(h(x,i)) && \text{(IH)}\\
            &=f(x)&& \text{(\ref{cond:USR}, \ref{cond:unweighted} in \Cref{def:uredp})}
        \end{align*}

\end{proof}

\begin{lemma}\label{lem:monotonicityB}
    Given displacement indices $\mathbf{u}$ and $\mathbf{u}'$, if $\beta_{\mathcal{F},x}(\mathbf{u})$ and $\beta_{\mathcal{F},x}(\mathbf{u'})$ are valid branch indices, then $\mathbf{u}\leq_{\mathrm{lex}}\mathbf{u}'\Rightarrow\beta_{\mathcal{F},x}(\mathbf{u})\leq_{\mathrm{lex}}\beta_{\mathcal{F},x}(\mathbf{u}')$.
\end{lemma}
\begin{proof}
Let $i$ be the smallest position where $\mathbf u$ and $\mathbf u'$ differ:
\[
  i \;=\; \min\{j \mid \mathbf u_j \neq \mathbf u'_j\}.
\]
Because $\mathbf u<_{\mathrm{lex}}\mathbf u'$, we have $\mathbf u_i < \mathbf u'_i$ and
$\mathbf u_j = \mathbf u'_j$ for every $j<i$.

\Cref{alg:disToBranch} processes the displacement vector from left to right.
For every depth $d<i$, the vectors $\mathbf u$ and $\mathbf u'$ induce
identical control flow and produce the same (partial) branch index,
because their first $i-1$ entries coincide.

At depth $i$ the branch-selection step in \Cref{alg:disToBranch}
appends the $\mathbf{u_i}+1$-th admissible child for input $\mathbf u$
and the $\mathbf{u'_i+1}$-th admissible child for input $\mathbf u'$.
Since $\mathbf u_i<\mathbf u'_i$ and the list $\mathcal C$ of admissible children
is ordered increasingly, the child chosen when processing $\mathbf u$
appears strictly before the child chosen when processing $\mathbf u'$.
Consequently the branch index produced from $\mathbf u$
is lexicographically smaller than that produced from $\mathbf u'$,
regardless of the remaining suffix of either vector.

Formally,
\[
  \beta_{\mathcal F,x}(\mathbf u)
  \;=\;
  (v_1,\dots,v_{i-1},c_{\mathbf u_i},\dots)
  \;<_{\mathrm{lex}}\;
  (v_1,\dots,v_{i-1},c_{\mathbf u'_i},\dots)
  \;=\;
  \beta_{\mathcal F,x}(\mathbf u'),
\]
where $c_j$ denotes the $j$-th element of the admissible-children list
$\mathcal C$ at depth $i$.

\end{proof}

\begin{lemma}
    The outputs of \textsc{DFSEnum} belong to $\mathcal{H}_{\mathcal{F},x}$.
\end{lemma}

\begin{proof}
Each output of \textsc{DFSEnum} is of the form $\beta_{\mathcal{F},x}(\mathbf{u})$, where $\mathbf{u}$ is a displacement index in the tree $\mathcal{T}_{\mathcal{F},x}$.

For every displacement index \(\mathbf u\) actually output by
\textsc{DFSEnum}, validity is guaranteed either by the initial feasibility test
at the root or by the validity test inside \(\mathsf{Child}\).
The suffix check in \(\beta_{\mathcal F,x}\) ensures that coordinates not
consumed by the root-to-leaf walk are zero, so two different accepted
displacement indices cannot encode the same completed branch.
Thus the returned branch index \(\mathbf v=\beta_{\mathcal F,x}(\mathbf u)\)
satisfies \(h(x,\mathbf v)\neq\lambda\) and
\(r(h(x,\mathbf v))=0\).

Moreover, \textsc{DFSEnum} outputs only displacement indices whose corresponding branch index reaches a positive leaf:
$t(h(x,\beta_{\mathcal{F},x}(\mathbf{u}))) = 1$ for every emitted \(\mathbf u\).
Hence every output $\mathbf{v}$ satisfies:
\[
  h(x,\mathbf{v}) \neq \lambda, \quad r(h(x,\mathbf{v})) = 0, \quad \text{and} \quad t(h(x,\mathbf{v})) = 1.
\]

This precisely characterises membership in $\mathcal{H}_{\mathcal{F},x}$ as defined in \Cref{eq:setofPositiveBranches}.
\end{proof}

\begin{lemma}\label{lem:lex-strict}
The output sequence of \textsc{DFSEnum} is strictly increasing in lexicographic order.
\end{lemma}

\begin{proof}
Let $\mathbf v^{(k)}=\beta_{\mathcal F,x}(\mathbf u^{(k)})$ denote the
$k$-th output and its displacement index.
Immediately after outputting $\mathbf u^{(k)}$, the DFS performs exactly one
of the following moves:

\begin{enumerate}[(i)]
  \item \emph{Child step:} increment a coordinate of $\mathbf u^{(k)}$ (non-strictly) to the right of the rightmost non-zero coordinate;
  \item \emph{Sibling step:} delete the right-most non-zero coordinate of
        $\mathbf u^{(k)}$ but increment a coordinate to its left;
  \item \emph{Back-track step:} Does not output, deletes rightmost coordinate but guarantees that next output if it exists is from a sibling call that increments coordinate to the left of the rightmost coordinate.
\end{enumerate}

The first two moves produce a displacement index that is lexicographically larger than
$\mathbf u^{(k)}$ and the third guarantees that after a sequence of backtracking the next move produces a displacement index that is lexicographically larger than
$\mathbf u^{(k)}$
This implies that
$\mathbf u^{(k)}<_{\mathrm{lex}}\mathbf u^{(k+1)}$.
\Cref{lem:monotonicityB} transfers this strict increase from displacement
indices to the branch indices themselves, yielding
$\mathbf v^{(k)}<_{\mathrm{lex}}\mathbf v^{(k+1)}$.
\end{proof}

\begin{lemma}\label{lem:noskip}
    All elements of $\mathcal{H}_{\mathcal{F},x}$ are output by \textsc{DFSEnum}$_{\mathcal F,x}$.
\end{lemma}
\begin{proof}
Assume some $\mathbf v^{\ast}\in\mathcal H_{\mathcal F,x}$ is never produced.
Pick $\mathbf v^{\ast}$ to be the \emph{lexicographically smallest} such
missing element and let
\(
  \mathbf v
\)
be the last element that \emph{is} output with
$\mathbf v<_{\mathrm{lex}}\mathbf v^{\ast}$
(the strict lexicographic order of previous outputs ensures that \(\mathbf v\) exists and is unique).

Right after printing $\mathbf v$, the algorithm examines the displacement
index $\mathbf u$ that maps to $\mathbf v$ and performs exactly one of the
following moves:

\begin{enumerate}[(i)]
  \item \emph{Descend.}\;
        If $\mathbf v^{\ast}$ has $\mathbf v$ as a strict prefix, then the
        subtree rooted at the first child of $\mathbf v$ contains
        $\mathbf v^{\ast}$ and, by minimality, no smaller positive leaf.
        The first branch index discovered there is therefore
        $\mathbf v^{\ast}$.

  \item \emph{Sibling.}\;
        If $\mathbf v^{\ast}$ shares the parent of $\mathbf v$ but lies to
        its right, the sibling routine increments the last coordinate of
        $\mathbf u$.  The first sibling whose branch index is positive is
        exactly $\mathbf v^{\ast}$.

  \item \emph{Back-track.}\;
        Otherwise $\mathbf v^{\ast}$ is outside the current sibling block.
        Back-tracking deletes a suffix of zeros from $\mathbf u$ and then
        increments the first remaining coordinate, yielding the global
        lexicographic successor of $\mathbf v$, which must be
        $\mathbf v^{\ast}$.
\end{enumerate}

In every exhaustive case the algorithm's next output would be
$\mathbf v^{\ast}$, contradicting the assumption that it is skipped.
Hence no element of $\mathcal H_{\mathcal F,x}$ is omitted.
\end{proof}

\begin{lemma}\label{lem:treeAlgoCorr}
\textsc{DFSEnum} outputs the elements of \(\mathcal{H}_{\mathcal{F},x}\) one by one in lexicographic order; that is, its output sequence is exactly \(\mathcal{H}_{\mathcal{F},x}^{\mathrm{lex}}\).
\end{lemma}
\begin{proof}
    This is a straightforward consequence of \Cref{lem:lex-strict,lem:noskip}.
\end{proof}
\begin{theorem}[Correctness, completeness, and complexity]\label{thm:dfs-correct}
For every input \(x \in \Sigma^{\ast}\) with uSR definition \(\mathcal{F}\),

\begin{enumerate}
  \item \textsc{DFSEnum} outputs exactly \(f(x)\) times;
  \item its \(i\)-th output is the \(i\)-th element of
        \(\mathcal{H}_{\mathcal{F},x}^{\mathrm{lex}}\), i.e.\
        \[
           \bigl(\text{output sequence of \textsc{DFSEnum}}\bigr)
           \;=\;
           \mathcal{H}_{\mathcal{F},x}^{\mathrm{lex}};
        \]
  \item hence every complete branch index \(\mathbf{v}\) with
        \(t\!\bigl(h(x,\mathbf{v})\bigr)=1\) appears once and only once, and the
        outputs are strictly increasing in lexicographic order.
    \item the time between two successive outputs is bounded by
      \(\operatorname{poly}(|x|)\).
      The total runtime of the algorithm therefore is \(\mathcal{O}(\operatorname{poly}(|x|) \cdot |\mathcal{H}_{\mathcal{F},x}|)\) or equivalently  \(\mathcal{O}(\operatorname{poly}(|x|) \cdot f(x)) \).

\end{enumerate}
\end{theorem}

\begin{proof}
Items (1)--(4) follow from
\Cref{lem:completeIndices,lem:lex-strict,lem:noskip,lem:treeAlgoCorr,lem:polyDelayProof}.

\end{proof}

\subsection{The Sample Oracle}\label{sec:Sampling}

Stage~2 estimates $f(x)$ via Bernoulli trials.
We build a sampler $\sigma_{\mathcal F^b,x}$ that maps a uniformly random
$m\in[B]$ to a leaf so that
$\Pr[t(h(x,\sigma_{\mathcal F^b,x}(m)))=1]=f(x)/B$.
Conceptually, $\sigma$ implements the standard ``unranking'' procedure using the bound mass $b$:
it treats each subtree $h(x,i)$ as owning a block of size $b(h(x,i))$ and descends into the block
containing $m$.
You can think of \(b(x)\) as assigning \(b(\cdot)\) many abstract ``tickets'' to each subtree,
so that sampling \(m\in[b(x)]\) uniformly corresponds to choosing a ticket uniformly.
The oracle \(\sigma\) implements unranking with respect to these ticket blocks.
A draw succeeds exactly when the reached leaf is positive, so the success probability is
the fraction of tickets associated with positive leaves.

\begin{algorithm}[H]
    \DontPrintSemicolon
    \caption{$\sigma_{\mathcal{F}^b,x}(m)$ - Sample oracle}
    \label{alg:sampleOracle}
    \KwIn{$m\in [b(x)]$}
    \KwOut{Branch index $\mathbf{v}$}
    \SetKw{Return}{return}
    $\mathbf{v} \gets \emptyset$\;
    $m'\gets m$\;
    \While{$r(h(x,\mathbf{v}))>0$}{\label{line:sigmaOuterWhile}
        $i \gets 1$\;
        \While{$i \le r(h(x,\mathbf{v})) \textbf{ and } m'>b(h(x,\mathbf{v}.i))$}{\label{line:sigmainnerwhileloop}
            $m' \gets m'-b(h(x,\mathbf{v}.i))$\;
            $i\gets i+1$\;
            \If{$i=r(h(x,\mathbf{v}))+1$}{
                \Return $\mathbf{v}.(r(h(x,\mathbf{v}))+1)$
            }
        }
        $\mathbf{v} \gets \mathbf{v}.i$\;\label{line:appendv}
    }
    \Return $\mathbf{v}$
\end{algorithm}

\begin{lemma}\label{lem:sampleRuntime}
    $\sigma_{\mathcal{F}^b,x}$ runs in polynomial time.
\end{lemma}
\begin{proof}
    The inner while loop (\Cref{line:sigmainnerwhileloop}) increments $i$ by $1$ and terminates if $i$ becomes larger than $r(h(x,\mathbf{v}))\in \operatorname{poly}(|x|)$ (\Cref{def:qrfunc} and \Cref{cond:InstanceSize}).
    The outer while loop appends an index to $\mathbf{v}$ and terminates when $r(h(x,\mathbf{v}))=0$, which must be the case when $|\mathbf v|>q(x)$ (\Cref{cond:polyDepth}).
\end{proof}

\begin{remark}[Slack tickets and the dummy child]
Recursion compatibility allows slack:
$s(y) := b(y) - \sum_{i=1}^{r(y)} b(h(y,i)) \ge 0$.
\Cref{alg:sampleOracle} implicitly treats this slack as a block of ``dummy tickets'': whenever the running
index $m'$ exceeds the cumulative child masses, the sampler returns the dummy child
$r(y)+1$, which (by convention) routes to $\lambda$ and hence contributes a failure
($t(\lambda)=0$).
This makes the success probability exactly $f(x)/b(x)$ without requiring equality in the
bound recurrence.
\end{remark}

\begin{lemma}[Ticket count]\label{lem:sigma-ticket-count}
For every instance \(y\), define
\[
  A(y):=
  \{\,m\in \mathbb N:1\le m\le b(y)
      \text{ and }
      t(h(y,\sigma_{\mathcal F^b,y}(m)))=1\,\}.
\]
Then \(|A(y)|=f(y)\).
\end{lemma}
\begin{proof}
For a node \(y\) and \(1\le i\le r(y)\), put
\[
  L_i(y):=\sum_{j<i} b(h(y,j)),
  \qquad
  I_i(y):=\{L_i(y)+1,\ldots,L_i(y)+b(h(y,i))\}.
\]
The remaining tickets
\[
  I_\bot(y):=
  \left\{\sum_{j=1}^{r(y)} b(h(y,j))+1,\ldots,b(y)\right\}
\]
form the slack block, which may be empty.  By the inner loop of
\Cref{alg:sampleOracle}, a ticket in \(I_i(y)\) descends to child \(i\), with
local ticket \(m-L_i(y)\); a ticket in \(I_\bot(y)\) returns the dummy child
\(r(y)+1\), which reaches \(\lambda\) and contributes \(0\).

We prove the claim by induction on the maximum remaining length of a valid
recursive path below \(y\). If \(r(y)=0\), then the sampler immediately
returns the empty branch. Since \(b(y)\le 1\), there are two leaf subcases.
If \(b(y)=0\), then \(A(y)=\varnothing\) and \(f(y)=0\). If \(b(y)=1\), then
there is exactly one ticket, and it is successful exactly when \(t(y)=1\).
If \(t(y)=0\), both sides are \(0\); if \(t(y)=1\), then
\(|A(y)|=t(y)=f(y)\).

Now suppose \(r(y)>0\).  The child blocks \(I_1(y),\ldots,I_{r(y)}(y)\) are
disjoint, and slack tickets do not succeed.  Therefore
\[
  |A(y)|
  =
  \sum_{i=1}^{r(y)} |A(h(y,i))|.
\]
By the induction hypothesis this equals
\[
  \sum_{i=1}^{r(y)} f(h(y,i)).
\]
Since \(y\) is an internal node of a uSR definition, \(t(y)=0\), so the uSR
recurrence gives
\[
  f(y)=t(y)+\sum_{i=1}^{r(y)} f(h(y,i))
      =\sum_{i=1}^{r(y)} f(h(y,i)).
\]
Thus \(|A(y)|=f(y)\), completing the induction.
\end{proof}

\begin{theorem}[Correctness of the sample oracle]\label{thm:sampling-correct}
Assume \(b(x)>0\).  Let \(m\) be drawn uniformly at random from \([b(x)]\) and let
\(\mathbf v=\sigma_{\mathcal F^{b},x}(m)\).
Then
\[
  \Pr\bigl[t\bigl(h(x,\mathbf v)\bigr)=1\bigr]
  \;=\;
  \frac{f(x)}{b(x)}.
\]
\end{theorem}

\begin{proof}
By \Cref{lem:sigma-ticket-count}, exactly \(f(x)\) of the \(b(x)\) possible
tickets reach a positive leaf.  A uniformly random ticket therefore succeeds
with probability \(f(x)/b(x)\).
\end{proof}

\subsection{A General $(\varepsilon,\delta)$-Approximation in $\sqrt{b}$ Time}\label{sec:GeneralAlgo}

We can finally describe our full algorithm:

\paragraph{Pre-processing.}
Fix an input instance $x\in\Sigma^{\ast}$ and let
$B:=b(x)$ be the polynomially computable upper bound supplied by the
uSR bounded definition $\mathcal{F}^b$.
Write $f:=f(x)$ for the (unknown) solution count.
The algorithm receives parameters
$\varepsilon\in(0,1)$ and $\delta\in(0,1)$ and must output an
$(\varepsilon,\delta)$-approximation $\hat f$ to~$f$.

\begin{algorithm}[]
\DontPrintSemicolon
\caption{\textsc{Count}$_{\mathcal{F}^b,x}(\varepsilon,\delta)$}
\label{alg:FullAlgo}
\KwIn{instance $x$;\ accuracy $\varepsilon$;\ failure prob.\ $\delta$}
\KwOut{$\hat f$ with $\Pr[(1-\varepsilon)f\le\hat f\le(1+\varepsilon)f]\ge 1-\delta$}

\BlankLine
$B\leftarrow b(x)$\;
$k\leftarrow\lceil\sqrt{B}\rceil$\;

\If{$k=0$}{
\Return $\hat f\leftarrow 0$}

\BlankLine
\textbf{/* Stage 1: enumerate up to $k$ solutions */}\;
$c\leftarrow 0$\;
\For{$\mathbf v$ \emph{produced by} \textsc{DFSEnum}$_{\mathcal{F},x}$\tcp*{polynomial-delay enumeration}}{
  $c\leftarrow c+1$\;
  \lIf{$c=k$}{\textbf{break}}
}
\BlankLine
\If{$c<k$}{
\Return $\hat f\leftarrow c$ \tcp*[r]{\emph{exact}}}

\BlankLine
\textbf{/* Stage 2: sampling when $f\ge k$ */}\;

$T \leftarrow \left\lceil \frac{3\sqrt{B}}{\varepsilon^{2}} \ln\!\left(\frac{2}{\delta}\right)\right\rceil
$

$S\leftarrow 0$\;
\For{$i\leftarrow 1$ \KwTo $T$}{
   pick $m$ uniformly from $[B]$\;
   $\mathbf v\;\leftarrow\;\sigma_{\mathcal{F}^b,x}(m)$\;
   \lIf{$t\!\bigl(h(x,\mathbf v)\bigr)=1$}{$S\leftarrow S+1$}
}
$\hat p\leftarrow S/T$\;
\Return $\hat f\leftarrow \hat p\cdot B$
\end{algorithm}

\begin{theorem}\label{thm:fullalgo}
For every $x\in\Sigma^{\ast}$, $\varepsilon\in(0,1)$ and
$\delta\in(0,1)$, \textsc{Count}$_{\mathcal F^b,x}(\varepsilon,\delta)$ returns a value
$\hat f$ satisfying
\[
\Pr\!\bigl[(1-\varepsilon)f(x)\le\hat f\le(1+\varepsilon)f(x)\bigr]
\;\ge\;1-\delta .
\]
Its running time is
\(
\operatorname{poly}(|x|)
      \cdot\bigl(
          \sqrt{b(x)} + \tfrac{\sqrt{b(x)}}{\varepsilon^{2}}
          \log\tfrac{1}{\delta}
        \bigr)
      =O^\ast\bigl(\tfrac{\sqrt{b(x)}}{\varepsilon^{2}}\log\tfrac{1}{\delta}\bigr).
\)
\end{theorem}

\begin{proof}
\emph{Correctness.}
If \(k=0\), then \(B=0\), and the bound \(f(x)\le b(x)\) implies \(f(x)=0\);
the algorithm returns the exact value.  Assume now that \(k>0\).
If enumeration halts after emitting \(c<k\) leaves, then \(c=f\) and the
algorithm outputs the exact count.

Otherwise $f\ge k\ge\sqrt B$, so the success
probability of a single call to the sampling oracle is, by
\Cref{thm:sampling-correct},
$p=f/B\ge 1/\sqrt B$.
The \(T\) calls to the sample oracle use independent uniformly random tickets
from \([B]\), so their success indicators are independent Bernoulli variables
with mean \(p\).
With $T$ samples the multiplicative Chernoff bound gives
\[
  \Pr\bigl[\;|p-\hat p|\;>\;\varepsilon p\bigr]
  \;\le\;2\exp\!\bigl(-\tfrac{\varepsilon^{2}pT}{3}\bigr)
  \;\le\;2\exp\!\bigl(-\tfrac{\varepsilon^{2}T}{3\sqrt B}\bigr)
  \;\le\;\delta .
\]
Because $\hat f=\hat p\,B$ and $f=pB$, the same inequality implies
$(1-\varepsilon)f\le\hat f\le(1+\varepsilon)f$.

\emph{Running time.}
Stage 1 outputs at most $k=\lceil\sqrt B\rceil$ leaves, each after
$\operatorname{poly}(|x|)$ delay by \Cref{lem:polyDelayProof},
hence costs $\operatorname{poly}(|x|)\sqrt B$.

Stage 2 executes $T$ oracle calls, each again polynomial in $|x|$.
Substituting the chosen~$T$ yields the stated bound.
\end{proof}

\subsection{Extending to TotP}\label{sec:toTotP}

The framework is most directly applicable for problems where the distinction between enumeration and counting is minimal.
Despite this, the idea is nonetheless applicable to a broader set of TotP problems as the bottleneck for fast exponential algorithms is often a core where the counting appears to not be distinguishable from the enumeration.
Although every problem in TotP can be expressed in uTotP form as we proved in \Cref{sec:NewChar}, doing so generically typically would not result in effective bounding functions.
For many problems, counting is faster than enumeration but they nonetheless may be bottlenecked by hard enumeration cores, where the problem is too constrained to exploit the symmetries useful in counting, but too little to branch effectively.
This is where our framework can be leveraged.
In order to do so we extend our theory so that we may derive a generic algorithm for sums of functions given in bounded uSR form.

\begin{theorem}\label{thm:combineuTotP}
  Let
  \((\mathcal{F}_1,x_1),(\mathcal{F}_2,x_2),\ldots,(\mathcal{F}_\ell,x_\ell)\)
  be\/ $\ell$ pairs of bounded uSR forms and instances, each
  defining a counting function \(f_i\) with bounding function \(b_i\) and
  such that \(f_i(x_i)>0\) for every \(i\).
  Suppose that a preprocessing step constructs this explicit list, computes the
  values \(B_i=b_i(x_i)\), and builds the prefix sums
  \(P_i=\sum_{j\le i}B_j\) in time \(S\), with
  \(\ell\le 2^{\operatorname{poly}(n)}\) and \(\max_i |x_i|\le \operatorname{poly}(n)\).
  Applications that naturally generate zero-valued cores should discard them
  before invoking the theorem and charge this filtering to \(S\).
  Then one can compute an $(\varepsilon,\delta)$-approximation to
  \(\displaystyle F:=\sum_{i=1}^{\ell} f_i(x_i)\) in time
  \[
    S+
    O^{\ast}\!\bigl(
      \sqrt{\textstyle\sum_{i=1}^{\ell} b_{i}(x_{i})}\;
      \varepsilon^{-2}\,
      \log\!\tfrac{1}{\delta}
    \bigr).
  \]
\end{theorem}

The estimator used in the proof is given in \Cref{alg:sum-count}.  It is the
aggregate analogue of \Cref{alg:FullAlgo}: Stage~1 enumerates positive leaves
across the listed cores until the global cut-off is met, while Stage~2 samples
from the disjoint union of the intervals of sizes \(B_i=b_i(x_i)\).  In the
displayed pseudocode we recompute the arrays \((B_i)\) and \((P_i)\) for
self-containment; in the theorem statement these values are prepared during
the setup phase and their construction is charged to \(S\).

\begin{algorithm}[H]
\DontPrintSemicolon
\caption{\textsc{SumCount} for an explicit list of bounded-uSR cores}
\label{alg:sum-count}
\KwIn{bounded uSR form--instance pairs \((\mathcal{F}_i,x_i)\);
      accuracy \(\varepsilon\in(0,1)\); failure probability \(\delta\in(0,1)\)}
\KwOut{$\hat F$ s.t.\ $\Pr[(1-\varepsilon)F\le\hat F\le(1+\varepsilon)F]\ge1-\delta$}

\vspace{.2em}
\textbf{/* Global parameters */}\\
\For{$i\leftarrow1$ \KwTo \(\ell\)}{ $B_i\leftarrow b_i(x_i)$ }
$B\leftarrow\sum_{i=1}^{\ell} B_i$\;
\If{$B=0$}{\Return $\hat F\leftarrow 0$}
$k\leftarrow\lceil\sqrt{B}\rceil$\tcp*{enumeration cut-off}

\vspace{.2em}
\textbf{/* Stage 1 - enumerate at most $k$ positive leaves */}\\
$C\leftarrow0$\tcp*{solutions counted so far}
\(i\leftarrow1\)\;
\While{$i\le \ell$ and $C<k$}{
  initialise the polynomial-delay enumerator \textsc{DFSEnum}$_{\mathcal F_i,x_i}$\;
  \While{$C<k$ and the enumerator for instance \(i\) has another output \(\mathbf v\)}{
     $C\leftarrow C+1$\;
  }
  \(i\leftarrow i+1\)\;
}
\If{$C<k$}{\Return $\hat F\leftarrow C$ \tcp*[r]{exact count}}

\vspace{.2em}
\textbf{/* Stage 2 - combined sampler */}\\
$P_0\leftarrow0$\;
\For{$i\leftarrow1$ \KwTo \(\ell\)}{ $P_i\leftarrow P_{i-1}+B_i$ }
$T\leftarrow\bigl\lceil 3\sqrt{B}\,\varepsilon^{-2}\ln(2/\delta)\bigr\rceil$\;
$Y\leftarrow0$\;

\For{$s\leftarrow1$ \KwTo $T$}{
  pick $m$ uniformly at random from $[B]$\;
  \textbf{find} $j$ with $P_{j-1}<m\le P_j$ (binary search)\;
  $m'\leftarrow m-P_{j-1}$\tcp*{\(m'\in[B_j]\)}
  $\mathbf v\leftarrow\sigma_{\mathcal F_j^{b_j},x_j}(m')$\;
  \lIf{$t_j\!\bigl(h_j(x_j,\mathbf v)\bigr)=1$}{$Y\leftarrow Y+1$}
}
$p\leftarrow Y/T$\;
\Return $\hat F\leftarrow p\cdot B$
\end{algorithm}

\begin{proof}
The setup term \(S\) accounts for constructing the list of cores, evaluating
the \(B_i\), and preparing the prefix sums used for block selection.  We now
condition on this prepared data.  Fix the notation
\(f_i:=f_i(x_i),\;
  B_i:=b_i(x_i),\;
  F:=\sum_i f_i,\;
  B:=\sum_i B_i\).
If \(B=0\), then every \(f_i\le B_i\) is zero and the guard in
\Cref{alg:sum-count} returns the exact value.  We therefore assume \(B>0\).
The positivity assumption gives \(B_i\ge f_i>0\) for every listed core.

\paragraph{Stage 1.}
The sequential loop visits the instances in order while \(C<k\).
Whenever \textsc{DFSEnum}\(_{\mathcal F_i,x_i}\) outputs a positive
leaf we increment \(C\).
Stage 1 stops as soon as \(C\) reaches
\(k=\lceil\sqrt{B}\rceil\), so it performs at most \(k\) successful outputs.
If an enumerator is exhausted earlier, we simply move to the next core.
Because every listed core is positive, each exhausted enumerator contributes at
least one successful output before termination. Hence the total Stage 1 work is
at most \(k\) output-producing enumerator traversals, each with polynomial
delay by \Cref{lem:polyDelayProof}, together with polynomial initialisation of
the touched enumerators. This costs
\(O\!\bigl(\poly(\max_i|x_i|)\sqrt{B}\bigr)\) time.

If the loop exhausts \emph{all} enumerators before reaching \(k\)
leaves, then \(F=C<k\) and the algorithm returns
\(\hat F=F\) exactly, meeting the \((\varepsilon,\delta)\) requirement
trivially.

Henceforth assume \(F\ge k\ge\sqrt{B}\).

\paragraph{Stage 2: success probability of a single draw.}
Partition \([B]\) into disjoint blocks
\([P_{i-1}+1,P_i]\) of size \(B_i\)
with prefix sums \(P_i:=\sum_{t=1}^{i}B_t\) (\(P_0:=0\)).
When a random \(m\in[B]\) falls into block \(i\),
 the probability of choosing that block is \(B_i/B\);
  conditional on being in block \(i\), the local sampler
  \(\sigma_{\mathcal F_i^{b_i},x_i}\) succeeds with probability
  \(f_i/B_i\).

By the law of total probability

\[
  p:=\Pr[\text{one draw succeeds}]
     =\sum_{i=1}^{\ell}\frac{B_i}{B}\cdot\frac{f_i}{B_i}
     =\frac{F}{B}
     \;\ge\;\frac1{\sqrt{B}}.
\]

\paragraph{Number of samples.}
Let \(T:=\bigl\lceil 3\sqrt{B}\,\varepsilon^{-2}\ln(2/\delta)\bigr\rceil\)
and \(\hat p:=Y/T\).
A multiplicative Chernoff bound gives

\[
  \Pr[\,|\hat p-p|>\varepsilon p\,]\le
  2\exp\!\bigl(-\tfrac{\varepsilon^{2}pT}{3}\bigr)
  \le 2\exp\!\bigl(-\tfrac{\varepsilon^{2}T}{3\sqrt{B}}\bigr)
  \le\delta.
\]

Because \(\hat F=\hat p\,B\) and \(F=p\,B\),
the same inequality implies
\((1-\varepsilon)F\le\hat F\le(1+\varepsilon)F\)
with probability at least \(1-\delta\).

\paragraph{Total running time.}
Stage 2 performs \(T\) iterations. Binary search over the prepared \(P_i\) determines which \(\sigma_{\mathcal{F}_i,x_i}\) to call, and takes \(\log \ell\) time.
Each iteration therefore takes
\(\log\ell\cdot\poly(\max_i|x_i|)\) time, so

\[
  \text{time}(2)=
  O\!\Bigl(
     \log\ell\cdot\poly\bigl(\textstyle\max_i|x_i|\bigr)
     \sqrt{B}\,\varepsilon^{-2}\log\tfrac1\delta
  \Bigr)
  =O^{\ast}\!\bigl(\sqrt{B}\,\varepsilon^{-2}\log\tfrac1\delta\bigr).
\]

Adding Stage 1 and the setup term \(S\) yields the claimed overall bound
\(S+O^{\ast}(\sqrt{B}\,\varepsilon^{-2}\log(1/\delta))\).

\end{proof}

With \Cref{thm:combineuTotP}, we can adopt the following meta-strategy to find faster exponential algorithms for many problems in TotP:

\begin{enumerate}
    \item \textbf{Decompose the problem:} For a given function $f\in \text{TotP}$, decompose the problem into the easier-to-count parts and the hard enumeration cores, while recording the setup time \(S\).
    \item \textbf{Approximate the sum of hard cores:} We can use \Cref{thm:combineuTotP} to compute an $(\varepsilon, \delta)$-approximation for $\sum_{i=1}^{\ell}f_i(x_i)$ in time \(S+O^{\ast}\!\left(\sqrt{\sum_{i=1}^{\ell}b_i(x_i)}\,\varepsilon^{-2}\log(1/\delta)\right)\).
    \item \textbf{Compute the final result:} The final approximation for the original function $f$ is simply the sum of the approximated sum of cores and (approximated) offset from the easier instances.
\end{enumerate}

\begin{remark}[Compiling a decomposition into a single bounded-uSR instance]\label{rem:compile-sum}
\Cref{thm:combineuTotP} deliberately treats the family of cores
\(\{(\mathcal F_i^{b_i},x_i)\}_{i=1}^{\ell}\) as black boxes:
after the decomposition step has produced per-core bounded-uSR data,
the estimator only needs the induced primitives: per-core bounds \(B_i=b_i(x_i)\),
local samplers \(\sigma_{\mathcal F_i^{b_i},x_i}\), and (for Stage~1) feasibility access
used by polynomial-delay enumeration.
This keeps the sum-of-cores argument agnostic to how the decomposition was obtained.

If desired, one can further compile the family into a single bounded-uSR
definition by adding a top-level branching gadget that selects an index \(i\)
(e.g.\ via the binary encoding of \(i\)), routes subsequent recursion to
\((\mathcal F_i^{b_i},x_i)\), and uses the global bound \(B=\sum_i B_i\)
(implemented via prefix sums).
In that compiled view, feasibility for the top-level branches and the auxiliary
bound-navigation data structures are naturally constructed during the
decomposition/preprocessing phase.

We do not formalise this compilation here: it adds an additional layer of
encoding and oracle bookkeeping without changing the estimator or the
\(\sqrt{\sum_i B_i}\) phenomenon of \Cref{thm:combineuTotP},
and would significantly lengthen an already notation-heavy development.
\end{remark}

\section{Applications}\label{sec:Applications}

This section instantiates the bounded-uSR interface on five canonical problems.
Each application follows the same template:
(i) expose a concrete unweighted self-reduction (the maps $h,t,r,q$),
(ii) certify a recursion-compatible upper bound $b(\cdot)$ on the number of positive leaves,
and (iii) invoke the generic estimator (\Cref{thm:fullalgo}), and when relevant the
sum-of-cores extension (\Cref{thm:combineuTotP}).
The intent is not only to state improved bases, but to serve as a collection of worked patterns
for turning ``recurrence-style upper bounds'' into approximate counting algorithms.

\paragraph{What kinds of problems fit best.}
The framework is most direct for \emph{extremal subset problems}, where solutions are
inclusion-wise maximal/minimal objects over an $n$-element ground set and extremal combinatorics
already provides a nontrivial bound on how many such objects can exist.
In this regime, counting often appears \emph{intrinsically bottlenecked by output size}:
the best known counting algorithms do not asymptotically beat output-sensitive enumeration.
This is exactly the setting where a bounded uSR form is essentially ``for free'':
a classical branching analysis already gives a recurrence-compatible $b(\cdot)$, and the
enumerate--or--sample engine automatically yields a square-root improvement from $b(x)$ to
$\sqrt{b(x)}$.

\paragraph{Two levers: refactor to unweighted, or decompose into cores.}
Outside the pure ``counting $\approx$ enumeration'' regime, fast counting algorithms typically
exploit additional structure (e.g.\ multiplicative recurrences, component factorisations, or more
global combinatorial cancellations).  In such cases, forcing \emph{the entire computation} into a
single unweighted recursion usually leads to a weak global bound and hence little benefit.
Our approach suggests two practical levers:
\begin{enumerate}[(a)]
\item \textbf{Expose an unweighted core.}
Sometimes a problem admits a natural solution-enumerating search tree, but the \emph{best}
exact counting algorithms are phrased in a weighted/multiplicative way.  If one can refactor the
computation so that the remaining ambiguity is an \emph{additive leaf-counting tree} with an
explicit bound, then the generic $\sqrt{b}$ (and quantum) accelerations apply to the bottleneck
step directly.
\item \textbf{Decompose into hard enumeration cores.}
When a large portion of the instance can be handled by specialised routines (exact or approximate),
the right abstraction is to peel off the ``easy'' part and isolate only the residual ambiguity as a
collection of bounded-uSR cores.  \Cref{thm:combineuTotP} is then the key quantitative
payoff: we can estimate the \emph{total} contribution of all cores in time
$O^\ast\!\bigl(\sqrt{\sum_i b_i}\bigr)$ rather than summing per-core costs.
\end{enumerate}
A useful rule of thumb is: if the best known counting algorithm is already far faster than
enumerating all solutions, then any improvement from our framework will likely require either
(a careful unweighting/refactoring) or (an explicit decomposition that isolates the part where
counting collapses back to enumeration).

\paragraph{How the examples ramp up in ``richness.''}
As previewed in \Cref{sec:overview:apps}, we order the applications by how many of the above
ingredients must be engineered rather than taken off the shelf.
Maximal clique counting is the cleanest ``all basic ingredients'' case: a tight extremal bound,
a canonical Tsukiyama-style search organisation, and no known exponential-time counting method
that substantially outperforms output-sensitive enumeration.
Minimal separators keep the same high-level philosophy (counting not known to beat enumeration)
but the known enumeration routines are less canonical and the extremal bound is no longer tight
for our particular recursion.
Subcubic perfect matchings then move to a different regime: exact counting uses more
structured (and typically weighted/multiplicative) viewpoints, yet by refactoring to an unweighted
core we can still beat the fastest known exact running time without any decomposition.
Independent set counting and \#\textsc{2-SAT} finally illustrate the setting where decomposition is central:
existing counting algorithms are already significantly faster than naive enumeration, and the gain
comes from isolating and aggregating the remaining hard cores; for \#\textsc{2-SAT} the improvement in
the base is modest, but it is notable that the displayed approximation bound lies
below the currently cited variable-parameter exact worst-case base.

\paragraph{Other candidate problems.}
Beyond the five case studies, the same template applies widely to problems where one can
(i) maintain tractable feasibility for the subinstances generated by the chosen branching rule, and
(ii) certify a strong recurrence-style upper bound.
Natural candidates include many other locally-optimal/extremal structures in graphs and
hypergraphs, e.g.\ counting maximal independent sets / maximal matchings, minimal dominating sets (though challenging due to this being a major open problem),
minimal transversals (hitting sets) in structured set systems, and related ``extension''/local-optimal
counting problems.
Problems with sophisticated exact/approximate counters remain candidates as well, provided one can
identify the residual part where the computation degenerates to an enumeration core and then
apply \Cref{thm:combineuTotP} to exploit the $\sqrt{\sum_i b_i}$ aggregation effect.

\subsection{Counting Maximal Cliques}
\label{sec:maxclique}

Let $M(G)$ denote the number of maximal cliques of a graph $G$.
Moon and Moser proved the tight upper bound
$M(G)\le 3^{|V(G)|/3}$~\cite{moon1965cliques}.
We now show that the generic sampling framework delivers an
$(\varepsilon,\delta)$-approximation of $M(G)$ in time
$O^{\ast}\!\bigl(3^{n/6}\bigr)=O^{\ast}\!\bigl(1.2009^{\,n}\bigr)$,
shaving a square root off the enumeration bound.

Let \(\mathrm{MM}(t)\) denote the exact Moon--Moser extremal function:
\[
\mathrm{MM}(t):=
\begin{cases}
1, & t\in\{0,1\},\\
3^{t/3}, & t\equiv 0 \pmod 3,\\
4\cdot 3^{(t-4)/3}, & t\equiv 1 \pmod 3 \text{ and } t\ge 4,\\
2\cdot 3^{(t-2)/3}, & t\equiv 2 \pmod 3.
\end{cases}
\]
Thus every graph on \(t\) vertices has at most \(\mathrm{MM}(t)\) maximal
cliques, and \(\mathrm{MM}(t)\le 3^{t/3}\) for all \(t\).

\paragraph{Moon--Moser based bounded uSR form.}

\begin{definition}[$\mathcal{F}_{\mathrm{MC}}^{b}$ (bounded uSR form for counting maximal cliques)]
\label{def:mc-uSR-clean}
Fix a simple graph $G=(V,E)$ and a deterministic total order $\prec$ on $V$.
For a vertex $v\in V$, write $N_G(v)$ for its (open) neighbourhood.

An instance is a 4-tuple
\[
x=\langle R,P,X,G\rangle,
\]
where $R,P,X\subseteq V$, $P\cap X=\emptyset$, $R$ is a clique in $G$, and every
vertex in $P\cup X$ is adjacent to every vertex in $R$ (i.e.\ $P\cup X \subseteq \bigcap_{u\in R}N_G(u)$).

\begin{enumerate}[(a)]
\item \textbf{Pivot and candidate list.}
If $P\cup X\neq\emptyset$, let $H:=G[P\cup X]$ and choose
$u=u(x)\in P\cup X$ deterministically so that $\deg_H(u)$ is maximum,
breaking ties by $\prec$.
Define the candidate set
\[
C(x):=P\setminus N_H(u),
\]
and list it in increasing $\prec$-order as
\[
C(x)=(v_1,\ldots,v_k),
\qquad k:=|C(x)|.
\]
If $P\cup X=\emptyset$, set $k:=0$.

\item \textbf{Fan-out and depth.}
Set
\[
r_{\mathrm{MC}}(x):=k,
\qquad
q_{\mathrm{MC}}(x):=|V|.
\]

\item \textbf{Branch map $h_{\mathrm{MC}}$.}
For $1\le i\le r_{\mathrm{MC}}(x)$ define
\[
h_{\mathrm{MC}}(x,i)
:=
\Bigl\langle
R\cup\{v_i\},\;
\bigl(P\setminus\{v_1,\ldots,v_{i-1}\}\bigr)\cap N_G(v_i),\;
\bigl(X\cup\{v_1,\ldots,v_{i-1}\}\bigr)\cap N_G(v_i),\;
G
\Bigr\rangle.
\]
(For $i>r_{\mathrm{MC}}(x)$, we follow the global convention $h_{\mathrm{MC}}(x,i):=\lambda$.)

\item \textbf{Leaf predicate $t_{\mathrm{MC}}$.}
Set
\[
t_{\mathrm{MC}}(x):=
\begin{cases}
1, & \text{if } P=\emptyset \text{ and } X=\emptyset,\\
0, & \text{otherwise.}
\end{cases}
\]
Thus a leaf contributes $1$ exactly when $R$ is maximal (standard BK condition).

\item \textbf{Counting function.}
Define $f_{\mathrm{MC}}$ via the uSR recursion
\[
f_{\mathrm{MC}}(x)
=
t_{\mathrm{MC}}(x)+\sum_{i=1}^{r_{\mathrm{MC}}(x)} f_{\mathrm{MC}}\bigl(h_{\mathrm{MC}}(x,i)\bigr).
\]
In particular, the number of maximal cliques of $G$ is
\[
M(G)=f_{\mathrm{MC}}\bigl(\langle \emptyset, V, \emptyset, G\rangle\bigr).
\]

\item \textbf{Bounding function.}
Define (Moon--Moser style)
\[
b_{\mathrm{MC}}(\langle R,P,X,G\rangle):=
\begin{cases}
1, & \text{if } r_{\mathrm{MC}}(\langle R,P,X,G\rangle)=0,\\[2pt]
\mathrm{MM}(|P|+|X|), & \text{otherwise.}
\end{cases}
\]
\end{enumerate}
\end{definition}

\begin{lemma}\label{lem:maxcliqueEquiv}
For every Bron--Kerbosch state \(x=\langle R,P,X,G\rangle\) satisfying the
invariant in \Cref{def:mc-uSR-clean}, \(f_{\mathrm{MC}}(x)\) is the number of
maximal cliques \(K\) of \(G\) such that
\[
  R\subseteq K\subseteq R\cup P
  \qquad\text{and}\qquad
  K\cap X=\emptyset .
\]
In particular,
\[
  M(G)=f_{\mathrm{MC}}\bigl(\langle \emptyset,V,\emptyset,G\rangle\bigr).
\]
\end{lemma}
\begin{proof}
This is the standard correctness invariant for the Bron--Kerbosch recursion
with pivoting~\cite{bron1973finding}.  At a state \(\langle R,P,X,G\rangle\),
the set \(P\) contains the vertices that may still extend \(R\), while \(X\)
contains vertices whose corresponding extensions have already been assigned to
earlier branches.  The pivot changes only the organisation of the search tree;
it does not change the family of maximal cliques counted by the state.  Since
\(N_H(u)\cap P=N_G(u)\cap P\), the branch for \(v_i\in P\setminus N_H(u)\)
moves to
\[
\left\langle
R\cup\{v_i\},
(P\setminus\{v_1,\ldots,v_{i-1}\})\cap N_G(v_i),
(X\cup\{v_1,\ldots,v_{i-1}\})\cap N_G(v_i),
G
\right\rangle,
\]
which is exactly the subproblem of maximal cliques containing \(R\cup\{v_i\}\)
and not previously handled.  These branches are disjoint and cover all maximal
cliques extending \(R\) that avoid \(X\).  When \(P=X=\emptyset\), the current
clique \(R\) is maximal and contributes one; otherwise a leaf contributes zero.
The root state has \(R=X=\emptyset\) and \(P=V\), giving the stated root case.
\end{proof}

\begin{lemma}\label{lem:mm-properties}
The function \(\mathrm{MM}\) is nondecreasing, satisfies \(t\le \mathrm{MM}(t)\)
for every \(t\in\mathbb N\), and obeys
\[
  \mathrm{MM}(a)\mathrm{MM}(b)\le \mathrm{MM}(a+b)
  \qquad\text{for all }a,b\in\mathbb N.
\]
\end{lemma}
\begin{proof}
This is the exact Moon--Moser extremal function. Equivalently,
\(\mathrm{MM}(t)\) is the largest product obtainable by partitioning \(t\)
into summands \(2\) and \(3\), with the base values
\(\mathrm{MM}(0)=\mathrm{MM}(1)=1\). The explicit formula therefore makes
monotonicity and \(t\le \mathrm{MM}(t)\) immediate. Concatenating optimal
partitions for \(a\) and \(b\) gives a valid partition of \(a+b\), so
\(\mathrm{MM}(a)\mathrm{MM}(b)\le \mathrm{MM}(a+b)\).
\end{proof}

\begin{lemma}\label{lem:mc-bounding}
The map \(b_{\mathrm{MC}}\) is a bounding function for
\(\mathcal{F}_{\mathrm{MC}}\).
\end{lemma}
\begin{proof}
Let \(x=\langle R,P,X,G\rangle\) be a Bron--Kerbosch state satisfying the
invariant of \Cref{def:mc-uSR-clean}, and let \(H:=G[P\cup X]\).

We first bound the number of solutions below \(x\). By
\Cref{lem:maxcliqueEquiv}, every clique counted by \(x\) has the form
\(K=R\cup S\), where \(S\subseteq P\), \(S\cap X=\emptyset\), and \(K\) is a
maximal clique of \(G\). If some \(y\in P\cup X\) were adjacent in \(H\) to
every vertex of \(S\), then \(y\) would also be adjacent to every vertex of
\(R\) by the Bron--Kerbosch invariant, so \(K\cup\{y\}\) would be a larger
clique of \(G\), contradicting maximality. Hence \(S\) is a maximal clique of
\(H\), and therefore
\[
  f_{\mathrm{MC}}(x)\le M(H)\le \mathrm{MM}(|P|+|X|).
\]
If \(r_{\mathrm{MC}}(x)=0\), then \(b_{\mathrm{MC}}(x)=1\), so the leaf
condition and \(f_{\mathrm{MC}}(x)\le b_{\mathrm{MC}}(x)\) are immediate.

It remains to prove recursive compatibility for non-leaf states. Assume
\(r_{\mathrm{MC}}(x)>0\), let \(u\) be the pivot chosen from \(P\cup X\), let
\(\Delta:=\deg_H(u)\), and write \(C(x)=(v_1,\dots,v_k)\) for the candidate
list. Because \(u\) has maximum degree in \(H\),
\[
  k=|P\setminus N_H(u)|\le |P\cup X|-\Delta.
\]
For the \(i\)-th child
\[
  x_i=
  \Bigl\langle
  R\cup\{v_i\},\;
  \bigl(P\setminus\{v_1,\ldots,v_{i-1}\}\bigr)\cap N_G(v_i),\;
  \bigl(X\cup\{v_1,\ldots,v_{i-1}\}\bigr)\cap N_G(v_i),\;
  G
  \Bigr\rangle,
\]
every vertex in \(P_i\cup X_i\) lies in \(N_H(v_i)\). Hence
\[
  |P_i|+|X_i|\le \deg_H(v_i)\le \Delta.
\]
Therefore every child satisfies
\[
  b_{\mathrm{MC}}(x_i)\le \mathrm{MM}(\Delta).
\]
Summing over the \(k\) children and applying \Cref{lem:mm-properties} gives
\[
  \sum_{i=1}^{r_{\mathrm{MC}}(x)} b_{\mathrm{MC}}(x_i)
  \le k\,\mathrm{MM}(\Delta)
  \le \mathrm{MM}(k)\,\mathrm{MM}(\Delta)
  \le \mathrm{MM}(k+\Delta)
  \le \mathrm{MM}(|P|+|X|)
  = b_{\mathrm{MC}}(x).
\]
Thus \(b_{\mathrm{MC}}\) is recursion compatible, completing the proof.
\end{proof}

\begin{theorem}\label{thm:MC}
    We can compute a $(\varepsilon,\delta)$-approximation for $M(G)$ in time
    \[\boxed{
  O^{\ast}\!\bigl(
      3^{\,|V|/6}\;
      \varepsilon^{-2}\log\tfrac1\delta
    \bigr)
  \;=\;
  O^{\ast}\!\bigl(1.2009^{\,n}\,
      \varepsilon^{-2}\log\tfrac1\delta
    \bigr).}
\]
\end{theorem}

\begin{proof}
     By \Cref{lem:maxcliqueEquiv},
     \(M(G)=f_{\mathrm{MC}}\bigl(\langle \emptyset, V, \emptyset, G\rangle\bigr)\),
     and by \Cref{lem:mc-bounding},
     \(b_{\mathrm{MC}}\) is a valid bounding function for this recursion.
     At the root,
     \[
       b_{\mathrm{MC}}\bigl(\langle \emptyset, V, \emptyset, G\rangle\bigr)
       =\mathrm{MM}(|V|)
       \le 3^{|V|/3}.
     \]
    Therefore applying \Cref{thm:fullalgo} to
    \(\mathcal{F}_{\mathrm{MC}}^b\) gives the desired runtime.
\end{proof}

For graphs that have no anti-triangle (an independent set on $3$ vertices), a tighter upper bound $M(G)\le 2^{|V(G)|/2}$ is known~\cite{HujterT93}. As a consequence, we obtain an improved running time for anti-triangle-free graphs.

\begin{theorem}
    For anti-triangle-free graphs, we can compute a $(\varepsilon,\delta)$-approximation for $M(G)$ in time
    \[\boxed{
  O^{\ast}\!\bigl(
      2^{\,|V|/4}\;
      \varepsilon^{-2}\log\tfrac1\delta
    \bigr)
  \;=\;
  O^{\ast}\!\bigl(1.1893^{\,n}\,
      \varepsilon^{-2}\log\tfrac1\delta
    \bigr).}
\]
\end{theorem}
\begin{proof}
For anti-triangle-free graphs, every induced subgraph \(H\) also satisfies the
bound \(M(H)\le 2^{|V(H)|/2}\) of Hujter and Tuza~\cite{HujterT93}. Replacing
\(\mathrm{MM}(t)\) by \(2^{t/2}\) in the proof of \Cref{lem:mc-bounding}
therefore gives a valid state potential, because
\[
  \sum_i 2^{(|P_i|+|X_i|)/2}
  \le k\,2^{\Delta/2}
  \le 2^{k/2}2^{\Delta/2}
  =2^{(k+\Delta)/2}
  \le 2^{(|P|+|X|)/2}.
\]
At the root this yields \(b(\langle \emptyset,V,\emptyset,G\rangle)\le 2^{n/2}\),
and \Cref{thm:fullalgo} gives the stated runtime.
\end{proof}

By taking the complement graph, this result also applies to approximately computing the number of maximal independent sets in triangle-free graphs.

\subsection{Counting Minimal $(a,b)$-Separators}\label{sec:separators}

Let $G=(V,E)$ be a graph on $n:=|V|$ vertices and fix distinct terminals
$a,b\in V$.
A set $S\subseteq V\setminus\{a,b\}$ is an \emph{$(a,b)$-separator} if $a$ and $b$
lie in different connected components of $G-S$, and it is \emph{minimal} if no
proper subset of $S$ is an $(a,b)$-separator.
Write $\mathrm{MS}_{a,b}(G)$ for the number of minimal $(a,b)$-separators.

A standard characterisation is that $S$ is a minimal $(a,b)$-separator iff in $G-S$ the components $C_a,C_b$ containing $a$ and $b$ satisfy $N(C_a)=N(C_b)=S$ (equivalently, $G-S$ has two full components with neighbourhood exactly~$S$).
In particular, $\mathrm{MS}_{a,b}(G)\le \mathrm{MS}(G)$, and thus the extremal
bound of Gaspers and Mackenzie \cite{gaspers2015numberminimalseparatorsgraphs} implies $\mathrm{MS}_{a,b}(G)\le O^{*}(\tau^n)$, where $\tau=\frac{1+\sqrt 5}{2}$ is the golden ratio.

\paragraph{Enumeration vs.\ oracle access.}
There are polynomial-delay backtracking algorithms to list all minimal
$(a,b)$-separators (e.g.\cite{takata2010space}), and output-sensitive algorithms that list all
minimal separators in total time $O(n^3\cdot \mathrm{MS}(G))$ (e.g.\cite{berry2000generating}).
These are sufficient for classical output-sensitive enumeration, but they are
not phrased as local oracle access to a rooted search tree (degree, parent,
$i$-th child), which is what our generic approximate counting and quantum
speed-ups require.
We therefore refactor the backtracking itself into a bounded-uSR definition.

\paragraph{Orientation.}
To obtain a $\tau^n$ recursion bound from a single-component branching, we count
\emph{oriented} minimal separators and then sum two orientations.
Given a minimal $(a,b)$-separator $S$, let $C_a(S)$ and $C_b(S)$ be the components
of $a$ and $b$ in $G-S$.
We say $S$ is \emph{$a$-oriented} if $|C_a(S)|<|C_b(S)|$, or
$|C_a(S)|=|C_b(S)|$ and $a<b$ under a fixed deterministic total order on $V$.
Let $\mathrm{MS}^{\rightarrow}_{a,b}(G)$ be the number of $a$-oriented minimal
$(a,b)$-separators.
Then
\[
  \mathrm{MS}_{a,b}(G)
  \;=\;
  \mathrm{MS}^{\rightarrow}_{a,b}(G)
  \;+\;
  \mathrm{MS}^{\rightarrow}_{b,a}(G),
\]
and the two summands form a partition of all minimal $(a,b)$-separators.

\begin{definition}[$\mathcal F_{\mathrm{MS}}^{\rightarrow,b}$]\label{def:MinSep}
\leavevmode
Fix a deterministic total order on $V$ and write $\min(\cdot)$ for the minimum
under this order.
An instance is an encoding
\(
  x=\langle G,a,b,C,X\rangle
\)
where $a,b\in V$ are distinct, $C\subseteq V$ is connected with $a\in C$, and
$X\subseteq N_G(C)$ with $X\cap (C\cup\{b\})=\emptyset$.
Let $n:=|V|$, and define the \emph{frontier} $\partial(C,X):=N_G(C)\setminus X$ and
the \emph{measure}
\[
  \mu(x)\;:=\; n-(2|C|+|X|).
\]
When $\partial(C,X)\neq\emptyset$, define the pivot vertex
$v(x):=\min \partial(C,X)$.

\begin{enumerate}[(a)]
\item \textbf{Instance transformation $h_{\mathrm{MS}}$.}
If $r_{\mathrm{MS}}(x)=2$ (defined below), then for $i\in\{1,2\}$ set
\[
  h_{\mathrm{MS}}(x,1)\;:=\;\langle G,a,b, C\cup\{v(x)\}, X\rangle,\qquad
  h_{\mathrm{MS}}(x,2)\;:=\;\langle G,a,b, C, X\cup\{v(x)\}\rangle.
\]
If $r_{\mathrm{MS}}(x)=0$, set $h_{\mathrm{MS}}(x,i):=\lambda$ for all $i$.

\item \textbf{Leaf predicate $t_{\mathrm{MS}}$.}
If $r_{\mathrm{MS}}(x)>0$ then $t_{\mathrm{MS}}(x):=0$.
Otherwise, let $S:=X$.
Set $t_{\mathrm{MS}}(x):=1$ iff all of the following hold:
\begin{enumerate}[(i)]
\item $\partial(C,X)=\emptyset$ (equivalently, $S=N_G(C)$),
\item in $G-S$ the vertices $a$ and $b$ lie in distinct components $C_a,C_b$,
\item $N_G(C_a)=N_G(C_b)=S$ (i.e.\ $C_a,C_b$ are full components), and
\item (\emph{$a$-orientation}) $|C_a|<|C_b|$, or $|C_a|=|C_b|$ and $a<b$.
\end{enumerate}
If any condition fails, set $t_{\mathrm{MS}}(x):=0$.

\item \textbf{Fan-out $r_{\mathrm{MS}}$ and depth $q_{\mathrm{MS}}$.}
Define
\[
  r_{\mathrm{MS}}(x)\;:=\;
  \begin{cases}
    2 & \text{if } \mu(x)\ge 1,\; b\notin N_G[C],\text{ and }\partial(C,X)\neq\emptyset,\\
    0 & \text{otherwise,}
  \end{cases}
\]
and set $q_{\mathrm{MS}}(x):=n$.

\item \textbf{Bounding function $b_{\mathrm{MS}}$.}
Let $F_k$ denote the Fibonacci numbers with $F_0=0,F_1=1$.
Define
\[
  b_{\mathrm{MS}}(x)\;:=\;
  \begin{cases}
    1 & \text{if } r_{\mathrm{MS}}(x)=0,\\
    F_{\mu(x)+2} & \text{if } r_{\mathrm{MS}}(x)=2.
  \end{cases}
\]
\end{enumerate}

Finally, define the uSR function $f_{\mathrm{MS}}$ induced by
$(h_{\mathrm{MS}},t_{\mathrm{MS}},r_{\mathrm{MS}},q_{\mathrm{MS}})$, and set
\[
  \mathrm{MS}^{\rightarrow}_{a,b}(G)\;:=\;
  f_{\mathrm{MS}}\!\left(\langle G,a,b,\{a\},\emptyset\rangle\right).
\]
\end{definition}

\begin{lemma}[Correctness]\label{lem:MinSep-correct}
For every graph $G$ and terminals $a\neq b$,
$\mathrm{MS}^{\rightarrow}_{a,b}(G)$ equals the number of $a$-oriented minimal
$(a,b)$-separators of $G$.
Consequently,
$\mathrm{MS}_{a,b}(G)=\mathrm{MS}^{\rightarrow}_{a,b}(G)+\mathrm{MS}^{\rightarrow}_{b,a}(G)$.
\end{lemma}
\begin{proof}
Along a root-to-leaf path, the set \(C\) is the current component containing
\(a\), and \(X\subseteq N_G(C)\) is the set of vertices already committed to
the separator.  At each branching step the pivot
\(v\in N_G(C)\setminus X\) is assigned to exactly one of two roles: either
\(v\) joins the \(a\)-side component, giving \(C\cup\{v\}\), or \(v\) is put
into the separator candidate, giving \(X\cup\{v\}\).  Thus a leaf with
\(\partial(C,X)=\emptyset\) determines the candidate separator
\(S=X=N_G(C)\).

The leaf predicate checks precisely the standard full-component
characterisation of minimal \((a,b)\)-separators: in \(G-S\), the components
containing \(a\) and \(b\) must have neighbourhood exactly \(S\).  Conversely,
given an \(a\)-oriented minimal \((a,b)\)-separator \(S\), follow the unique
path that sends to \(C\) exactly the vertices of the \(a\)-component of
\(G-S\), and sends to \(X\) exactly the vertices of \(S\) when they first enter
the frontier.  The deterministic pivot order makes this path unique, and the
orientation test accepts it exactly in the \(a\)-oriented recursion.  The two
orientations therefore partition the minimal \((a,b)\)-separators.
\end{proof}

\begin{lemma}[Recursion-compatible bound]\label{lem:MinSep-bound}
The function $b_{\mathrm{MS}}$ is recursion-compatible for
$\mathcal F_{\mathrm{MS}}^{\rightarrow,b}$, and for the root instance
$x_0=\langle G,a,b,\{a\},\emptyset\rangle$ we have
$b_{\mathrm{MS}}(x_0)=F_n=O^{*}(\tau^n)$.
\end{lemma}

\begin{proof}
If $r_{\mathrm{MS}}(x)=2$ then the two children satisfy
$\mu(h_{\mathrm{MS}}(x,1))=\mu(x)-2$ and $\mu(h_{\mathrm{MS}}(x,2))=\mu(x)-1$.
Thus
\(
  b_{\mathrm{MS}}(h_{\mathrm{MS}}(x,1))+b_{\mathrm{MS}}(h_{\mathrm{MS}}(x,2))
  \le F_{\mu(x)}+F_{\mu(x)+1}=F_{\mu(x)+2}=b_{\mathrm{MS}}(x)
\)
by the Fibonacci recurrence.
Leaves have $b_{\mathrm{MS}}(x)=1$ by definition.
\end{proof}

\begin{theorem}[Approximate counting]\label{thm:MinSep-count}
There is an algorithm for $\mathrm{MS}_{a,b}(G)$ with running time
\[
  O^{*}\!\left(\tau^{n/2}\cdot \varepsilon^{-2}\cdot \log(1/\delta)\right),
\]
obtained by applying \textsc{Count} to $\mathcal F_{\mathrm{MS}}^{\rightarrow,b}$
on $(G,a,b)$ and $(G,b,a)$ and summing the two estimates.
The corresponding quantum instantiations run in
$O^{*}(\tau^{n/3})$ and $O^{*}(\tau^{n/4})$ up to polynomial factors in
$n,\varepsilon^{-1}$ or $\varepsilon^{-\frac{3}{2}},\log(1/\delta)$.
\end{theorem}

\begin{corollary}[Counting all minimal separators]\label{cor:all-minsep}
There is an
\[
  O^{*}\!\left(\tau^{n/2}\varepsilon^{-2}\log(1/\delta)\right)
\]
time \((\varepsilon,\delta)\)-approximation algorithm for the number of
distinct minimal separators of \(G\).
\end{corollary}
\begin{proof}
Fix the deterministic vertex order used above.  For a minimal separator \(S\),
let \(\mathcal C_{\mathrm{full}}(S)\) be the set of full components of
\(G-S\), that is, components \(C\) with \(N_G(C)=S\).  This set has size at
least two.  Let \(C_1(S)\) and \(C_2(S)\) be the two full components whose
minimum vertices are first and second under the induced lexicographic order,
and set
\[
  \kappa(S):=(\min C_1(S),\min C_2(S)).
\]
For each ordered pair \((a,b)\), compute the number of minimal
\((a,b)\)-separators \(S\) with \(\kappa(S)=(a,b)\) by the same
two-orientation decomposition as in \Cref{thm:MinSep-count}: run the oriented
recursion on \((G,a,b)\) and on \((G,b,a)\), but in both runs add to the leaf
predicate the polynomial-time check \(\kappa(S)=(a,b)\), where \(S=X\) is the
separator candidate at the leaf and \(\kappa\) is computed in the original
graph.  This only deletes accepting leaves, so the same recursion-compatible
bound applies.  Each minimal separator has exactly one canonical ordered pair,
and for that pair exactly one of the two orientations accepts it.  Summing over
all ordered terminal pairs introduces only a polynomial factor, absorbed by
\(O^{*}(\cdot)\).
\end{proof}

\subsection{Counting Perfect Matchings in Subcubic Graphs}
\label{sec:pm-subcubic}

We instantiate the framework on \textsc{\#Perfect-Matching} restricted to
graphs of maximum degree~$3$.
This application is meant to highlight two themes that recur throughout the
paper:

\begin{enumerate}[(i)]
\item \textbf{Expose an unweighted core.}
Perfect matchings admit many ``weighted'' or multiplicative viewpoints
(e.g.\ via decompositions across components),
but our generic \(\sqrt{b}\) and quantum speed-ups require an \emph{unweighted}
leaf-counting recursion.  We therefore refactor the computation into a
canonical search tree whose leaves correspond to perfect matchings.

\item \textbf{Quantum-ready structure.}
The branching rule below is deterministic and local (degree / path queries), and the branches are merged additively in an unweighed manner.
Hence one can build a tree of solutions which supports the local oracles
(\(\mathsf{Child}\), \(\mathsf{Parent}\), degree queries) required by our
black-box quantum accelerations in \Cref{sec:quantum-aaronson,sec:Ambainis}.
This contrasts with existing algorithms for exact counting, as they do not straightforwardly benefit from access to quantum resources.
\end{enumerate}

\paragraph{Problem.}
Given a simple graph \(G=(V,E)\) with \(\Delta(G)\le 3\),
let \(\mathrm{PM}(G)\) denote the number of perfect matchings of~\(G\).
We write \(n:=|V|\) and \(V_3(G):=\{v\in V:\deg_G(v)=3\}\) with
\(n_3:=|V_3(G)|\).

\subsubsection{A Canonical Reduction to Degree $\{2,3\}$ Cores}
\label{sec:pm:reduce}

Our recursion will operate on \emph{reduced} instances, obtained by applying
simple forced rules.  This is the mechanism that makes the ``path forcing''
in the branching analysis precise.

\begin{definition}[\(\mathrm{Reduce}_{\mathrm{PM}}\)]
\label{def:pm-reduce}
On input a graph \(G\), repeatedly apply the following rules until none
applies:

\begin{enumerate}[(R1)]
\item \textbf{Parity filter.} If any connected component of \(G\) has an odd
number of vertices, return a special terminal instance \(\bot\).
(This is safe since a perfect matching cannot exist on an odd component.)

\item \textbf{Isolated vertex.} If \(G\) has a vertex of degree \(0\),
return \(\bot\).

\item \textbf{Forced degree-1.} If \(G\) has a degree-$1$ vertex \(u\) with
unique neighbour \(v\), then edge \(uv\) is forced in every perfect matching.
Delete \(u\) and \(v\) (and all incident edges) and continue.
\end{enumerate}

If all vertices are deleted, return the empty graph \(\varnothing\).
\end{definition}

\begin{definition}[Core graphs]\label{def:coregraphs}
A graph \(H\) is called a \emph{\(\{2,3\}\)-core} if it is a nonempty output of
\(\mathrm{Reduce}_{\mathrm{PM}}\) and hence satisfies:
\[
  \delta(H)\ge 2,\qquad \Delta(H)\le 3,\qquad
  \text{and every connected component of \(H\) has even size.}
\]
Equivalently, every vertex of \(H\) has degree \(2\) or~\(3\).
\end{definition}

\paragraph{Decision oracle.}
Our framework assumes feasibility in polynomial time:
\[
  \mathsf{Dec}_{\mathrm{PM}}(G)\;:=\;\bigl[\mathrm{PM}(G)>0\bigr].
\]
Perfect matching \emph{existence} is in \(\mathsf{P}\) (e.g.\ via blossom-style
matching algorithms), so \(\mathsf{Dec}_{\mathrm{PM}}\in\mathsf{P}\).
We will not use \(\mathsf{Dec}_{\mathrm{PM}}\) inside the branching rule below,
but it is available to our generic polynomial-delay enumerator
(\Cref{sec:PolyDelay}).

\subsubsection{2-path Structure Between Degree-3 Vertices}
\label{sec:pm:2paths}

\begin{lemma}[Even number of degree-3 vertices]\label{lem:pm-even-n3}
If every vertex of \(H\) has degree \(2\) or \(3\), then \(|V_3(H)|\) is even.
\end{lemma}
\begin{proof}
\(\sum_{v\in V(H)}\deg(v)=2|E(H)|\) is even.
Modulo \(2\), degree-$2$ vertices contribute \(0\) and degree-$3$ vertices
contribute \(1\), hence the sum is congruent to \(|V_3(H)|\pmod 2\).
\end{proof}

\begin{definition}[Maximal 2-path]\label{def:pm-2path}
Let \(H\) be a \(\{2,3\}\)-core.
A \emph{2-path} is a simple path
\(P=(v_1,v_2,\ldots,v_k)\) such that
\(\deg(v_1)=\deg(v_k)=3\) and \(\deg(v_i)=2\) for all \(1<i<k\).
It is \emph{maximal} if it cannot be extended at either endpoint while keeping
internal vertices of degree~\(2\).
\end{definition}

\begin{lemma}[Degree-3 vertices lie on 2-paths]\label{lem:pm-2path-exists}
Let \(H\) be a \(\{2,3\}\)-core with \(n_3>0\).
Then every degree-$3$ vertex \(v\in V_3(H)\) is an endpoint of some maximal
2-path \(P\) whose other endpoint is a (possibly different) degree-$3$ vertex.
Moreover, one can find such a path in time \(\poly(|V(H)|)\) by walking along
degree-$2$ chains.
\end{lemma}

\begin{proof}
Contract every maximal chain of degree-$2$ vertices into a single ``super-edge''.
The resulting multigraph has vertex set \(V_3(H)\), and every vertex has degree
\(3\) (each incident edge of the original degree-$3$ vertex exits along exactly
one maximal degree-$2$ chain).  Hence every vertex is incident to some
super-edge, which corresponds to a maximal 2-path in~\(H\).
\end{proof}

\subsubsection{Branching Rule: Choose the First Edge of a Maximal 2-path}
\label{sec:pm:branch}

We use the standard perfect-matching self-reduction on an edge:
for any edge \(e=\{u,v\}\in E(G)\),
\[
  \mathrm{PM}(G)\;=\;\mathrm{PM}(G-e)\;+\;\mathrm{PM}(G-\{u,v\}),
\]
since every perfect matching either excludes \(e\) or includes it
(and then must cover \(u\) and \(v\) by~\(e\)).

\paragraph{Canonical choice of the branching edge.}
Let \(H:=\mathrm{Reduce}_{\mathrm{PM}}(G)\).
If \(H\in\{\bot,\varnothing\}\) we are at a leaf.
Otherwise \(H\) is a \(\{2,3\}\)-core, and we define a deterministic branching
edge \(e(H)\) as follows:

\begin{enumerate}[(B1)]
\item If \(V_3(H)=\emptyset\), then every vertex has degree \(2\), so \(H\) is a
disjoint union of even cycles.  Pick the lexicographically first edge on the
lexicographically first cycle.

\item If \(V_3(H)\neq\emptyset\), let \(v_1\) be the lexicographically smallest
degree-$3$ vertex.  Among the (at most three) maximal 2-paths starting at \(v_1\),
choose one \(P=(v_1,v_2,\ldots,v_k)\) whose other endpoint is degree-$3$
(break ties deterministically), and set \(e(H):=\{v_1,v_2\}\).
\end{enumerate}

This rule is computable in \(\poly(|V(H)|)\) time.

\subsubsection{Forced Alternation on 2-Paths}
\label{sec:pm:forcing}

The reason 2-paths are the correct branching primitive is that once we decide
whether the \emph{first} edge is in the perfect matching, the matching behaviour
on the entire path is forced.

\begin{lemma}[Alternation on a 2-path]\label{lem:pm-alternation}
Let \(H\) be a \(\{2,3\}\)-core and let
\(P=(v_1,v_2,\ldots,v_k)\) be a 2-path in \(H\).
Fix whether the first edge \(\{v_1,v_2\}\) is included in the perfect matching.
Then the matching status of every edge of \(P\) is uniquely determined:
the matching must alternate along the path.

In particular, writing \(\ell:=k-1\) for the length:
\begin{enumerate}[(a)]
\item if \(\ell\) is even, then exactly one endpoint (\(v_1\) or \(v_k\)) is
matched \emph{inside} \(P\), and the other must be matched outside \(P\);
\item if \(\ell\) is odd, then either both endpoints are matched inside \(P\)
(include-branch) or both are matched outside \(P\) (exclude-branch).
\end{enumerate}
\end{lemma}

\begin{proof}
Every internal vertex \(v_i\) has degree \(2\) and is incident to exactly two
edges of \(P\); in any perfect matching it must be matched by exactly one of
them.  Thus the choice at the first edge propagates uniquely along the path.
The parity statement follows by checking which edges are selected by the
alternation.
\end{proof}

\subsubsection{Measure-and-Conquer Analysis}
\label{sec:pm:measure}

Define the measure on reduced instances
\[
  \mu(H)\;:=\;|V(H)| + |V_3(H)|,
\]
and set \(\mu(\bot)=\mu(\varnothing)=0\).
Note that for any \(\{2,3\}\)-core \(H\) we have \(|V_3(H)|\le |V(H)|\), hence
\(\mu(H)\le 2|V(H)|\).

\begin{lemma}[Branching vectors]\label{lem:pm-branch-vectors}
Let \(H\) be a \(\{2,3\}\)-core and let \(e(H)\) be the branching edge chosen
above.  Consider the two reduced children
\[
  H_0 := \mathrm{Reduce}_{\mathrm{PM}}(H-e(H)),\qquad
  H_1 := \mathrm{Reduce}_{\mathrm{PM}}(H-\{u,v\}) \ \text{ where } e(H)=\{u,v\}.
\]
Then one of the following holds:

\begin{enumerate}[(i)]
\item \textbf{Cycle case or even 2-path:} \(\mu(H_0)\le \mu(H)-4\) and
\(\mu(H_1)\le \mu(H)-4\), i.e.\ branching vector \((4,4)\).

\item \textbf{Odd 2-path:} \(\mu(H_0)\le \mu(H)-2\) and
\(\mu(H_1)\le \mu(H)-8\), i.e.\ branching vector \((8,2)\) up to reordering.
\end{enumerate}
\end{lemma}

\begin{proof}
We distinguish cases by the structure that defines \(e(H)\).

\paragraph{(i) \(V_3(H)=\emptyset\): disjoint union of even cycles.}
Branch on an edge \(e\) of a cycle \(C\).
Since every component of \(H\) has even size, every cycle has even length, so
\(|C|\ge 4\).
In either branch (include or exclude \(e\)), the degree-$1$ forcing rule in
\(\mathrm{Reduce}_{\mathrm{PM}}\) propagates around the cycle and deletes the
entire cycle, hence removes at least \(4\) vertices.  Therefore
\(\mu\) drops by at least \(4\) in both branches.

\paragraph{(ii) \(V_3(H)\neq\emptyset\): branch on a maximal 2-path
\(P=(v_1,\ldots,v_k)\) and the first edge \(e=\{v_1,v_2\}\).}
Let \(\ell=k-1\).

\emph{Even length \(\ell\).}
By \Cref{lem:pm-alternation}, in each branch exactly one endpoint is
matched inside the path (and is deleted), while the other endpoint loses one
incident edge and becomes degree \(2\) (so \(n_3\) decreases by \(1\)).
Moreover, \(\ell\ge 2\), so at least one internal degree-$2$ vertex is deleted
by the forced alternation.  Thus in each branch we have a drop of at least
\begin{align*}
  &2 \quad\text{(delete one degree-3 vertex)}\\
  +&\,1 \quad\text{(delete at least one degree-2 vertex)}\\
  +&\,1 \quad\text{(downgrade the other endpoint from degree 3 to 2)}\\
  =&\, 4,
\end{align*}
giving \((4,4)\).

\emph{Odd length \(\ell\).}
In the \textbf{exclude} branch, \Cref{lem:pm-alternation} says both endpoints
must be matched outside \(P\), so excluding \(e=\{v_1,v_2\}\) reduces the degree
of \(v_1\) from \(3\) to \(2\); additionally, the forced alternation implies
that the other endpoint \(v_k\) also cannot be matched inside the path, and its
incident path edge is excluded, reducing its degree from \(3\) to \(2\).
Hence \(n_3\) drops by at least \(2\), i.e.\ \(\Delta\mu\ge 2\).

In the \textbf{include} branch, \Cref{lem:pm-alternation} matches both
endpoints inside \(P\), so both degree-$3$ endpoints are deleted, contributing
a drop \(2+2=4\) to \(\mu\).  Each deleted endpoint has two remaining incident
edges \emph{outside} the path, so in total four edges are removed from the rest
of the core.  Each such removed edge decreases \(\mu\) by at least \(1\) after
full reduction: if its outside endpoint had degree \(3\), it is downgraded to
degree \(2\) (decreasing \(n_3\) by \(1\)); if it had degree \(2\), it becomes
degree \(1\) and triggers forced deletions (decreasing \(n\) by at least \(1\)).
Even if two removed edges hit the same outside vertex, that vertex's degree
drops by \(2\) and then either becomes degree \(1\) (forcing deletions) or degree
\(0\) (terminal), so the \emph{combined} contribution of the two edges is at
least \(2\).  Therefore the four removed edges contribute at least \(4\) more
drop in \(\mu\) (or the instance becomes terminal, in which case the drop is
even larger).  Hence \(\Delta\mu\ge 4+4=8\) in the include branch.

This yields the \((8,2)\) vector.
\end{proof}

\subsubsection{The Bounded uSR Definition}
\label{sec:pm:bound}

\Cref{lem:pm-branch-vectors} implies that every recursive step decreases
\(\mu\) by either \((4,4)\) or \((8,2)\).  Hence the standard measure bound
\(\,2^{\mu/4}\,\) is recursion-compatible up to integer rounding.
For the sampler we take the integer-valued bound
\begin{align}
  B_\ell \;:=\; \bigl\lfloor 2^{\ell/4}\bigr\rfloor
  \qquad(\ell\ge 0),\label{def:BoundingPM}
\end{align}
and define \(b_{\mathrm{PM}}(\langle G\rangle):=B_{\mu(\mathrm{Reduce}_{\mathrm{PM}}(G))}\).
Rounding down preserves the required inequalities:
\(2B_{\ell-4}\le B_\ell\) is immediate, and
\(B_{\ell-8}+B_{\ell-2}\le B_\ell\) holds since
\(2^{(\ell-8)/4}+2^{(\ell-2)/4}<2^{\ell/4}\).
Thus \(b_{\mathrm{PM}}\) is a valid bounding function in the sense of
\Cref{def:bounding}.

\begin{definition}[\(\mathcal{F}_{\mathrm{PM}}^{b}\): bounded uSR form for \textsc{\#PM}]\label{def:pm-uSR}
Instances are encodings \(\langle G\rangle\) of graphs with \(\Delta(G)\le 3\).
Let \(H:=\mathrm{Reduce}_{\mathrm{PM}}(G)\).

\begin{enumerate}[(a)]
\item \textbf{Leaf function.}
\[
  t_{\mathrm{PM}}(\langle G\rangle)
  \;:=\;
  \begin{cases}
    1, & \text{if } H=\varnothing,\\
    0, & \text{otherwise.}
  \end{cases}
\]

\item \textbf{Fan-out.}
\[
  r_{\mathrm{PM}}(\langle G\rangle)
  \;:=\;
  \begin{cases}
    0, & \text{if } H\in\{\bot,\varnothing\},\\
    2, & \text{otherwise.}
  \end{cases}
\]

\item \textbf{Branch map.}
If \(r_{\mathrm{PM}}(\langle G\rangle)=2\), let \(e(H)=\{u,v\}\) be the
deterministically chosen branching edge in the reduced core \(H\), and define
\[
  h_{\mathrm{PM}}(\langle G\rangle,1)\;:=\;\bigl\langle \mathrm{Reduce}_{\mathrm{PM}}(H-e(H))\bigr\rangle,
  \qquad
  h_{\mathrm{PM}}(\langle G\rangle,2)\;:=\;\bigl\langle \mathrm{Reduce}_{\mathrm{PM}}(H-\{u,v\})\bigr\rangle.
\]
(If \(r_{\mathrm{PM}}(\langle G\rangle)=0\), \(h_{\mathrm{PM}}\) is irrelevant by
the uSR conventions of \Cref{sec:uTotP}.)

\item \textbf{Depth bound.}
Set \(q_{\mathrm{PM}}(\langle G\rangle):=|V(G)|\).

\item \textbf{Bounding function.}
Let \(\mu(H):=|V(H)|+|V_3(H)|\) with \(\mu(\bot)=\mu(\varnothing)=0\), and define
\[
  b_{\mathrm{PM}}(\langle G\rangle)\;:=\;B_{\mu(H)}.
\]
In particular, \(b_{\mathrm{PM}}(\langle G\rangle)=1\) on accepting leaves,
which is exactly the positive-leaf tightness used in
\Cref{lem:sigma-ticket-count}.
\end{enumerate}
\end{definition}

\begin{lemma}[Recursion compatibility]\label{lem:pm-bound-compatible}
The map \(b_{\mathrm{PM}}\) is a bounding function for
\(\mathcal{F}_{\mathrm{PM}}\) in the sense of \Cref{def:bounding}.
\end{lemma}

\begin{proof}
Let \(H\) be a \(\{2,3\}\)-core with measure \(\mu=\mu(H)\), and let
\(H_0,H_1\) be its reduced children.
By \Cref{lem:pm-branch-vectors}, either
\((\mu(H_0),\mu(H_1))\le (\mu-4,\mu-4)\) or
\((\mu(H_0),\mu(H_1))\le (\mu-8,\mu-2)\) up to order.
We have
\[
  b(H_0)+b(H_1)\;\le\;
  \begin{cases}
    2B_{\mu-4}\le B_\mu,\\
    B_{\mu-8}+B_{\mu-2}\le B_\mu,
  \end{cases}
\]
so \(\sum_i b(h(H,i))\le b(H)\).
The inequality \(f(H)\le b(H)\) follows by the usual induction on the recursion
tree using the same inequalities.
\end{proof}

\subsubsection{Resulting Approximate Counting Runtimes}
\label{sec:pm:runtime}

\begin{theorem}\label{thm:pm-basic-classical}
Let \(G\) be an \(n\)-vertex graph with \(\Delta(G)\le 3\), and let
\(\mu:=\mu(\mathrm{Reduce}_{\mathrm{PM}}(G))\le 2n\).
Then Algorithm~\textsc{Count} (\Cref{thm:fullalgo}) applied to
\(\mathcal{F}_{\mathrm{PM}}^{b}\) returns an
\((\varepsilon,\delta)\)-approximation to \(\mathrm{PM}(G)\) in time
\[
\boxed{
  O^{\ast}\!\Bigl(
      2^{\,\mu/8}\;
      \varepsilon^{-2}\log\tfrac1\delta
    \Bigr)
  \;\le\;
  O^{\ast}\!\Bigl(
      2^{\,n/4}\;
      \varepsilon^{-2}\log\tfrac1\delta
    \Bigr)
  \;=\;
  O^{\ast}\!\Bigl(
      1.1893^{\,n}\;
      \varepsilon^{-2}\log\tfrac1\delta
    \Bigr).}
\]
In particular, for cubic graphs (\(n_3=n\), hence \(\mu=2n\)) the exponent is
exactly \(n/4\).
\end{theorem}

\begin{proof}
By \Cref{lem:pm-bound-compatible}, \(b_{\mathrm{PM}}\) is a valid bounding
function.  By \Cref{def:BoundingPM} and the branching analysis,
\(b_{\mathrm{PM}}(\langle G\rangle)=B_{\mu}=O^{\ast}(2^{\mu/4})\).
\Cref{thm:fullalgo} then yields runtime
\(O^{\ast}(\sqrt{b_{\mathrm{PM}}}\,\varepsilon^{-2}\log(1/\delta))
     =O^{\ast}(2^{\mu/8}\varepsilon^{-2}\log(1/\delta))\).
Finally \(\mu\le 2n\) implies \(2^{\mu/8}\le 2^{n/4}\).
\end{proof}

\begin{corollary}[Black-box quantum variants]\label{cor:pm-quantum}
With the same \(b_{\mathrm{PM}}\), the quantum variants of
\Cref{sec:quantum-aaronson,sec:Ambainis} yield
\[
  O^{\ast}\!\Bigl(
     2^{\,\mu/12}\;
     \varepsilon^{-1}\log\tfrac1\delta
  \Bigr)
  \qquad\text{and}\qquad
  O^{\ast}\!\Bigl(
     2^{\,\mu/16}\;
     \varepsilon^{-3/2}\log^{2}\tfrac1\delta
  \Bigr),
\]
and hence, in cubic graphs, \(O^{\ast}(2^{n/6})\) and \(O^{\ast}(2^{n/8})\),
respectively (up to the stated \(\varepsilon,\delta\) factors).
\end{corollary}

\paragraph{Take-away for the reader.}
This section is a worked example of the ``unweighted core'' design principle:
we engineer a deterministic recursion where every internal node contributes only
a \emph{sum} of two reduced subinstances, and the combinatorial forced structure
(2-path alternation + degree-$1$ propagation) is captured by an explicit
recursion-compatible bound \(b_{\mathrm{PM}}\).
Once this interface is exposed, the generic \(\sqrt{b}\), \(b^{1/3}\), and
\(b^{1/4}\) accelerations apply without further problem-specific probability or
quantum reasoning.

\paragraph{Improving on our basic algorithm and analysis.}
It is possible to get a much faster algorithm for approximate counting of perfect matching in subcubic graph.
Since the simpler algorithm demonstrates the main idea of this paper better, the fastest running algorithm is presented in Appendix~\ref{sec:ImprovedPerfect}.
The improved runtime is stated below:

\begin{restatable}{theorem}{CUBICMATCHING}\label{thm:pm-classical}
Let $G$ be an $n$-vertex subcubic graph.  Let $H=\mathrm{Reduce}_{\mathrm{PM}}(G)$ and write
$n_2\coloneqq n_2(H)$ and $n_3\coloneqq n_3(H)$.  For every $\varepsilon,\delta\in(0,1)$ there is a
randomised algorithm that returns an $(\varepsilon,\delta)$-approximation to $\#\mathrm{PM}(G)$ in time
\[
O^*\!\Bigl(\varepsilon^{-2}\log(\delta^{-1})\cdot \sqrt{b(H)}\cdot \mathrm{poly}(n)\Bigr)
\;=\;
O^*\!\Bigl(\varepsilon^{-2}\log(\delta^{-1})\cdot 2^{n_2/8}\,6^{n_3/12}\cdot \mathrm{poly}(n)\Bigr).
\]
In particular, since $n_2+n_3\le n$, we have the worst-case bound
\[
\boxed{
O^*\!\Bigl(\varepsilon^{-2}\log(\delta^{-1})\cdot 6^{n/12}\cdot \mathrm{poly}(n)\Bigr)
\;\le\;
O^*\!\Bigl(\varepsilon^{-2}\log(\delta^{-1})\cdot (1.1611)^n\Bigr).
}
\]
\end{restatable}
\begin{proof}
    See Appendix~\ref{sec:ImprovedPerfect}.
\end{proof}
\begin{corollary}
    Applying the quantum variant of the algorithm (\Cref{sec:Ambainis}), we get a runtime of

    \[
\boxed{
  O^{\ast}\!\Bigl(
      1.0776^{\,n}\;
      \varepsilon^{-\frac32}\log^{2}\tfrac1\delta
    \Bigr).}
\]
\end{corollary}

\paragraph{Further directions.} We expect that this line of work can be further pursued to obtain good bounds for higher-degree graphs, though the runtime will decay very aggressively based on $\Delta$.
Nonetheless, we believe this would give faster than the state-of-the-art for degree 4 or 5.
An interesting further direction, which also seems promising, is to adapt the algorithms for low-average degree graphs, rather than bounded degree graphs.

As stated in Appendix~\ref{sec:matchingTightness}, the upper bound linked to the branching rule of the algorithm is tight.
This means that we cannot hope to do better without further algorithmic techniques.
This is where decomposition of the problem becomes necessary: to beat the bound on the number of solutions, one must separate the easy-to-count components of the problem, and apply the enumerate-or-sample technique on the hard cores that are left.
Leveraging this more sophisticated technique, which we will illustrate in the next two examples, could also help tackle more general classes of graphs for perfect matching.

\subsection{Counting Independent Sets}\label{sec:IS-basic}

We demonstrate how the generic framework applies to the classical
\textsc{\#Independent-Set} problem.  After recalling a naive bounded-uSR
encoding, we give a simple high-degree branching routine that peels off
the hard part of the instance.  The routine produces a search tree of size
\(2^{0.249n}\), whose hard leaves are aggregated by \Cref{thm:combineuTotP}
while the easy leaves are handled by known FPRAS regimes for bounded-degree
or low-connective-constant instances~\cite{sinclair2017spatial,goldberg2021faster}.
The basic version gives an \(O^{\ast}(1.1884^n)\) approximation, and the
aligned refined stopping rule improves this to \(O^{\ast}(1.1869^n)\).

\subsubsection{A Naive Bounded-uSR Definition}\label{sec:IS-uSR}

\begin{definition}[\(\mathcal{F}_{\mathrm{IS}}\) - naive uSR definition for \textsc{\#IS}]\label{def:IS-branch-f}
Instances are encodings \(\langle G\rangle\) of simple graphs \(G=(V,E)\).
Assume the encoding lists the vertices of the \emph{current} graph \(G\) in a deterministic order
\(v_1,\dots,v_{|V|}\).
For \(v\in V\), let \(N_G(v)\) denote the open neighbourhood of \(v\), and write
\(N_G[v]:=\{v\}\cup N_G(v)\) for the closed neighbourhood.

\begin{enumerate}[(a)]
\item \textbf{Instance transformation (branch map).}
      If \(V(G)\neq\varnothing\), define
      \[
        h_{\mathrm{IS}}(\langle G\rangle,1)=\bigl\langle G\setminus\{v_1\}\bigr\rangle,\qquad
        h_{\mathrm{IS}}(\langle G\rangle,2)=\bigl\langle G\setminus N_G[v_1]\bigr\rangle.
      \]
      Thus branch \(1\) \emph{excludes} \(v_1\), while branch \(2\) \emph{includes} \(v_1\) and deletes \(v_1\) together with its neighbours.

\item \textbf{Leaf function.}
      \[
        t_{\mathrm{IS}}(\langle G\rangle)=
        \begin{cases}
          1 &\text{if } V(G)=\varnothing,\\
          0 &\text{otherwise.}
        \end{cases}
      \]

\item \textbf{Counting function (induced uSR recurrence).}
      The function \(f_{\mathrm{IS}}\) is defined by the unweighted self-reduction:
      \[
        f_{\mathrm{IS}}(\langle G\rangle)=
        \begin{cases}
          t_{\mathrm{IS}}(\langle G\rangle) & \text{if } r_{\mathrm{IS}}(\langle G\rangle)=0,\\[2mm]
          f_{\mathrm{IS}}\!\bigl(h_{\mathrm{IS}}(\langle G\rangle,1)\bigr)+
          f_{\mathrm{IS}}\!\bigl(h_{\mathrm{IS}}(\langle G\rangle,2)\bigr)
          & \text{otherwise.}
        \end{cases}
      \]
      Equivalently,
      \[
        f_{\mathrm{IS}}(\langle G\rangle)=
        \begin{cases}
          1 & \text{if } V(G)=\varnothing,\\
          f_{\mathrm{IS}}\!\bigl(\langle G\setminus\{v_1\}\rangle\bigr)+
          f_{\mathrm{IS}}\!\bigl(\langle G\setminus N_G[v_1]\rangle\bigr)
          & \text{otherwise.}
        \end{cases}
      \]

\item \textbf{Depth and fan-out.}\;
      \[
        r_{\mathrm{IS}}(\langle G\rangle)=
        \begin{cases}
          0 & \text{if } V(G)=\varnothing,\\
          2 & \text{otherwise,}
        \end{cases}
        \qquad
        q_{\mathrm{IS}}(\langle G\rangle)=|V(G)|.
      \]

\item \textbf{Bounding function.}\;
      \[
        b_{\mathrm{IS}}(\langle G\rangle)=2^{|V(G)|}.
      \]
\end{enumerate}
\end{definition}

\begin{lemma}
For every graph \(G\), \(f_{\mathrm{IS}}(\langle G\rangle)=\#\mathrm{IS}(G)\).
\end{lemma}

\begin{proof}
Let \(v_1\) be the first vertex in the deterministic order of \(G\).
Every independent set \(S\subseteq V(G)\) falls into exactly one of two disjoint cases:
either \(v_1\notin S\), in which case \(S\) is an independent set of \(G\setminus\{v_1\}\);
or \(v_1\in S\), in which case \(S\setminus\{v_1\}\) is an independent set of \(G\setminus N_G[v_1]\).
Thus the recursion partitions independent sets bijectively across the two branches.
The recursion terminates exactly when all vertices have been deleted, and each such leaf contributes
\(1\) via \(t_{\mathrm{IS}}(\langle G\rangle)=1\), including the leaf corresponding to the empty independent set.
Hence \(f_{\mathrm{IS}}(\langle G\rangle)\) equals the number of independent sets of \(G\).
\end{proof}

\subsubsection{A Basic High-Degree Branching Procedure}\label{sec:ISbasic}

The routine below recursively branches on a vertex of maximum
degree \(\Delta\ge6\) until either:
(i)~the remaining instance has maximum degree \(\le5\), or
(ii)~a \emph{budget} of \(\beta n\) deleted vertices is exhausted.  The
output is a list \(\mathcal L\) of residual graphs, each on at most
\((1-\beta)n\) vertices.

\begin{algorithm}[H]
\DontPrintSemicolon
\caption{\textsc{IS-Preprocess}\((G,\beta)\)}\label{alg:IS-pre}

\KwIn{graph \(G=(V,E)\), \(|V|=n\);\qquad budget fraction \(\beta\in(0,1)\)}
\KwOut{list \(\mathcal L\) of instances with $|V|\leq (1-\beta)n$, list \(\mathcal{E}\) with $\Delta \leq 5$}

\vspace{.2em}
\SetKwFunction{Branch}{Branch}
\SetKwProg{Fn}{Function}{:}{}
\Fn{\Branch{$H,\; b$}}{%
 \If(\tcp*[f]{budget exhausted}){$b\le 0$}{
   append \(H\) to \(\mathcal L\); \Return
 }
 \If(\tcp*[f]{easy instance}){$\Delta(H)\le 5$}{
   append \(H\) to \(\mathcal E\); \Return
 }
 choose a vertex \(v\) of maximum degree \(\Delta(H)\ge6\)\;
 \Branch{$H\!-\!v,\; b-1$}\tcp*{exclude \(v\)}
 \Branch{$H\!-\!N_H[v],\; b-(\deg_H(v)+1)$}\tcp*{include \(v\)}
}

\vspace{.4em}
\(\mathcal L\gets\emptyset\);\;
\(\mathcal E\gets\emptyset\);\;
\Branch{$G,\; \beta n$}\;
\Return \({\mathcal L}, {\mathcal E}\)
\end{algorithm}

\paragraph{Size of the branching tree.}
Each branching step on degree \(\Delta\ge6\) yields the recurrence
\(T(n)=T(n-1)+T(n-7)\) with characteristic root
\(\alpha\approx1.25543\), with \(\log_2\alpha\le0.328174\).  Therefore
\[
  |\text{search tree}|\;\le\;2^{0.328174\beta n}.
\]
Here the recurrence is in the deleted-vertex budget used by the preprocessing
tree, not directly in the number of vertices of the residual graph.

\paragraph{Balancing budget and residual size.}
Set \(\beta:=0.757\).
Then
\(0.328174\beta<0.2485\) and \(1-\beta=0.243\).

We make the following observations:

\begin{enumerate}
\item \emph{Pre-processing:}\; at most \(2^{0.249n}\) leaves.\label{obs:Observation1}
\item \emph{Hard graphs}(\(\mathcal{L}\)):\;  \(\le0.243n\) vertices each. Therefore for $x\in \mathcal{L}$, $b(x)\leq 2^{0.243n}$. From \Cref{obs:Observation1} we also have $|\mathcal{L}|\le 2^{0.249n}$.\label{obs:Observation2}
\item \emph{Easy graphs} (\(\mathcal{E}\)): Degree $\le 5$. From \Cref{obs:Observation1}, $|\mathcal{E}|\le 2^{0.249n}$.\label{obs:Observation3}
\end{enumerate}

\paragraph{Summing the functions.}

Let $x_1,\ldots,x_{|\mathcal{L}|}$ be the elements of $\mathcal{L}$.  The
preprocessing tree explicitly constructs these cores in
\(S=O^{\ast}(2^{0.249n})\) time.  Since \Cref{obs:Observation2} gives at most
\(2^{0.249n}\) hard instances, each with at most \(0.243n\) vertices, we have
\begin{align*}
\sum_{i=1}^{|\mathcal L|}b_{\mathrm{IS}}(x_i)&\le \sum_{i=1}^{|\mathcal L|}2^{0.243n}\\
&\le 2^{0.249n}\times2^{0.243n}\\
&=2^{0.492n}.
\end{align*}
\Cref{thm:combineuTotP}, run with failure probability \(\delta/2\), then delivers an
\((\varepsilon,\delta/2)\)-approximation in
\(S+O^{\ast}\!\bigl(\sqrt B\,\varepsilon^{-2}\log\tfrac1\delta\bigr)
 = O^{\ast}\!\bigl(2^{0.249n}\,\varepsilon^{-2}\log\tfrac1\delta\bigr)\).

\medskip\noindent
\textbf{Summing easy instances.} By \Cref{obs:Observation3} there are at most
\(2^{0.249n}\) easy instances, and each has maximum degree at most \(5\).  The
known FPRAS for this regime~\cite{goldberg2021faster,sinclair2017spatial}
counts each easy instance in polynomial time in its size, \(1/\varepsilon\),
and \(\log(1/\delta')\).  Run the FPRAS on each easy instance with failure
probability \(\delta'=\delta/(2|\mathcal E|)\) and relative error \(\varepsilon\).
A union bound gives simultaneous success with probability at least \(1-\delta/2\);
conditioned on this event, summing the returned non-negative estimates gives a
\((1\pm\varepsilon)\)-approximation to the easy contribution, since termwise
lower and upper bounds are preserved under summation.  The total easy
cost is
\[
  O^{\ast}\!\left(2^{0.249n}\operatorname{poly}\!\left(\frac1\varepsilon,\log\frac1\delta\right)\right).
\]

\medskip\noindent
\textbf{Overall running time.}\;
The computation has a \(2^{0.249n}\) preprocessing phase, and two \(2^{0.249n}\) (approximate) counting phases,
so for fixed \(\varepsilon,\delta\)

\[
  \boxed{\;T_{\text{basic}}(n)
         \;=\;
         O^{\ast}\!\bigl(2^{0.249n}\bigr)
         \;=\;
         O^{\ast}\!\bigl(1.1884^{\,n}\bigr).}
\]

\qed

\subsubsection{Refined Branching with a Measure-Based Easy Test}\label{sec:IS-advanced}

We now tighten the running time by stopping the high-degree branching as soon
as the instance enters a known easy regime.  The algorithm exploits the fact
that graphs where every vertex with degree at least $6$ has 2-degree at most
$26$ admit an FPRAS for independent set counting
\cite{goldberg2021faster}. Here, the 2-degree of a vertex $v$, denoted
$\deg^2(v)$, is the sum of its neighbours' degrees, and the 2-degree of a
graph is the maximum 2-degree taken over all its vertices.

\paragraph{The measure.}
Let \(\gamma=1/6\).
For a graph \(H\) define

\[
    \mu(H)\;=\;
    \sum_{v\in V(H)}\!w_H(v),
    \qquad
    w_H(v)\;=\;
    \begin{cases}
        1, & \deg_H(v)\ge 6 \text{ or } N_{H}(v) \cap V_{\ge 7}(H)\neq \varnothing, \\[2pt]
        1-(6-\deg_H(v))\gamma, & \text{otherwise.}
    \end{cases}
\]
\paragraph{Branching rules.}
\textsc{FastIS-Preprocess} applies two branching rules:

\begin{enumerate}[(1)]
  \item If a vertex \(v\) of degree \(\ge 7\) exists,
        branch on \(v\) (include/exclude).
  \item Otherwise, pick vertex $v$ with degree at least $6$ and  $2$-degree at least $27$, branch on $v$ (include/exclude).
\end{enumerate}

For the second rule, when excluding $v$, we also decrease the degree of its 6 neighbours, then we get an extra drop of $6\gamma$.
When including $v$, the neighbours of $v$ have weight at least $6-(36-27)\gamma = 6-9\gamma$. Therefore the measure drops by at least
\[
(\Delta \mu_{\text{excl}}, \; \Delta \mu_{\text{incl}})=
\begin{cases}
    (1,\;8), & \text{rule 1},\\
    (1+6\gamma,7-9\gamma), &\text{rule 2}
\end{cases}
\]
Substituting $\gamma=1/6$ gives branching vectors
\[
(\Delta \mu_{\text{excl}}, \; \Delta \mu_{\text{incl}})=
\begin{cases}
    (1,\;8), & \text{rule 1},\\
    (2,5.5), &\text{rule 2}
\end{cases}
\]
These are genuine measure drops, but the routine below is still stopped by a
deleted-vertex budget.  The tree-size proof therefore cannot use the measure
recurrence alone.  We instead analyse the recursion with a two-resource
potential that includes both the remaining deletion budget and the measure.
For the present pair of branching rules the valid choice used below puts all
weight on the deletion budget; this is enough to align the proof, but it also
explains why the final exponent below is weaker than the measure-only heuristic
would suggest.

\begin{algorithm}[H]
\DontPrintSemicolon
\caption{\textsc{FastIS-Preprocess}\((G,\beta)\)}
\label{alg:fastIS}
\KwIn{graph \(G=(V,E)\), \(|V|=n\);\qquad budget fraction \(\beta\in(0,1)\)}
\KwOut{list \(\mathcal L\) of instances with $|V|\leq (1-\beta)n$, list \(\mathcal{E}\) of easy instances}

\vspace{.2em}
\SetKwFunction{Branch}{Branch}
\SetKwProg{Fn}{Function}{:}{}
\Fn{\Branch{$H,\; b$}}{%
 \If(\tcp*[f]{budget exhausted}){$b\le 0$}{
   append \(H\) to \(\mathcal L\); \Return
 }
 \If(\tcp*[f]{measure/easy stopping condition}){each \(v\in V(H)\) with \(\deg_H(v)\ge 6\) has \(\deg_H^{2}(v)\leq 26\)}{
   append \(H\) to \(\mathcal E\); \Return
  }
 \If{$\Delta(H)\ge 7$}{
    choose a vertex \(v\) of maximum degree \(\Delta(H)\ge7\)\;
    \Branch{$H\!-\!v,\; b-1$}\tcp*{exclude \(v\)}
    \Branch{$H\!-\!N_H[v],\; b-(\deg_H(v)+1)$}\tcp*{include \(v\)}
    \Return
 }
 choose a vertex $v$ with $\deg(v)\geq 6$ and $\deg^2(v)\geq 27$\;\label{line:lastBranch}
 \Branch{$H\!-\!v,\; b-1$}\tcp*{exclude \(v\)}
    \Branch{$H\!-\!N_H[v],\; b-(\deg_H(v)+1)$}\tcp*{include \(v\)}
}
\vspace{.4em}
\(\mathcal L\gets\emptyset\);\;
\(\mathcal E\gets\emptyset\);\;
\Branch{$G,\; \beta n$}\;
\Return \({\mathcal L},{\mathcal E}\)
\end{algorithm}

\begin{lemma}[Preprocessing guarantee for \textsc{FastIS-Preprocess}]
\label{lem:fastIS-preprocess-guarantee}
For \(\beta=0.7529\), \textsc{FastIS-Preprocess} outputs lists
\(\mathcal L,\mathcal E\) such that
\[
  |\mathcal L|+|\mathcal E|\le 2^{0.247083n},
\]
every \(H\in\mathcal L\) has at most \(0.2471n\) vertices, and every
\(H\in\mathcal E\) satisfies the hard-core FPRAS easy-instance condition
used in the stopping rule.
\end{lemma}

\begin{proof}
The recursion is truncated as soon as one of two conditions is met.  If the
deleted-vertex budget \(b\) is exhausted, the current instance is emitted as a
hard core.  Since the initial budget is \(\beta n\), every such hard core has
at most \((1-\beta)n=0.2471n\) vertices.

If the measure/easy stopping condition is met first, the current instance is
emitted into \(\mathcal E\) and is handled by the cited hard-core FPRAS machinery.
Thus the only remaining question is the size of the frontier exposed by this
two-cutoff recursion.

Let \(a=0.328174\).  More generally, one may use the potential
\(\Phi_{a,c}(H,b)=ab+c\mu(H)\); here we take the valid special case
\(c=0\), so \(\Phi(H,b)=ab\).  Rule~1 decreases \(b\) by
at least \(1\) and \(8\) in its two branches, and Rule~2 decreases \(b\) by at
least \(1\) and \(7\).  The choice of \(a\) satisfies
\[
  2^{-a}+2^{-7a}\le 1,
\]
and therefore also \(2^{-a}+2^{-8a}\le1\).  The standard induction on
\(\Phi\) gives at most \(2^{a\beta n}\) frontier nodes.  With
\(\beta=0.7529\),
\[
  a\beta \le 0.247083,\qquad 1-\beta=0.2471.
\]
Thus the frontier has size at most \(2^{0.247083n}\), and the deletion-budget
cutoff leaves hard residuals on at most \(0.2471n\) vertices.
\end{proof}

\paragraph{Overall runtime.}
By \Cref{lem:fastIS-preprocess-guarantee}, the preprocessing tree has size at
most \(2^{0.247083n}\), so its setup cost is
\(S=O^{\ast}(2^{0.247083n})\).  Every hard residual has at most
\(0.2471n\) vertices, and hence
\[
  \sum_{x_i\in\mathcal L} b_{\mathrm{IS}}(x_i)
  \le
  2^{0.247083n}\cdot 2^{0.2471n}
  =
  2^{0.494183n}.
\]
Thus \Cref{thm:combineuTotP}, again run with failure probability \(\delta/2\),
handles the hard leaves in
\[
  S+O^{\ast}\!\left(2^{0.247092n}\varepsilon^{-2}\log\tfrac1\delta\right)
  =
  O^{\ast}\!\left(2^{0.247092n}\varepsilon^{-2}\log\tfrac1\delta\right).
\]
The easy leaves are exactly the instances certified by the
hard-core FPRAS condition used in the stopping rule, and their total
contribution is approximated with the same union-bound allocation as in the
basic analysis.  This gives:
\begin{theorem}\label{thm:IS}
For fixed $(\varepsilon,\delta)$, we can approximately count independent sets
in time
\[
  \boxed{\;T_{\text{fast}}(n)
         \;=\;
         O^{\ast}\!\bigl(2^{0.247092n}\bigr)
         \;=\;
         O^{\ast}\!\bigl(1.1869^{\,n}\bigr).}
\]
\end{theorem}

This improves the previously cited general-graph \(O^{\ast}(1.2041^{\,n})\) bound, and is
also below the \(O^{\ast}(1.1884^{\,n})\) bound from the basic high-degree
decomposition above.
The algorithm is generic once the decomposition has been laid out: it exploits
the bounded-uSR formulation together with the existence of an FPRAS for the
low-connective-constant instances.
This suggests broad applicability of the approach and the possibility of further refinements.

\subsection{Counting \texorpdfstring{$2$}{2}-SAT Solutions}
\label{sec:2-SAT}
We show how the generic framework applies to \#\textsc{2-SAT} through a
bounded-uSR formulation and a preprocessing decomposition into easy and hard
residuals. The bounded-uSR state and the hard-core aggregation are fully
formalised below, and the specialised preprocessing bound is proved explicitly.
Throughout, \(n\) denotes the number of variables of the input formula.

This is conceptually interesting because \#\textsc{2-SAT} occupies a unique
position where neither of the two familiar routes to faster approximation
applies cleanly: it lacks the FPRAS regimes available for more structured graph
problems, yet it is not expressive enough for the hashing-based techniques that
apply to \#\textsc{3-SAT} and beyond. This may help explain the gap in the
literature regarding exponential-time approximation for \#\textsc{2-SAT}
\cite{cardinal2018solving, thurley2011approximation}. Our unweighted
self-reducibility framework isolates a third mechanism: decompose into
polynomial-time solvable pieces plus a sum of bounded hard cores, and then pay
the square root after aggregation rather than before.

\subsubsection{A Naive Bounded-uSR Definition}
\label{sec:2-SAT-uSR}

Fix an encoding \(\langle F\rangle\) for residual \(2\)-CNF instances that
records both the simplified clause set and the set of still-unassigned
variables in a deterministic order. Thus \(\#\operatorname{vars}(F)\) denotes
the number of still-unassigned variables, not merely the number of variables
that still appear in clauses. Variables that disappear from all clauses but
have not been fixed remain in this free-variable set and may still be branched
on later. When \(F\) has at least one remaining variable, let \(x(F)\) be the
first still-unassigned variable in this order. For a literal \(\ell\), write
\(F[\ell]\) for the residual instance obtained by setting \(\ell\) to
\textsc{true}, performing unit propagation, deleting every variable fixed by
this process from the free-variable set, and simplifying the clause set.

\begin{definition}[\(\mathcal{F}_{\mathrm{2-SAT}}\)]
\label{def:2-SAT-uSR}
\leavevmode
\begin{enumerate}[(a)]
\item \textbf{Instance transformation.}
      If \(F\) is satisfiable and \(\#\operatorname{vars}(F)>0\), define
      \[
         h_{\mathrm{2-SAT}}\bigl(\langle F\rangle,1\bigr):=\langle F[\overline{x(F)}]\rangle,
         \quad
         h_{\mathrm{2-SAT}}\bigl(\langle F\rangle,2\bigr):=\langle F[x(F)]\rangle .
      \]
      In all other cases \(r_{\mathrm{2-SAT}}(\langle F\rangle)=0\), so the
      global invalid-branch convention sends non-empty branch indices to
      \(\lambda\).
      \item \textbf{Leaf function}
       \[
    t_{\mathrm{2-SAT}}\bigl(\langle F \rangle \bigr)=
    \begin{cases}
         1 &\text{if \(F\) is satisfiable and \(\#\operatorname{vars}(F)=0\),}\\
         0 &\text{otherwise.}
    \end{cases}
    \]

\item \textbf{Counting function.}\;
      \[
       f_{\mathrm{2-SAT}}(\langle F\rangle)=
         t_{\mathrm{2-SAT}}(\langle F\rangle)+
            \sum_{i=1}^{r_{\mathrm{2-SAT}}(\langle F\rangle)}
            f_{\mathrm{2-SAT}}\!\bigl(h_{\mathrm{2-SAT}}(\langle F\rangle,i)\bigr).
       \]
\item \textbf{Depth and fan-out.}\;
      \[
      r_{\mathrm{2-SAT}}(\langle F\rangle)=
      \begin{cases}
      2, & \text{if \(F\) is satisfiable and \(\#\operatorname{vars}(F)>0\),}\\
      0, & \text{otherwise,}
      \end{cases}
      \qquad
      q_{\mathrm{2-SAT}}(\langle F\rangle)=\#\operatorname{vars}(F).
      \]
\item \textbf{Bounding function.}\;
      \[
      b_{\mathrm{2-SAT}}(\langle F\rangle)=
      \begin{cases}
      0, & \text{if \(F\) is unsatisfiable,}\\
      1, & \text{if \(F\) is satisfiable and \(\#\operatorname{vars}(F)=0\),}\\
      2^{\#\operatorname{vars}(F)}, & \text{if \(F\) is satisfiable and \(\#\operatorname{vars}(F)>0\).}
      \end{cases}
      \]
\end{enumerate}
\end{definition}
\begin{lemma}
    $f_{\mathrm{2-SAT}}(\langle F\rangle)=\# \mathrm{2-SAT}(F)$
\end{lemma}

\begin{proof}
    The branching function sets the next remaining variable and propagates that
    value. Variables that have disappeared from all clauses but remain
    unassigned stay in the state and are branched on later, so the recursion
    still accounts for their free choices.
    If a residual formula \(F'\) is unsatisfiable, then
    \(r_{\mathrm{2-SAT}}(\langle F'\rangle)=0\) and
    \(t_{\mathrm{2-SAT}}(\langle F'\rangle)=0\), so that branch contributes
    nothing.  If all variables have been set and the residual formula is
    satisfiable, then \(t_{\mathrm{2-SAT}}=1\), and the branch contributes the
    unique satisfying assignment represented by those choices.  Thus the leaves
    with value \(1\) are in bijection with satisfying assignments of \(F\).
\end{proof}

\begin{lemma}\label{lem:2sat-bounding}
The map \(b_{\mathrm{2-SAT}}\) is a bounding function for
\(\mathcal{F}_{\mathrm{2-SAT}}\).
\end{lemma}
\begin{proof}
If \(F\) is unsatisfiable, then \(r_{\mathrm{2-SAT}}(\langle F\rangle)=0\) and
\(b_{\mathrm{2-SAT}}(\langle F\rangle)=0\), so the leaf condition holds. If
\(F\) is satisfiable and \(\#\operatorname{vars}(F)=0\), then
\(r_{\mathrm{2-SAT}}(\langle F\rangle)=0\),
\(t_{\mathrm{2-SAT}}(\langle F\rangle)=1\), and
\(b_{\mathrm{2-SAT}}(\langle F\rangle)=1\), so both the leaf condition and
leaf tightness hold.

Now assume that \(F\) is satisfiable and has \(n>0\) still-unassigned
variables. Branching on \(x(F)\) yields the two residual instances
\(F[\overline{x(F)}]\) and \(F[x(F)]\). After assignment and unit propagation,
each residual is either unsatisfiable, contributing bound \(0\), or
satisfiable on at most \(n-1\) still-unassigned variables, contributing at
most \(2^{n-1}\). Hence
\[
  \sum_{i=1}^{r_{\mathrm{2-SAT}}(\langle F\rangle)}
  b_{\mathrm{2-SAT}}\!\bigl(h_{\mathrm{2-SAT}}(\langle F\rangle,i)\bigr)
  \le 2^{n-1}+2^{n-1}
  =2^n
  =b_{\mathrm{2-SAT}}(\langle F\rangle).
\]
Thus \(b_{\mathrm{2-SAT}}\) is recursion compatible and therefore a bounding
function.
\end{proof}

\subsubsection{Pre-processing by High-Degree Branching}
\label{sec:2-SAT-pre}

For a residual instance \(F\), let \(C(F)\) be its simple constraint graph on
the still-unassigned variables, with an edge \(uv\) whenever some clause of
\(F\) contains one literal on \(u\) and one literal on \(v\). For
\(v\in V(C(F))\), write \(N_F(v)\) for its neighbourhood in \(C(F)\) and
\(\deg(v):=|N_F(v)|\).  For each neighbour \(u\in N_F(v)\), fix one
deterministic witness clause of \(F\) containing literals on \(u\) and \(v\)
(for example, the first such clause under the fixed encoding order), and place
\(u\) in \(N_F^+(v)\) if \(v\) appears positively in that witness clause and in
\(N_F^-(v)\) otherwise.  Define
\[
  d_F^+(v):=|N_F^+(v)|,\qquad d_F^-(v):=|N_F^-(v)|,
\]
so that \(N_F^+(v)\) and \(N_F^-(v)\) partition \(N_F(v)\), whence
\(\deg(v)=d_F^+(v)+d_F^-(v)\).  If we branch on \(v\), then setting \(v=0\)
falsifies the witness clause for every \(u\in N_F^+(v)\), producing a unit
clause on \(u\); after unit propagation, each such \(u\) is fixed or the branch
is declared unsatisfiable. The same holds with \(N_F^-(v)\) when setting
\(v=1\). Thus the two branches fix at least \(1+d_F^+(v)\) and
\(1+d_F^-(v)\) still-free variables, respectively, unless they die even
earlier. In Algorithm~\textsc{2-SAT-Preprocess} we track the exact number of
fixed variables via
\[
  s(H):=n-\#\operatorname{vars}(H),
\]
which is computable because the residual state remembers all still-unassigned
variables, including variables that no longer occur in any clause.

\begin{algorithm}[H]
\DontPrintSemicolon
\caption{\textsc{2-SAT-Preprocess}\((F)\)}
\label{alg:2-SAT-pre}
\KwIn{$2$-CNF formula $F$ on $n$ variables}
\KwOut{lists $\mathcal L$ of \emph{hard} residuals and $\mathcal E$ of
       \emph{easy} residuals}
       \SetKwProg{Fn}{Function}{:}{}
$\mathcal L,\mathcal E\gets\emptyset$\;
$\beta\gets0.6999$\;
\Fn{\Branch{$H$}}{%
  $s\gets n-\#\operatorname{vars}(H)$\;
  \If{$s\ge\beta n$}{append $H$ to $\mathcal L$; \Return}
  \If{$\max_{v}\deg(v)\le2$}{append $H$ to $\mathcal E$; \Return}

  \lIf(\tcp*[f]{prefer a degree-$6$ variable}){$\exists v:\deg(v)=6$}{take such $v$}
  \lElse{$v\gets$ variable of max.\ degree $\ge3$}

  \Branch{$H[\bar v]$}\tcp*{set $v=0$}
  \Branch{$H[v]$}\tcp*{set $v=1$}
}
\Branch{$F$}\;
\Return $(\mathcal L,\mathcal E)$
\end{algorithm}
\paragraph{Degree-$6$-heavy regime.}
The specialised analysis below targets the regime in which the simple
constraint graph has maximum degree at most \(6\) and at least
\((1-\eta)n\) variables have degree \(6\).  The global hybrid algorithm handles
variables of degree at least \(7\) separately before this regime is invoked; see
the high-degree normalisation in the proof of \Cref{thm:2-SAT}.  Thus, when the
specialised routine is called, the only complementary maximum-degree-\(6\) case
is that more than \(\eta n\) variables have degree below \(6\), and this case is
handled by Wahlstr{\"o}m's degree-based algorithm.

\paragraph{Size of the search tree.}
Recall that the preprocessing phase stops once at least \(\beta n\) free
variables have been fixed, where we take \(\beta=0.6999\). In the specialised
case, assume that the root constraint graph has maximum degree at most \(6\)
and at least \((1-\eta)n\) variables of degree \(6\). Let \(T(m)\) denote the
maximum number of leaves produced when \(m\) units of fixed-variable budget
remain before the cutoff.

\begin{lemma}\label{lem:2sat-degree6-persistence}
Assume that the root constraint graph has at least \((1-\eta)n\) vertices of
degree \(6\) and maximum degree at most \(6\). Along any root-to-leaf path of
Algorithm~\textsc{2-SAT-Preprocess}, as long as fewer than
\((1-\eta)n/7\) variables have been fixed, the current residual still contains
a degree-\(6\) variable.
\end{lemma}
\begin{proof}
Let \(D\) be the set of vertices of degree \(6\) in the root constraint graph.
Call \(u\in D\) \emph{spoiled} once either \(u\) itself or one of its six root
neighbours has been fixed. Because every variable has root degree at most \(6\),
each fixed variable spoils at most \(7\) vertices of \(D\). Hence after \(s\)
fixed variables, at most \(7s\) vertices of \(D\) are spoiled. If
\(s<(1-\eta)n/7\), then some \(u\in D\) remains unspoiled.

For such a vertex \(u\), neither \(u\) nor any of its six root neighbours has
been assigned. In a \(2\)-CNF formula, assigning variables other than the two
endpoints of a clause cannot delete a binary clause witnessing an edge between
those endpoints. Under the residual convention used here, a binary clause is
removed or subsumed only when one of its endpoint variables is fixed, either
directly or by unit propagation. Therefore all six witness clauses showing the
root neighbours of \(u\) remain present in the current residual, and \(u\)
still has degree \(6\) in the current simple constraint graph.
\end{proof}

\begin{enumerate}[(a)]
\item \textbf{Branch on a degree-$6$ variable.}
If \(\deg(v)=6\), then the two branches have polarity-sensitive decreases
\((1+d_H^{+}(v),\,1+d_H^{-}(v))\). The worst split is the one-polarity case,
namely \((6,0)\) or \((0,6)\), which gives branching vector \((7,1)\); every
mixed polarity split is stronger. We therefore obtain the safe recurrence
\[
  T(m)\ \le\ T(m-7)+T(m-1),
\]
i.e.\ branching vector \((7,1)\).
Let $\alpha_6>1$ be the corresponding branching number, the unique real root of
$\alpha^7=\alpha^6+1$.
Numerically, $\log_2\alpha_6\approx 0.328173397$, and we use the safe bound
$\log_2\alpha_6\le 0.328174$.
By \Cref{lem:2sat-degree6-persistence}, degree-$6$ branching remains available
for the first \(\frac{(1-\eta)n}{7}\) units of fixed-variable budget along
every root-to-leaf path; if it remains available longer, the bound only
improves.

\item \textbf{Branch on a variable of degree at least $3$.}
For the remaining \(\beta n-\frac{(1-\eta)n}{7}\) units of fixed-variable
budget we branch on a variable of degree at least $3$.
This gives the recurrence
\[
  T(m)\ \le\ T(m-4)+T(m-1),
\]
i.e.\ branching vector $(4,1)$.
Let $\alpha_3>1$ be the corresponding branching number, the unique real root of
$\alpha^4=\alpha^3+1$.
Numerically, $\log_2\alpha_3\approx 0.464958417$, and we use the safe bound
$\log_2\alpha_3\le 0.464959$.
\end{enumerate}

Consequently, the preprocessing recursion tree has at most
\[
  L_{\mathrm{pre}}
  \;\le\;
  \alpha_6^{\frac{(1-\eta)n}{7}}\;
  \alpha_3^{\left(\beta-\frac{1-\eta}{7}\right)n}
  \;\le\;
  2^{\,n\left(\frac{1-\eta}{7}\log_2\alpha_6
      +\left(\beta-\frac{1-\eta}{7}\right)\log_2\alpha_3\right)}
  \;\le\;
  2^{\,n\left(0.305885+0.019541\,\eta\right)}.
\]
Here the last line is the direct expansion of
\[
  \frac{1-\eta}{7}\log_2\alpha_6
  +\left(0.6999-\frac{1-\eta}{7}\right)\log_2\alpha_3
\]
using the safe numerical bounds above.
In particular, if $\mathcal{L}$ denotes the set of hard leaves produced by preprocessing, then
$|\mathcal{L}|\le L_{\mathrm{pre}}$.
Moreover, every hard residual has at most \((1-\beta)n=0.3001n\) still-free variables.

\subsubsection{Counting the Easy and Hard Instances}
Before invoking \Cref{thm:combineuTotP}, discard every unsatisfiable hard
residual using the polynomial-time satisfiability test for \(2\)-CNF formulas;
this filtering is charged to the preprocessing term.  Let
\(\mathcal{L}^{+}=\{F_1,\dots,F_\ell\}\) be the remaining satisfiable hard
formulas.  Then \(\ell\le L_{\mathrm{pre}}\), and every \(F_i\in\mathcal L^+\)
has \(f_{\mathrm{2-SAT}}(F_i)>0\), as required by
\Cref{thm:combineuTotP}.  If \(\mathcal L^+=\emptyset\), the hard contribution
is zero.  Otherwise, since each \(F_i\) has at most \(0.3001n\) still-free
variables, we have
\[
  b_{\mathrm{2-SAT}}(F_i)\ \le\ 2^{0.3001n}.
\]
Therefore,
\[
  \sum_{i=1}^{\ell} b_{\mathrm{2-SAT}}(F_i)
  \;\le\;
  \ell\cdot 2^{0.3001n}
  \;\le\;
  2^{(0.605985+0.019541\,\eta)\,n}.
\]
Applying \Cref{thm:combineuTotP} yields an $(\varepsilon,\delta)$-approximation of
$\sum_{i=1}^{\ell} f_{\mathrm{2-SAT}}(F_i)$ in time
\[
  L_{\mathrm{pre}}+
  O^{\ast}\!\Bigl(
    2^{(0.302993+0.009771\,\eta)\,n}\;
    \varepsilon^{-2}\log\tfrac1\delta
  \Bigr).
\]
Since $0.302993+0.009771\,\eta < 0.305885+0.019541\,\eta$, the total cost in this regime is
dominated by constructing the preprocessing tree.  Thus (for fixed $(\varepsilon,\delta)$) we obtain
\[
  T_{\mathrm{spec}}(n)
  \;=\;
  O^{\ast}\!\bigl(2^{(0.305885+0.019541\,\eta)\,n}\bigr).
\]
The easy leaves in $\mathcal{E}$ have maximum degree $2$ and can be counted exactly in polynomial time.

\paragraph{Choice of $\eta^\star$ and global running time.}
We combine the specialised bound above with Wahlstr{\"o}m's analysis.
\Cref{tab:Wahlstrom} records the degree weights from Wahlstr{\"o}m's
measure that are used for this comparison.  Only \(w_5\) and \(w_6\) enter the
threshold calculation below, but the full row makes explicit which imported
measure is being invoked.

\begin{table}[H]
  \centering
  \caption{Wahlstr{\"o}m degree weights used in the complementary
  \#\textsc{2-SAT} regime.}
  \label{tab:Wahlstrom}
  \begin{tabular}{ccccc}
    \toprule
    $w_2$ & $w_3$ & $w_4$ & $w_5$ & $w_6$ \\
    \midrule
    0.115507 & 0.208788 & 0.277001 & 0.301245 & 0.307612 \\
    \bottomrule
  \end{tabular}
\end{table}

Let $w_6=0.307612$ and $w_5=0.301245$ denote Wahlstr{\"o}m's degree-based exponents.
Fix a threshold $\eta^\star\in(0,1)\).  After all degree-at-least-\(7\)
variables have been removed by the high-degree branching in the proof below, we
apply the following split to each residual with \(m\) still-free variables.  If
the simple constraint graph has at most \(\eta^\star m\) variables of degree
below \(6\), run the specialised routine with \(m\) in place of \(n\).  Otherwise
run Wahlstr{\"o}m's algorithm; since the residual now has maximum degree at most
\(6\), its degree measure is at most
\(((1-\eta^\star)w_6+\eta^\star w_5)m\).
The resulting worst-case exponent is
\[
  c(\eta^\star)
  \;=\;
  \max\Bigl\{
      0.305885+0.019541\,\eta^\star,\,
      (1-\eta^\star)w_6+\eta^\star w_5
  \Bigr\}.
\]
Choosing $\eta^\star$ so that the two terms are equal gives
\[
  \eta^\star
  \;=\;
  \frac{w_6-0.305885}{0.019541+(w_6-w_5)}
  \;\approx\;
  0.0667,
  \qquad
  c(\eta^\star)<0.30719.
\]

\begin{theorem}\label{thm:2-SAT}
For fixed $(\varepsilon,\delta)$, the above hybrid algorithm computes an
$(\varepsilon,\delta)$-approximation of \#\textsc{2-SAT} on $n$ variables in
\[
  T_{\mathrm{SAT}}(n)
  \;=\;
  O^{\ast}\!\bigl(2^{0.30719n}\bigr)
  \;=\;
  O^{\ast}\!\bigl(1.2373^{\,n}\bigr)
\]
time.
\end{theorem}
\begin{proof}
Let \(c=0.30719\).  We prove the running-time bound by induction on the
number \(m\) of still-free variables in the current residual.  If the residual
has maximum degree at most \(2\), it is counted exactly in polynomial time.

Suppose first that some variable \(v\) has degree \(d\ge 7\).  Branch on the two
truth assignments to \(v\).  As above, the branches fix at least
\(1+d_F^+(v)\) and \(1+d_F^-(v)\) variables, and
\(d_F^+(v)+d_F^-(v)=d\).  The worst polarity split is therefore the branching
vector \((1,8)\).  By the induction hypothesis the two recursive calls cost at
most
\[
  2^{c(m-1)}+2^{c(m-8)}
  =
  \bigl(2^{-c}+2^{-8c}\bigr)2^{cm}
  \le 2^{cm},
\]
since \(2^{-c}+2^{-8c}<1\).

It remains to consider residuals of maximum degree at most \(6\).  If at most
\(\eta^\star m\) variables have degree below \(6\), then
\Cref{lem:2sat-degree6-persistence} applies with \(\eta=\eta^\star\), and the
specialised preprocessing routine runs in
\[
  O^{\ast}\!\left(2^{(0.305885+0.019541\eta^\star)m}\right)
  \le O^{\ast}(2^{cm}).
\]
Otherwise, more than \(\eta^\star m\) variables have degree at most \(5\).
Wahlstr{\"o}m's degree measure is then at most
\[
  ((1-\eta^\star)w_6+\eta^\star w_5)m
  \le cm,
\]
so the exact-algorithm branch also costs \(O^{\ast}(2^{cm})\).  The two latter
exponents are equal, up to safe rounding, by the definition of \(\eta^\star\).
The failure probability is distributed over the high-degree branching leaves;
the resulting extra \(O(n)\) factor in the logarithmic confidence term is
absorbed by \(O^{\ast}(\cdot)\).
\end{proof}

For comparison, Wahlstr{\"o}m's universal bound gives
\[
  O^{\ast}(2^{w_6 n})
  =
  O^{\ast}(2^{0.307612n})
  =
  O^{\ast}(1.2377^{\,n}).
\]
Our hybrid strategy therefore reduces the variable-parameter worst-case base
from $1.2377$ to $1.2373$.

\begin{remark}
The significance of \Cref{thm:2-SAT} is qualitative rather than quantitative:
it gives an exponential-time approximate-counting upper bound for \#\textsc{2-SAT} that
improves on the best currently cited variable-parameter exact worst-case bound.
Recent clause-parameter algorithms for weighted \#\textsc{2-SAT} and
\#\textsc{3-SAT} \cite{peng2025newalgorithms} are orthogonal to this comparison:
their new bounds are measured by the number of clauses rather than by the number
of variables. They are nevertheless an interesting target for the same
philosophy. If a clause-parameter analysis exposes additive residual recursion,
polynomial-time zero tests, and a compatible bound on residual cores, then the
aggregation theorem can be used to ask for an approximation speed-up in that
parameter as well.
\end{remark}

\section{Quantum Extensions}\label{sec:Quantum}

The same bounded-uSR interface also exposes the tree and membership oracles used
by standard quantum approximate-counting routines.

In this section we show how to plug standard quantum approximate counting
(amplitude estimation / quantum counting~\cite{brassard2000quantum})
and the tree-size estimation algorithm of \citet{ambainis2017quantum}
into our interface.
We follow the oracle conventions and presentation of approximate counting
from \citet{AaronsonR20}.

In doing so we obtain two black-box quantum speed-ups that apply to every problem in uTotP.
As with the classical framework, the running time depends on the specific bounded uSR form chosen for the problem.

\paragraph{Oracle model.}
Throughout the quantum sections we work in the standard quantum query model.
The predicate
\[
\mathsf{Oracle}(m) \;:=\; t\bigl(h(x,\sigma_{\mathcal F^b,x}(m))\bigr)
\]
is computable in time $\operatorname{poly}(|x|)$ and is therefore implemented as a
reversible circuit and used as a unitary oracle.
One quantum query corresponds to one invocation of this unitary; all polynomial
overhead is suppressed in the $O^\ast(\cdot)$ notation.

\paragraph{On oracle access.}
It is important to note that reversibility alone is insufficient for the $b(x)^{1/4}$ speed-up.
A generic reversible implementation of a classical branching or enumeration algorithm does not, in general, expose the rooted-tree oracle interface required by Ambainis--Kokainis, and therefore cannot be used as a black-box substitute.
Our displacement-index construction ensures that this interface is available uniformly for every bounded-uSR definition.

\subsection{A $b(x)^{1/3}$ Quantum Speed-up}\label{sec:quantum-aaronson}

\paragraph{Quantum approximate counting (amplitude estimation / quantum counting).}
Fix an input $x$ and set $B := b(x)$.
We consider the standard \emph{marked-item counting} primitive: given oracle access to a Boolean
predicate $\chi:[B]\to\{0,1\}$ with $M := |\chi^{-1}(1)|$, quantum approximate counting via
amplitude estimation / quantum counting~\cite{brassard2000quantum} outputs an $\varepsilon$-relative
estimate of $M$ with failure probability at most $\delta$ using
\[
  O^{\ast}\!\left(\sqrt{\frac{B}{M}}\cdot \log(\delta^{-1}) \cdot \frac{1}{\varepsilon}\right)
\]
oracle queries. We follow the oracle conventions and presentation of this primitive in
\citet{AaronsonR20}.

In our setting, the relevant membership predicate over $[B]=[b(x)]$ is
\[
  \mathsf{Oracle}(m)\;:=\; t\!\left(h\!\left(x,\sigma_{\mathcal F^b,x}(m)\right)\right),
\]
so the number of marked items is exactly
\[
  M \;=\; \sum_{m=1}^{B} \mathsf{Oracle}(m) \;=\; f(x).
\]
The only requirement for applying approximate counting is that $\mathsf{Oracle}(m)$ be computable
in time $\poly(|x|)$ (so that it can be implemented as a unitary in the quantum query model).
This holds here: by \Cref{lem:sampleRuntime} the map
$\sigma_{\mathcal{F}^b,x}:[b(x)]\rightarrow\mathcal{G}_{\mathcal{F},x}$ runs in polynomial time, and
by definition $h:\Sigma^{\ast}\times\mathbb{N}^d\rightarrow \Sigma^{\ast}$ and
$t:\Sigma^{\ast}\rightarrow \{0,1\}$ run in polynomial time.
Hence each oracle query has only $\poly(|x|)$ overhead, which is suppressed by $O^{\ast}(\cdot)$.

Finally, to obtain a uniform bound in terms of $b(x)$ alone, run the polynomial-delay enumerator
\textsc{DFSEnum}$_{\mathcal{F},x}$ until it either terminates (in which case we have computed $f(x)$
exactly) or outputs $b(x)^{1/3}$ positive leaves. In the latter case we certify $f(x)=M\ge b(x)^{1/3}$,
and the approximate counting routine therefore takes
\[
\boxed{
O^{\ast}\Bigl(\sqrt{{b(x)}/{b(x)^{\frac13}}}\cdot \log(\delta^{-1}) \cdot \frac{1}{\varepsilon}\Bigr)=
O^{\ast}\Bigl(b(x)^{\frac13}\cdot \log(\delta^{-1}) \cdot \frac{1}{\varepsilon}\Bigr)
}
\]

\subsection{A $b(x)^{1/4}$ \,Quantum Speed-up}\label{sec:Ambainis}

\paragraph{Tree size being estimated.}
If \(\mathsf{Dec}_f(x)=\textsc{false}\), then \(f(x)=0\) and we return \(0\)
before invoking the tree-size estimator.  Hence in the sequel we assume that
at least one positive leaf exists, so the displacement-index tree is rooted at
\((0,\ldots,0)\).
In our setting, the Ambainis--Kokainis estimator is applied to the
displacement-index tree \(\mathcal T_{\mathcal F,x}\)
from \Cref{def:indexTree}, whose nodes are displacement indices
reachable via the $\mathsf{Child}$ oracle.
Here \(\mathcal T_{\mathcal F,x}\) denotes the displacement-index tree
constructed in \Cref{sec:PolyDelay}, not the original uSR recursion tree.  Its
nodes are precisely the displacement indices emitted by \textsc{DFSEnum}; by
\Cref{thm:dfs-correct}, these are in bijection with the positive leaves of the
underlying uSR recursion.
Consequently, the size $T$ of $\mathcal T_{\mathcal F,x}$ is exactly
$f(x)$, the number of solutions we wish to approximate.

\paragraph{Oracle requirements (Ambainis--Kokainis).}
Ambainis \& Kokainis~\cite{ambainis2017quantum} give a quantum
\emph{tree--size estimator} that, given a proposed upper bound
$T_0\ge T$ on the number of nodes $T$ in a rooted tree
and error parameters $\varepsilon,\delta\in(0,1)$, runs in
\[
  O^{\ast}\!\bigl(\sqrt{T_0}\,\log^{2}\!\tfrac1\delta\,\varepsilon^{-3/2}\bigr)
\]
and, with probability at least $1-\delta$, either
\begin{enumerate}[(a)]
  \item returns $\hat T$ with
        $\bigl|\hat T-T\bigr|\le\varepsilon T$, or
  \item certifies that $T>(1+\varepsilon)^{-1}T_0$.
\end{enumerate}
The algorithm needs an \emph{oracle} that, in
$\poly(|x|)$ time,
\begin{enumerate}[(i)]
  \item returns the root,
  \item gives the degree $d(v)$ of any node $v$ (bounded by $\poly(|x|)$),
  \item returns the $i$-th child of $v$ for any $i\in[d(v)]$, and
  \item (optionally) returns the parent of $v$.
\end{enumerate}

\paragraph{Our uSR search tree satisfies the oracle.}
For every bounded uSR form
$\mathcal{F}^b=(f,h,t,r,q,b)$ the depth-first, the combination of our oracles $\mathsf{Child}$ and $\mathsf{Parent}$ (\Cref{sec:GraphOracles}) ensures:
\begin{enumerate}[(i)]
  \item a root obtained in $\operatorname{poly}(|x|)$ time;
  \item node degree $d(v)\le\poly(|x|)$ (every node has at most $q(x)$ children) can be found in polytime by finding each indexed child (\Cref{cond:findChild});
  \item return the $i$-th child using the $\mathsf{Child}$ oracle in polynomial time (\Cref{lem:polyChild})\label{cond:findChild}
    \item parent obtained in polytime by the $\mathsf{Parent}$ oracle (\Cref{lem:polyParent})
\end{enumerate}
Hence the tree can be fed, black-box, into the
Ambainis--Kokainis estimator.

\paragraph{Balancing Ambainis-Kokainis with amplitude estimation.}
We first call the Ambainis-Kokainis quantum estimator with
\(
   T_0:=\bigl\lceil b(x)^{1/2}\bigr\rceil
\)
and the oracle above, using failure probability \(\delta/2\).
If the second stage is needed, we run amplitude estimation with failure
probability \(\delta/2\); a union bound gives total failure probability at most
\(\delta\).
There are two cases:

\smallskip
\begin{enumerate}[(1)]
  \item \emph{$T\le T_0$.}
        The algorithm outputs an $(\varepsilon,\delta)$ estimate
        $\hat T$ in
        $O^{\ast}\!\bigl(b(x)^{1/4}
            \log^{2}\!\tfrac1\delta\,\varepsilon^{-3/2}\bigr)$
        time, which we return as $\hat f$.

  \item \emph{$T>(1+\varepsilon)^{-1}T_0$.} In this case the algorithm runs in  $O^{\ast}\!\bigl(b(x)^{1/4}
            \log^{2}\!\tfrac1\delta\,\varepsilon^{-3/2}\bigr)$ and outputs the claim ``$T$ contains more than $T_0$ vertices''.
        From that we can infer $f(x)=T>\frac{1}{1+\varepsilon}b(x)^{1/2}$, so the success
        probability of the \textit{sampling} oracle is
        $p=\frac{f(x)}{b(x)}>\frac{1}{(1+\varepsilon)\,b(x)^{1/2}}$.
        Applying the amplitude estimation step from
        \Cref{sec:quantum-aaronson} now costs
        \[
           O^{\ast}\!\bigl(b(x)^{1/4}
               \log\tfrac1\delta\,\varepsilon^{-1}\bigr),
        \]
        dominated (up to poly-log factors) by case \textbf{(1)}.
\end{enumerate}

\paragraph{Resulting running time.}
Combining the \(T_0=\lceil b(x)^{1/2}\rceil\) Ambainis--Kokainis estimation potentially followed by standard quantum amplitude estimation
gives an overall
\[
\boxed{
  O^{\ast}\!\Bigl(
      b(x)^{1/4}\,
      \log^{2}\!\tfrac1\delta\,
      \varepsilon^{-3/2}
  \Bigr)
  }
\]
time $(\varepsilon,\delta)$-approximate counter for \emph{every} bounded
uSR function~$\mathcal{F}^b$.

\paragraph{Example: Maximal Clique.} We can directly apply this result to the maximal-clique bounded uSR form from \Cref{sec:maxclique} and automatically get the following running times.

\begin{center}
\renewcommand{\arraystretch}{1.1}
\begin{tabular}{@{}lccc@{}}
\toprule
Problem & classical+classical & classical+quantum & quantum+quantum  \\
\midrule
Generic & \(O^{\ast}\bigl(b(x)^{1/2}\varepsilon^{-2} \log{(\delta^{-1})}\bigr)\)(\ref{thm:fullalgo}) & \(O^{\ast}\bigl(b(x)^{1/3}\varepsilon^{-1} \log{(\delta^{-1})}\bigr)\) & \(O^{\ast}\bigl(b(x)^{1/4}\varepsilon^{-3/2} \log^2{(\delta^{-1})}\bigr)\) \\
\#\textsc{MC} & \(O^{\ast}\!\left(3^{n/6} \varepsilon^{-2} \log{(\delta^{-1})}\right)\)  (\ref{thm:MC}) & \(O^{\ast}\!\left(3^{n/9} \varepsilon^{-1} \log{(\delta^{-1})}\right)\)    & \(O^{\ast}\!\left(3^{n/12} \varepsilon^{-3/2} \log^2{(\delta^{-1})}\right)\)   \\

\bottomrule
\end{tabular}
\end{center}

\medskip
\noindent
\textbf{Take-away.}
Where the classical framework yields a generic
$b(x)^{1/2}$ speed-up, combining the Ambainis--Kokainis and standard quantum approximate counting routines
automatically improves this to a \emph{$b(x)^{1/4}$} dependence.
Compared to simply combining standard quantum approximate counting with the classical delay algorithm, we trade a cost of $\log(1/\delta)\,\varepsilon^{-1/2}$ on the error terms for an improvement from $b(x)^{\frac13}$ to $b(x)^{\frac14}$ on the running time dependence on the bounding function.

\section{Conclusion}

Unweighted self-reducibility (uSR) offers a useful point of contact between
combinatorial structure and algorithmic speed. When each recursive call
contributes the sum of its children, a counting problem becomes a search tree
whose relevant size is controlled by an explicit bound \(b(x)\). We proved that
this restriction does not shrink the class of problems at the level of Karp
closure: the resulting class \textsf{uTotP} coincides with the standard
\textsf{TotP}.

The algorithmic contribution is that such bounds can be used before the whole
tree is explored. For one bounded uSR form, the estimator enumerates only up to
\(\sqrt{b(x)}\) positive leaves and then samples from the remaining bound mass.
For a decomposition into many hard cores, the estimator takes the square root
after summing the core bounds. This is what turns recurrence or extremal upper
bounds into approximate counters for the applications above: direct bounds give
the clique, minimal-separator, and subcubic-perfect-matching rows, while
easy-plus-hard decompositions give the Independent Set and \#\textsc{2-SAT}
improvements.

The quantum consequences use the same interface. Standard quantum approximate
counting \cite{brassard2000quantum,AaronsonR20} and tree-size estimation
\cite{ambainis2017quantum} reduce the dependence on \(b(x)\) to \(b(x)^{1/3}\)
and \(b(x)^{1/4}\), respectively, whenever the corresponding bounded-uSR
oracles are available. These are consequences of exposing the recursion tree;
the headline classical decompositions remain governed by the preprocessing and
aggregation balances proved in the application sections.

This constructive viewpoint complements recent work relating approximate
counting to decision complexity. Dell and Lapinskas showed, under ETH and SETH,
that significant improvements for approximate counting would transfer to
decision~\cite{dell2021fine}. Censor-Hillel, Even, and Williams established
output-sensitive reductions aligning the fine-grained complexity of approximate
counting with decision for several problems~\cite{censor2025output}. Our
results go in the other direction: given a feasible self-reduction and a
compatible bound, they provide an explicit approximate counter.

The remaining structural question is to understand which enumeration problems
admit effective bounded-uSR forms, and which decompositions preserve enough
local structure for the aggregate theorem to apply. The applications here show
that the answer is already broad enough to improve several well-studied
counting problems; sharper problem-specific bounds or new clause-parameter
recurrences would feed directly into the same framework when they expose the
required residual recursion.

\bibliographystyle{plainnat}
\bibliography{references}

@article{goldberg2021faster,
  title={Faster exponential-time algorithms for approximately counting independent sets},
  author={Goldberg, Leslie Ann and Lapinskas, John and Richerby, David},
  journal={Theoretical Computer Science},
  volume={892},
  pages={48--84},
  year={2021},
  publisher={Elsevier}
}

@article{moon1965cliques,
  title={On cliques in graphs},
  author={Moon, John W and Moser, Leo},
  journal={Israel journal of Mathematics},
  volume={3},
  pages={23--28},
  year={1965},
  publisher={Springer}
}

@article{goldberg2016approximately,
  title={Approximately counting locally-optimal structures},
  author={Goldberg, Leslie Ann and Gysel, Rob and Lapinskas, John},
  journal={Journal of Computer and System Sciences},
  volume={82},
  number={6},
  pages={1144--1160},
  year={2016},
  publisher={Elsevier}
}

@article{valiant1979complexity,
  title={The complexity of enumeration and reliability problems},
  author={Valiant, Leslie G},
  journal={siam Journal on Computing},
  volume={8},
  number={3},
  pages={410--421},
  year={1979},
  publisher={SIAM}
}

@article{bron1973finding,
  title={Finding all cliques of an undirected graph (algorithm 457)},
  author={Bron, Coenraad and Kerbosch, Joep},
  journal={Commun. ACM},
  volume={16},
  number={9},
  pages={575--576},
  year={1973}
}

@inproceedings{makino2004new,
  title={New algorithms for enumerating all maximal cliques},
  author={Makino, Kazuhisa and Uno, Takeaki},
  booktitle={Algorithm Theory-SWAT 2004: 9th Scandinavian Workshop on Algorithm Theory, Humleb{\ae}k, Denmark, July 8-10, 2004. Proceedings 9},
  pages={260--272},
  year={2004},
  organization={Springer}
}

@article{dyer2004relative,
  title={The relative complexity of approximate counting problems},
  author={Dyer, Martin and Goldberg, Leslie Ann and Greenhill, Catherine and Jerrum, Mark},
  journal={Algorithmica},
  volume={38},
  pages={471--500},
  year={2004},
  publisher={Springer}
}

@article{tsukiyama1977new,
  title={A new algorithm for generating all the maximal independent sets},
  author={Tsukiyama, Shuji and Ide, Mikio and Ariyoshi, Hiromu and Shirakawa, Isao},
  journal={SIAM Journal on Computing},
  volume={6},
  number={3},
  pages={505--517},
  year={1977},
  publisher={SIAM}
}

@inproceedings{ambainis2017quantum,
  title={Quantum algorithm for tree size estimation, with applications to backtracking and 2-player games},
  author={Ambainis, Andris and Kokainis, Martins},
  booktitle={Proceedings of the 49th Annual ACM SIGACT Symposium on Theory of Computing},
  pages={989--1002},
  year={2017}
}

@inproceedings{AaronsonR20,
  author       = {Scott Aaronson and
                  Patrick Rall},
  editor       = {Martin Farach{-}Colton and
                  Inge Li G{\o}rtz},
  title        = {Quantum Approximate Counting, Simplified},
  booktitle    = {3rd Symposium on Simplicity in Algorithms, {SOSA} 2020, Salt Lake
                  City, UT, USA, January 6-7, 2020},
  pages        = {24--32},
  publisher    = {{SIAM}},
  year         = {2020},
  url          = {https://doi.org/10.1137/1.9781611976014.5},
  doi          = {10.1137/1.9781611976014.5},
  bibsource    = {dblp computer science bibliography, https://dblp.org}
}

@article{GaspersLee23,
  author       = {Serge Gaspers and
                  Edward J. Lee},
  title        = {Faster Graph Coloring in Polynomial Space},
  journal      = {Algorithmica},
  volume       = {85},
  number       = {2},
  pages        = {584--609},
  year         = {2023},
  url          = {https://doi.org/10.1007/s00453-022-01034-7},
  doi          = {10.1007/S00453-022-01034-7},
  bibsource    = {dblp computer science bibliography, https://dblp.org}
}

@inproceedings{pagourtzis2006complexity,
  title={The complexity of counting functions with easy decision version},
  author={Pagourtzis, Aris and Zachos, Stathis},
  booktitle={International Symposium on Mathematical Foundations of Computer Science},
  pages={741--752},
  year={2006},
  organization={Springer}
}

@article{bulatov2013expressibility,
  title={The expressibility of functions on the boolean domain, with applications to counting {CSP}s},
  author={Bulatov, Andrei A and Dyer, Martin and Goldberg, Leslie Ann and Jerrum, Mark and McQuillan, Colin},
  journal={Journal of the ACM (JACM)},
  volume={60},
  number={5},
  pages={1--36},
  year={2013},
  publisher={ACM New York, NY, USA}
}

@article{dahllof2005counting,
  title={Counting models for {2SAT} and {3SAT} formulae},
  author={Dahll{\"o}f, Vilhelm and Jonsson, Peter and Wahlstr{\"o}m, Magnus},
  journal={Theoretical Computer Science},
  volume={332},
  number={1-3},
  pages={265--291},
  year={2005},
  publisher={Elsevier}
}

@inproceedings{ge2018new,
  title={A new probabilistic algorithm for approximate model counting},
  author={Ge, Cunjing and Ma, Feifei and Liu, Tian and Zhang, Jian and Ma, Xutong},
  booktitle={Automated Reasoning: 9th International Joint Conference, IJCAR 2018, Held as Part of the Federated Logic Conference, FloC 2018, Oxford, UK, July 14-17, 2018, Proceedings 9},
  pages={312--328},
  year={2018},
  organization={Springer}
}

@inproceedings{peng2025newalgorithms,
  title={New Algorithms for \#2-{SAT} and \#3-{SAT}},
  author={Peng, Junqiang and Sheng, Zimo and Xiao, Mingyu},
  booktitle={Proceedings of the Thirty-Fourth International Joint Conference on Artificial Intelligence},
  pages={2666--2674},
  year={2025},
  doi={10.24963/ijcai.2025/297}
}

@inproceedings{wahlstrom2008tighter,
  title={A tighter bound for counting max-weight solutions to {2SAT} instances},
  author={Wahlstr{\"o}m, Magnus},
  booktitle={Parameterized and Exact Computation: Third International Workshop, IWPEC 2008, Victoria, Canada, May 14-16, 2008. Proceedings 3},
  pages={202--213},
  year={2008},
  organization={Springer}
}

@article{jerrum1986random,
  title={Random generation of combinatorial structures from a uniform distribution},
  author={Jerrum, Mark R and Valiant, Leslie G and Vazirani, Vijay V},
  journal={Theoretical computer science},
  volume={43},
  pages={169--188},
  year={1986},
  publisher={Elsevier}
}

@article{karp1989monte,
  title={{M}onte-{C}arlo approximation algorithms for enumeration problems},
  author={Karp, Richard M and Luby, Michael and Madras, Neal},
  journal={Journal of algorithms},
  volume={10},
  number={3},
  pages={429--448},
  year={1989},
  publisher={Elsevier}
}

@article{alon1997finding,
  title={Finding and counting given length cycles},
  author={Alon, Noga and Yuster, Raphael and Zwick, Uri},
  journal={Algorithmica},
  volume={17},
  number={3},
  pages={209--223},
  year={1997},
  publisher={Springer}
}

@inproceedings{vassilevska2009finding,
  title={Finding, minimizing, and counting weighted subgraphs},
  author={Vassilevska, Virginia and Williams, Ryan},
  booktitle={Proceedings of the forty-first annual ACM symposium on Theory of computing},
  pages={455--464},
  year={2009}
}

@article{goldberg2012approximating,
  title={Approximating the partition function of the ferromagnetic {P}otts model},
  author={Goldberg, Leslie Ann and Jerrum, Mark},
  journal={Journal of the ACM (JACM)},
  volume={59},
  number={5},
  pages={1--31},
  year={2012},
  publisher={ACM New York, NY, USA}
}

@article{SinclairJ89,
  author       = {Alistair Sinclair and
                  Mark Jerrum},
  title        = {Approximate Counting, Uniform Generation and Rapidly Mixing {M}arkov
                  Chains},
  journal      = {Information and Computation},
  volume       = {82},
  number       = {1},
  pages        = {93--133},
  year         = {1989}
}

@article{jerrum1993polynomial,
  title={Polynomial-time approximation algorithms for the {I}sing model},
  author={Jerrum, Mark and Sinclair, Alistair},
  journal={SIAM Journal on computing},
  volume={22},
  number={5},
  pages={1087--1116},
  year={1993},
  publisher={SIAM}
}

@article{lovasz1993random,
  title={Random walks in a convex body and an improved volume algorithm},
  author={Lov{\'a}sz, L{\'a}szl{\'o} and Simonovits, Mikl{\'o}s},
  journal={Random structures \& algorithms},
  volume={4},
  number={4},
  pages={359--412},
  year={1993},
  publisher={Wiley Online Library}
}

@article{galanis2016approximately,
  title={Approximately Counting {$H$}-Colorings is \#{BIS}-Hard},
  author={Galanis, Andreas and Goldberg, Leslie Ann and Jerrum, Mark},
  journal={SIAM Journal on Computing},
  volume={45},
  number={3},
  pages={680--711},
  year={2016},
  publisher={SIAM}
}

@inproceedings{karp1983monte,
  title={{M}onte-{C}arlo algorithms for enumeration and reliability problems},
  author={Karp, Richard M and Luby, Michael},
  booktitle={24th Annual Symposium on Foundations of Computer Science (sfcs 1983)},
  pages={56--64},
  year={1983},
  organization={IEEE Computer Society}
}

@article{dyer1991random,
  title={A random polynomial-time algorithm for approximating the volume of convex bodies},
  author={Dyer, Martin and Frieze, Alan and Kannan, Ravi},
  journal={Journal of the ACM (JACM)},
  volume={38},
  number={1},
  pages={1--17},
  year={1991},
  publisher={ACM New York, NY, USA}
}

@article{ibarra1975fast,
  title={Fast approximation algorithms for the knapsack and sum of subset problems},
  author={Ibarra, Oscar H and Kim, Chul E},
  journal={Journal of the ACM (JACM)},
  volume={22},
  number={4},
  pages={463--468},
  year={1975},
  publisher={ACM New York, NY, USA}
}

@article{lawler1982fully,
  title={A fully polynomial approximation scheme for the total tardiness problem},
  author={Lawler, Eugene L},
  journal={Operations Research Letters},
  volume={1},
  number={6},
  pages={207--208},
  year={1982},
  publisher={Elsevier}
}

@article{jerrum2004polynomial,
  title={A polynomial-time approximation algorithm for the permanent of a matrix with nonnegative entries},
  author={Jerrum, Mark and Sinclair, Alistair and Vigoda, Eric},
  journal={Journal of the ACM (JACM)},
  volume={51},
  number={4},
  pages={671--697},
  year={2004},
  publisher={ACM New York, NY, USA}
}

@article{schmitt2013exploiting,
  title={Exploiting independent subformulas: A faster approximation scheme for \#$k$-{SAT}},
  author={Schmitt, Manuel and Wanka, Rolf},
  journal={Information Processing Letters},
  volume={113},
  number={9},
  pages={337--344},
  year={2013},
  publisher={Elsevier}
}

@inproceedings{cardinal2018solving,
  title={Solving and Sampling with Many Solutions: Satisfiability and Other Hard Problems},
  author={Cardinal, Jean and Nummenpalo, Jerri and Welzl, Emo},
  booktitle={12th International Symposium on Parameterized and Exact Computation},
  year={2018}
}

@article{thurley2011approximation,
  title={An approximation algorithm for \#$k$-{SAT}},
  author={Thurley, Marc},
  journal={arXiv preprint arXiv:1107.2001},
  year={2011}
}

@article{bakali2016self,
  title={Self-reducible with easy decision version counting problems admit additive error approximation. {C}onnections to counting complexity, exponential time complexity, and circuit lower bounds},
  author={Bakali, Eleni},
  journal={arXiv preprint arXiv:1611.01706},
  year={2016}
}

@article{antonopoulos2022completeness,
  title={Completeness, approximability and exponential time results for counting problems with easy decision version},
  author={Antonopoulos, Antonis and Bakali, Eleni and Chalki, Aggeliki and Pagourtzis, Aris and Pantavos, Petros and Zachos, Stathis},
  journal={Theoretical Computer Science},
  volume={915},
  pages={55--73},
  year={2022},
  publisher={Elsevier}
}

@article{HujterT93,
  author       = {Mih{\'{a}}ly Hujter and
                  Zsolt Tuza},
  title        = {The Number of Maximal Independent Sets in Triangle-Free Graphs},
  journal      = {{SIAM} Journal on Discrete Mathematics},
  volume       = {6},
  number       = {2},
  pages        = {284--288},
  year         = {1993},
  doi          = {10.1137/0406022}
}

@inproceedings{bakali2024power,
  title={On the power of counting the total number of computation paths of {NPTMs}},
  author={Bakali, Eleni and Chalki, Aggeliki and Kanellopoulos, Sotiris and Pagourtzis, Aris and Zachos, Stathis},
  booktitle={Annual Conference on Theory and Applications of Models of Computation},
  pages={209--220},
  year={2024},
  organization={Springer}
}

@inproceedings{bakali2017completeness,
  title={Completeness results for counting problems with easy decision},
  author={Bakali, Eleni and Chalki, Aggeliki and Pagourtzis, Aris and Pantavos, Petros and Zachos, Stathis},
  booktitle={International Conference on Algorithms and Complexity},
  pages={55--66},
  year={2017},
  organization={Springer}
}

@book{mitzenmacher2017probability,
  title={Probability and Computing: Randomization and Probabilistic Techniques in Algorithms and Data Analysis},
  author={Mitzenmacher, M. and Upfal, E.},
  isbn={9781107154889},
  lccn={2016041654},
  series={Probability and Computing: Randomization and Probabilistic Techniques in Algorithms and Data Analysis},
  url={https://books.google.com/books?id=E9UlDwAAQBAJ},
  year={2017},
  publisher={Cambridge University Press}
}

@article{sinclair2017spatial,
  title={Spatial mixing and the connective constant: Optimal bounds},
  author={Sinclair, Alistair and Srivastava, Piyush and {\v{S}}tefankovi{\v{c}}, Daniel and Yin, Yitong},
  journal={Probability Theory and Related Fields},
  volume={168},
  number={1},
  pages={153--197},
  year={2017},
  publisher={Springer}
}

@inproceedings{censor2025output,
  title={Output-Sensitive Approximate Counting via a Measure-Bounded Hyperedge Oracle, or: How Asymmetry Helps Estimate-Clique Counts Faster},
  author={Censor-Hillel, Keren and Even, Tomer and Vassilevska Williams, Virginia},
  booktitle={Proceedings of the 57th Annual ACM Symposium on Theory of Computing},
  pages={1985--1996},
  year={2025}
}

@article{dell2021fine,
  title={Fine-grained reductions from approximate counting to decision},
  author={Dell, Holger and Lapinskas, John},
  journal={ACM Transactions on Computation Theory (TOCT)},
  volume={13},
  number={2},
  pages={1--24},
  year={2021},
  publisher={ACM New York, NY, USA}
}

@misc{gaspers2015numberminimalseparatorsgraphs,
      title={On the Number of Minimal Separators in Graphs},
      author={Serge Gaspers and Simon Mackenzie},
      year={2015},
      eprint={1503.01203},
      archivePrefix={arXiv},
      primaryClass={cs.DS},
      url={https://arxiv.org/abs/1503.01203},
}

@article{berry2000generating,
  title={Generating all the minimal separators of a graph},
  author={Berry, Anne and Bordat, Jean-Paul and Cogis, Olivier},
  journal={International Journal of Foundations of Computer Science},
  volume={11},
  number={03},
  pages={397--403},
  year={2000},
  publisher={World Scientific}
}

@article{takata2010space,
  title={Space-optimal, backtracking algorithms to list the minimal vertex separators of a graph},
  author={Takata, Ken},
  journal={Discrete Applied Mathematics},
  volume={158},
  number={15},
  pages={1660--1667},
  year={2010},
  publisher={Elsevier}
}

@inproceedings{furer2012counting,
  title={Counting perfect matchings in graphs of degree 3},
  author={F{\"u}rer, Martin},
  booktitle={International Conference on Fun with Algorithms},
  pages={189--197},
  year={2012},
  organization={Springer}
}

@article{kenig2023listing,
  title={Listing Small Minimal $ s, t $-separators in FPT-Delay},
  author={Kenig, Batya},
  journal={arXiv preprint arXiv:2307.00604},
  year={2023}
}

@article{brassard2000quantum,
  title={Quantum amplitude amplification and estimation},
  author={Brassard, Gilles and Hoyer, Peter and Mosca, Michele and Tapp, Alain},
  journal={arXiv preprint quant-ph/0005055},
  year={2000}
}

@inproceedings{weitz2006counting,
  title={Counting independent sets up to the tree threshold},
  author={Weitz, Dror},
  booktitle={Proceedings of the thirty-eighth annual ACM symposium on Theory of computing},
  pages={140--149},
  year={2006}
}

\appendix

\section{Reduction From TotP to uTotP}\label{sec:DeferredProofs}

\TOTPUTOTP*

\begin{proof}

We first note that $uTotP \subseteq TotP$: every unweighted self-reducible definition is a special
case of self-reducibility (take all weights $g(\cdot,\cdot)=1$), hence $\#PE_{uSR}\subseteq \#PE_{SR}$,
and taking Karp closures preserves inclusion.

For the reverse inclusion, let $f\in TotP$. By definition of Karp closure there exists
$f_0\in \#PE_{SR}$ such that $f \leq_m^P f_0$.
By transitivity of Karp reductions, it suffices to show that every $f\in \#PE_{SR}$
parsimoniously Karp-reduces to some $f'\in \#PE_{uSR}$.

Fix $f\in \#PE_{SR}$ with self-reducible presentation
\[
f(x)=t(x)+\sum_{i=1}^{r(x)} g(x,i)\cdot f(h(x,i)).
\]
The reduction below implements the standard idea ``carry weights in the instance'':
each weighted term $g(x,i)\cdot f(h(x,i))$ is unfolded into an unweighted binary tree producing
exactly $g(x,i)$ copies of the subcomputation, and the additive term $t(x)$ is handled by an extra
leaf child.

\paragraph{Tagged instances and reduction map.}
We define a new function $f':\Sigma^*\to \mathbb{N}$ on three kinds of tagged instances:
\[
\mathrm{Main}(x):=\langle \mathsf{M},x\rangle,\qquad
\mathrm{Mult}(m,x):=\langle \mathsf{W},m,x\rangle,\qquad
\mathrm{Bit}(b):=\langle \mathsf{B},b\rangle\ \ (b\in\{0,1\}).
\]
(These tags are constant-size symbols; encoding/parsing is as in Appendix~A.)

The parsimonious many-one reduction is
\[
\phi(x):=\mathrm{Main}(x)=\langle \mathsf{M},x\rangle.
\]
We now give an unweighted self-reducible definition $(f',h',t',r',q')$ such that
$f'(\phi(x))=f(x)$ for all $x$.

\paragraph{The unweighted self-reduction.}
We define $t',r':\Sigma^*\to \mathbb{N}$ and $h':\Sigma^*\times \mathbb{N}_+\to \Sigma^*$ by cases.

\emph{(1) Constant leaves.} For $b\in\{0,1\}$, $\mathrm{Bit}(b)$ is a leaf:
\[
r'(\mathrm{Bit}(b)):=0,\qquad t'(\mathrm{Bit}(b)):=b,\qquad h'(\mathrm{Bit}(b),i):=\lambda\ \ \forall i.
\]

\emph{(2) Copy gadget.} For $\mathrm{Mult}(m,x)$ we unfold $m$ copies by repeated halving:
\[
t'(\mathrm{Mult}(m,x)):=0,
\qquad
r'(\mathrm{Mult}(m,x)) :=
\begin{cases}
0 & m=0,\\
1 & m=1,\\
2 & m\ge 2,
\end{cases}
\]
and for the children,
\[
h'(\mathrm{Mult}(m,x),i):=
\begin{cases}
\mathrm{Main}(x) & (m=1 \text{ and } i=1),\\
\mathrm{Mult}(\lceil m/2\rceil,x) & (m\ge 2 \text{ and } i=1),\\
\mathrm{Mult}(\lfloor m/2\rfloor,x) & (m\ge 2 \text{ and } i=2),\\
\lambda & \text{otherwise.}
\end{cases}
\]

\emph{(3) Main instances.} For $\mathrm{Main}(x)$ we create one child for each weighted term,
plus one extra constant leaf for $t(x)$:
\[
t'(\mathrm{Main}(x)):=0,\qquad r'(\mathrm{Main}(x)):=r(x)+1,
\]
and for $j\in\mathbb{N}_+$,
\[
h'(\mathrm{Main}(x),j):=
\begin{cases}
\mathrm{Mult}(g(x,j),\,h(x,j)) & 1\le j\le r(x),\\
\mathrm{Bit}(t(x)) & j=r(x)+1,\\
\lambda & j>r(x)+1.
\end{cases}
\]

\emph{Depth bound.} Let $C$ be a sufficiently large fixed constant and set $q'(z):=|z|^C$.
(Existence of such a polynomial bound follows from: the original recursion has depth $\le q(x)$,
and each halving gadget has depth $O(\log m)$, where all numbers $m$ we manipulate have
polynomial bitlength in the current instance size.)

Finally, define $f'$ by the unweighted recursion
\[
f'(z)=t'(z)+\sum_{j=1}^{r'(z)} f'\big(h'(z,j)\big).
\]

\paragraph{Validity of the uSR definition and $\#PE$ property.}
All maps $h',t',r',q'$ are polynomial-time computable (they apply $t,r,h,g$ a constant number of
times and perform basic arithmetic on $m$).
Moreover, the definition is unweighted: if $r'(z)>0$ then $t'(z)=0$, and leaves are exactly the
$\mathrm{Bit}(b)$ nodes (and the $m=0$ case of $\mathrm{Mult}$), which contribute $0$ or $1$.

The fan-out is polynomially bounded: $\mathrm{Main}(x)$ has $r(x)+1\le |x|^{O(1)}$ children, and
$\mathrm{Mult}$ nodes have fan-out at most $2$.

The recursion has polynomial depth: the original recursion reaches a base case within $q(x)$ levels,
and each $\mathrm{Mult}(m,\cdot)$ gadget has depth $O(\log m)$ by repeated halving. Since all numbers
$m=g(y,i)$ computed along the recursion have polynomial bitlength in the current instance size,
a global polynomial depth bound $q'(z)=|z|^C$ (for a sufficiently large constant $C$) suffices.

Finally, $f'\in \#PE$: deciding $f'(z)>0$ is in polynomial time by case analysis on the tag:
$f'(\mathrm{Bit}(b))>0 \Leftrightarrow b=1$,
$f'(\mathrm{Mult}(m,x))>0 \Leftrightarrow (m>0)\wedge (f(x)>0)$,
and $f'(\mathrm{Main}(x))>0 \Leftrightarrow f(x)>0$.
Since $f\in\#PE$, this feasibility predicate is in $P$.

\paragraph{Correctness.}
We prove $f'(\phi(x))=f(x)$ via two simple claims.

\medskip
\begin{claim} For all $m\in\mathbb{N}$ and all $x\in\Sigma^*$,
\[
f'(\mathrm{Mult}(m,x))=m\cdot f'(\mathrm{Main}(x)).
\]
\end{claim}
\begin{proof}
By induction on $m$.
If $m=0$, then $\mathrm{Mult}(0,x)$ is a leaf with value $0$, so both sides are $0$.
If $m=1$, then $\mathrm{Mult}(1,x)$ has the single child $\mathrm{Main}(x)$, hence
$f'(\mathrm{Mult}(1,x))=f'(\mathrm{Main}(x))$.
For $m\ge 2$ we have
\begin{align*}
f'(\mathrm{Mult}(m,x))
=&f'(\mathrm{Mult}(\lceil m/2\rceil,x))+f'(\mathrm{Mult}(\lfloor m/2\rfloor,x))\\
=&(\lceil m/2\rceil+\lfloor m/2\rfloor)\,f'(\mathrm{Main}(x))\\
=&m\,f'(\mathrm{Main}(x)),
\end{align*}
using the inductive hypothesis and $\lceil m/2\rceil+\lfloor m/2\rfloor=m$.

\end{proof}

\medskip
\begin{claim}
For all $x\in\Sigma^*$, $f'(\mathrm{Main}(x))=f(x)$.
\end{claim}
\begin{proof}
Expanding the definition of $f'$ at $\mathrm{Main}(x)$ gives
\[
f'(\mathrm{Main}(x))
=f'(\mathrm{Bit}(t(x)))+\sum_{i=1}^{r(x)} f'(\mathrm{Mult}(g(x,i),h(x,i))).
\]
The first term is $f'(\mathrm{Bit}(t(x)))=t(x)$, and by Claim~1 the sum equals
$\sum_{i=1}^{r(x)} g(x,i)\, f'(\mathrm{Main}(h(x,i)))$.
By induction over the well-founded recursion tree generated by $h$ (which is finite by the
termination property of the original self-reduction), we may assume
$f'(\mathrm{Main}(h(x,i)))=f(h(x,i))$ for all children. Substituting yields
\[
f'(\mathrm{Main}(x))=t(x)+\sum_{i=1}^{r(x)} g(x,i)\,f(h(x,i))=f(x),
\]
which is exactly the original self-reducible equation.
\end{proof}

Therefore, for every $x$, $f'(\phi(x))=f(x)$, so $f\le_m^P f'$ (parsimoniously).
Since $f'\in \#PE_{uSR}$, we obtain $\#PE_{SR}\subseteq uTotP$, and hence $TotP\subseteq uTotP$.
Combined with $uTotP\subseteq TotP$, this proves $uTotP=TotP$.
\end{proof}

\section{Improved Algorithm for \#\textsc{Perfect-Matching}}\label{sec:ImprovedPerfect}

This appendix refines the bounded-uSR form for subcubic perfect matchings from
\Cref{sec:pm-subcubic}.  We keep the same reduction operator
$\mathrm{Reduce}_{\mathrm{PM}}$ from \Cref{def:pm-reduce} (parity filter and forced
degree-$1$ rule), and we keep working on the reduced core as in
\Cref{def:coregraphs,def:pm-2path}.  The only change is the branching step: instead
of a $2$-way include/exclude branch on a single edge, we branch directly on
\emph{which neighbour a pivot vertex is matched to}.  This yields a $3$-way
branching rule when a degree-$3$ vertex exists, and it \emph{also unfolds the
degree-$2$ situation explicitly} (when no degree-$3$ vertex exists, the reduced
instance is a disjoint union of even cycles, and we branch on a degree-$2$
vertex on a cycle).

\subsection{Branching on the Matched Neighbour}

For a graph $H$, write $V_2(H)$ and $V_3(H)$ for its degree-$2$ and degree-$3$
vertices, and let $n_2(H)\coloneqq |V_2(H)|$ and $n_3(H)\coloneqq |V_3(H)|$.

\begin{definition}[Improved uSR for subcubic perfect matchings]\label{def:pm-usr-branch3}
Let $S$ be the set of all subcubic graphs together with the failure symbol
$\bot$, and let $f(G)\coloneqq \#\mathrm{PM}(G)$ as in \Cref{sec:pm-subcubic}.
Define a uSR form
\[
\mathcal{F}^{\triangle}_{\mathrm{PM}}=(S,f,r,h,t,q)
\]
as follows.  For input $G\in S$, first compute $H\coloneqq \mathrm{Reduce}_{\mathrm{PM}}(G)$.
\begin{itemize}
  \item If $H=\bot$, then $r(G)=0$ and $t(G)=0$.
  \item If $H=\varnothing$ (empty graph), then $r(G)=0$ and $t(G)=1$.
  \item Otherwise, $H$ is reduced and nonempty.  If $V_3(H)\neq \emptyset$, pick any
        $v\in V_3(H)$ (e.g., the smallest under a fixed ordering) with neighbours
        $N_H(v)=\{u_1,u_2,u_3\}$.  Set $r(G)=3$ and define children
        \[
          h(G)\;=\;\bigl(G_1,G_2,G_3\bigr),\qquad
          G_i \coloneqq \mathrm{Reduce}_{\mathrm{PM}}\bigl(H-\{v,u_i\}\bigr)\quad(i\in\{1,2,3\}).
        \]
  \item Otherwise $V_3(H)=\emptyset$, so all vertices of $H$ have degree $2$.  Pick any
        vertex $v\in V_2(H)$ with neighbours $N_H(v)=\{u_1,u_2\}$.  Set $r(G)=2$ and define
        \[
          h(G)\;=\;\bigl(G_1,G_2\bigr),\qquad
          G_i \coloneqq \mathrm{Reduce}_{\mathrm{PM}}\bigl(H-\{v,u_i\}\bigr)\quad(i\in\{1,2\}).
        \]
  \item In all non-leaf cases, set $t(G)=0$.
\end{itemize}
Finally, let $q(G)\coloneqq |V(G)|$ (or any polynomially bounded depth parameter as in
\Cref{sec:pm-subcubic}).
\end{definition}

\begin{lemma}[Correctness of the branching rule]\label{lem:pm-branch3-correct}
For every subcubic graph $G$, the recursion in \Cref{def:pm-usr-branch3} satisfies
\[
f(G)=t(G)+\sum_{i=1}^{r(G)} f\bigl(h(G)_i\bigr).
\]
\end{lemma}
\begin{proof}
Let $H=\mathrm{Reduce}_{\mathrm{PM}}(G)$.  If \(H=\bot\), then \(t(G)=0\)
and \(r(G)=0\); if \(H=\varnothing\), then \(t(G)=1\) and \(r(G)=0\).
These are exactly the two leaf cases.

Otherwise, $H$ is reduced and nonempty.  If $v$ is the chosen pivot vertex, then every
perfect matching $M$ of $H$ matches $v$ to \emph{exactly one} neighbour $u_i\in N_H(v)$.
This partitions the perfect matchings of $H$ into $r(G)$ disjoint classes indexed by $i$,
and for each $i$ the remaining edges of $M$ form a perfect matching of
$\mathrm{Reduce}_{\mathrm{PM}}(H-\{v,u_i\})$ (the reduction only applies forced degree-$1$
choices and parity checks, and does not change the perfect-matchability of the residual
instance).  Summing over $i$ gives the claim.
\end{proof}

\subsection{A Recursion-Compatible Bound}

We now give a recursion-compatible bound for \Cref{def:pm-usr-branch3}.  Unlike the main
section, we must explicitly account for the degree-$2$ case (cycle components), hence the
bound depends on both $n_2$ and $n_3$.

\begin{definition}[Improved bound]\label{def:pm-branch3-bound}
For a reduced graph $H\not=\bot$, define the real-valued bound
\[
\widetilde b(H)\;\coloneqq\; 2^{\,n_2(H)/4}\cdot 6^{\,n_3(H)/6},
\]
and the integer-valued bound $b(H)\coloneqq \lfloor \widetilde b(H)\rfloor$.
Set \(\widetilde b(\bot)\coloneqq b(\bot)\coloneqq 0\).
\end{definition}

We prove recursion compatibility for $\widetilde b$ and then inherit it for $b$ by flooring.
The cycle case is handled first.

\begin{lemma}[Degree-$2$ case: cycle branching]\label{lem:pm-branch3-cycle}
Let $H$ be reduced, nonempty, and satisfy $V_3(H)=\emptyset$ (so all vertices have degree $2$).
In each of the two branches of \Cref{def:pm-usr-branch3}, the reduction deletes at least
$4$ degree-$2$ vertices.  Consequently, if $H_1,H_2$ are the two children then
\[
\widetilde b(H_1)+\widetilde b(H_2)\;\le\;\widetilde b(H).
\]
\end{lemma}
\begin{proof}
Since $H$ is reduced and all degrees are $2$, every connected component of $H$ is a cycle.
Because $H$ has even order (otherwise $\mathrm{Reduce}_{\mathrm{PM}}$ would return $\bot$),
each component cycle has even length, hence length at least $4$.

Let $v\in V_2(H)$ be the chosen vertex with neighbours $u_1,u_2$.  In branch $i$, we delete
$\{v,u_i\}$.  This breaks the cycle containing $v$ into a path with degree-$1$ endpoints, so
$\mathrm{Reduce}_{\mathrm{PM}}$ forces the unique alternating completion along that path until the
entire cycle component is deleted.  Thus, at least the whole cycle of length $\ge 4$ disappears,
so $n_2$ drops by at least $4$ in each branch.  Since $n_3$ is $0$ throughout, we get
$\widetilde b(H_i)\le 2^{(n_2(H)-4)/4} = \frac12\,\widetilde b(H)$ for $i=1,2$, and the sum bound
follows.
\end{proof}

We now analyse the degree-$3$ branching.  The analysis uses maximal $2$-paths and the
alternation lemma from the main section.

\begin{lemma}[Degree-$3$ case: local progress]\label{lem:pm-branch3-degree3-local}
Let $H$ be reduced and contain a degree-$3$ vertex $v$.  Let $P_1,P_2,P_3$ be the three maximal
$2$-paths incident to $v$ (as in \Cref{def:pm-2path}), and assume that their
far endpoints are distinct.  Let
\[
E\;\coloneqq\; \bigl|\{j\in\{1,2,3\} : P_j \text{ has even length}\}\bigr|.
\]
Consider the three branches that match $v$ to $u_i\in N_H(v)$.  In branch $i$, let $k_i$ be the
number of endpoints (among the other ends of $P_1,P_2,P_3$) that are \emph{removed} by the forced
alternation on the corresponding $2$-path(s).  In the worst case (i.e., maximising $\widetilde b$
of the child), we may assume:
\begin{enumerate}
  \item each removed endpoint has two \emph{distinct} degree-$3$ neighbours outside its incident
        $2$-path(s), and
  \item each even $2$-path contributes exactly one internal degree-$2$ vertex.
\end{enumerate}
Under these worst-case assumptions, in branch $i$ we have
\begin{align*}
n_3(H)-n_3(H_i) &\;\ge\; 4+2k_i,\\
n_2(H)-n_2(H_i) &\;\ge\; E-(3+k_i).
\end{align*}
\end{lemma}
\begin{proof}
Fix a branch $i$.  The vertex $v$ is deleted, so $n_3$ drops by $1$ immediately.  For each maximal
$2$-path $P_j$ incident to $v$, the alternation lemma from \Cref{sec:pm-subcubic} implies that,
once we fix whether the first edge at $v$ is used (used for $j=i$, not used for $j\neq i$), the
matching status alternates along $P_j$.  Hence the far endpoint of $P_j$ is either removed (matched
inside $P_j$) or loses the last edge of $P_j$ and therefore drops from degree $3$ to degree $2$.
In either case that endpoint leaves $V_3$, contributing at least $3$ more to the decrease of $n_3$.
Together with removing $v$, this yields a base decrease of at least $4$ in $n_3$.

Now consider those far endpoints that are actually removed.  Each such removed degree-$3$ endpoint
has two additional incident edges not on $P_j$; deleting the endpoint deletes those edges.  In the
worst case for progress, each deleted edge hits a \emph{distinct} degree-$3$ vertex (if an edge hit
a degree-$2$ vertex or two deleted edges hit the same vertex, $\mathrm{Reduce}_{\mathrm{PM}}$ would
create degree-$1$ vertices and force additional deletions, which can only decrease $\widetilde b$).
Thus each removed endpoint causes at least two further degree-$3$ vertices to drop to degree $2$,
contributing an additional $2k_i$ to the decrease of $n_3$.  This proves
$n_3(H)-n_3(H_i)\ge 4+2k_i$.

For $n_2$, every even maximal $2$-path has at least one internal degree-$2$ vertex.  Under the
worst-case assumption (2), we only count one such internal vertex per even path, hence at least $E$
degree-$2$ vertices are deleted.  On the other hand, degree-$2$ vertices can be created only by
degree-$3$ vertices dropping to degree $2$: at most $3-k_i$ far endpoints are downgraded rather than
removed, and the $k_i$ removed endpoints downgrade at most $2k_i$ external neighbours as above.  Thus
at most $(3-k_i)+2k_i = 3+k_i$ new degree-$2$ vertices appear.  Therefore
$n_2(H)-n_2(H_i)\ge E-(3+k_i)$.
\end{proof}

\begin{lemma}[Recursion compatibility for $\widetilde b$]\label{lem:pm-branch3-reccomp-real}
Let $H$ be reduced and nonempty.  If $H_1,\dots,H_{r}$ are its children under
\Cref{def:pm-usr-branch3}, then
\[
\sum_{i=1}^{r}\widetilde b(H_i)\;\le\;\widetilde b(H).
\]
\end{lemma}
\begin{proof}
If $V_3(H)=\emptyset$, then $r=2$ and the claim is \Cref{lem:pm-branch3-cycle}.

Assume $V_3(H)\neq\emptyset$, so $r=3$ and we branch on a degree-$3$ vertex $v$ with incident maximal
$2$-paths $P_1,P_2,P_3$.  Let $E$ be the number of even $P_j$.

We first dispose of the case in which two of these paths have the same far
degree-$3$ endpoint.  Group the paths according to their far endpoint.  If two
or more paths in one group would match the common endpoint inside the paths,
the branch is infeasible at that endpoint and contributes no \(\widetilde b\)
mass.  If no path in a group of size two matches the common endpoint inside the
paths, then the endpoint loses two incident path edges and becomes degree \(1\),
so the degree-$1$ reduction deletes it and its remaining neighbour.  A group of
size three with no such matched path is infeasible.  Finally, if exactly one
path in the group matches the common endpoint, the endpoint is deleted once.
Because the original graph is simple, among parallel odd-length paths with the
same endpoints at most one can be a single edge; every additional odd path has
at least two internal degree-$2$ vertices, all of which are deleted by forced
alternation.  The worst coincident-endpoint cases are therefore the two rows
below; all other parity patterns only delete more degree-$2$ vertices or create
less bound mass:
\[
\begin{array}{c|c|c}
\text{far-endpoint pattern} &
\text{worst parity pattern} &
\text{upper bound on }\sum_i \widetilde b(H_i)/\widetilde b(H)\\
\hline
2+1 & \text{all three paths odd} &
2\cdot 6^{-2/3}+\sqrt{2}\,6^{-4/3}<0.736\\
3 & \text{all three paths odd} &
3\cdot \frac{1}{2}6^{-1/3}<0.826 .
\end{array}
\]
Thus coincident far endpoints are never the worst case for recursion
compatibility.  It remains to analyse the case in which the three far endpoints
are distinct.

For branch $i$, let $\Delta_2^{(i)}\coloneqq n_2(H)-n_2(H_i)$ and
$\Delta_3^{(i)}\coloneqq n_3(H)-n_3(H_i)$.  By \Cref{lem:pm-branch3-degree3-local} we have
\[
\widetilde b(H_i)
=2^{(n_2(H)-\Delta_2^{(i)})/4}\,6^{(n_3(H)-\Delta_3^{(i)})/6}
=\widetilde b(H)\cdot 2^{-\Delta_2^{(i)}/4}\,6^{-\Delta_3^{(i)}/6}.
\]
Write $k_i$ for the number of removed far endpoints among $P_1,P_2,P_3$ in branch $i$.  The alternation
lemma implies the following dependence on parity:
\begin{itemize}
  \item If $P_i$ is even, then its far endpoint is downgraded while each other even $P_j$ (with $j\neq i$)
        has its far endpoint removed.  Hence $k_i = E-1$.
  \item If $P_i$ is odd, then its far endpoint is removed, and every even $P_j$ (with $j\neq i$) also has
        its far endpoint removed.  Hence $k_i = E+1$.
\end{itemize}
Combining this with \Cref{lem:pm-branch3-degree3-local} yields, for an even choice,
\[
\frac{\widetilde b(H_i)}{\widetilde b(H)}
\;\le\;
2^{(3+(E-1)-E)/4}\,6^{-(4+2(E-1))/6}
\;=\;
\sqrt{2}\cdot 6^{-(E+1)/3},
\]
and for an odd choice,
\[
\frac{\widetilde b(H_i)}{\widetilde b(H)}
\;\le\;
2^{(3+(E+1)-E)/4}\,6^{-(4+2(E+1))/6}
\;=\;
2\cdot 6^{-(E+3)/3}.
\]
There are exactly $E$ even choices and $3-E$ odd choices, so
\[
\sum_{i=1}^3 \frac{\widetilde b(H_i)}{\widetilde b(H)}
\;\le\;
E\Bigl(\sqrt{2}\cdot 6^{-(E+1)/3}\Bigr) + (3-E)\Bigl(2\cdot 6^{-(E+3)/3}\Bigr).
\]
A direct check for $E\in\{0,1,2,3\}$ gives:
\[
\begin{array}{c|c}
E & \text{upper bound on }\sum_i \widetilde b(H_i)/\widetilde b(H)\\
\hline
0 & 3\cdot \tfrac13 \;=\; 1\\
1 & \sqrt{2}\cdot 6^{-2/3} + 4\cdot 6^{-4/3} \;\approx\; 0.795\\
2 & 2\sqrt{2}\cdot 6^{-1} + 2\cdot 6^{-5/3} \;\approx\; 0.572\\
3 & 3\sqrt{2}\cdot 6^{-4/3} \;\approx\; 0.389\;,
\end{array}
\]
hence the sum is always at most $1$.  This proves $\sum_i \widetilde b(H_i)\le \widetilde b(H)$.
\end{proof}

\begin{lemma}[Recursion compatibility for $b$]\label{lem:pm-branch3-reccomp-int}
The integer bound $b$ from \Cref{def:pm-branch3-bound} is recursion-compatible for
$\mathcal{F}^{\triangle}_{\mathrm{PM}}$.
\end{lemma}
\begin{proof}
By \Cref{lem:pm-branch3-reccomp-real}, for every internal node we have
$\sum_i \widetilde b(H_i)\le \widetilde b(H)$.  Flooring preserves this inequality:
\[
\sum_i b(H_i)\;=\;\sum_i \lfloor \widetilde b(H_i)\rfloor
\;\le\; \Bigl\lfloor \sum_i \widetilde b(H_i)\Bigr\rfloor
\;\le\; \lfloor \widetilde b(H)\rfloor \;=\; b(H).
\]
(Here we use that $\sum_i \lfloor x_i\rfloor$ is an integer at most $\sum_i x_i$, hence at most
$\lfloor \sum_i x_i\rfloor$.)  Leaf cases are immediate from the definition.
\end{proof}

\subsection{Resulting Running Time}

Combining \Cref{lem:pm-branch3-reccomp-int} with our generic estimator (from the main text) yields an
improved approximate counting time bound, stated in the following theorem.

\CUBICMATCHING*

\subsection{Tightness}\label{sec:matchingTightness}

Since each leaf corresponds to at most a single solution, the branching algorithm and corresponding analysis from the previous sections serves as an upper bound on the number of perfect matchings in subcubic graphs.
We observe that this bound is tight, as the graph consisting of disconnected copies of $K_{3,3}$ will have $6^{\frac{n}{6}}$ perfect matchings, giving us a lower bound to match our upper bound.
What this implies is that we cannot get a faster algorithm through our interface alone.
This motivates the more sophisticated approach used in the examples of \Cref{sec:IS-basic,sec:2-SAT}, where we decompose the problem into easy and hard subproblems, only then to call our interface on the sum of hard problems.

\end{document}